\documentclass[aap]{imsart}

\RequirePackage{amsthm,amsmath,amsfonts,amssymb}
\RequirePackage[numbers]{natbib}
\RequirePackage[colorlinks,citecolor=blue,urlcolor=blue]{hyperref}
\RequirePackage{graphicx}

\startlocaldefs

\newtheorem{theorem}{Theorem}[section]
\newtheorem{lemma}[theorem]{Lemma}
\newtheorem{proposition}[theorem]{Proposition}
\theoremstyle{remark}
\newtheorem{definition}[theorem]{Definition}
\newtheorem{remark}[theorem]{Remark}

\newcommand{\idc}{\mathbf{1}}
\newcommand{\Px}{ \mathbb{P} }
\newcommand{\Qx}{ \mathbb{Q} }
\newcommand{\N}{ \mathbb{N} }
\newcommand{\Ex}{ \mathbb{E} }
\newcommand{\Hxx}{ \mathbb{H} }

\newcommand{\FxM}{\mathbb{F}^{\rm M}}
\newcommand{\FM}{\mathcal{F}^{\rm M}}
\newcommand{\FMi}{\mathcal{F}^{{\rm M}i}}

\def\esssup_#1{\underset{#1}{\mathrm{ess\,sup\, }}}
\def\essinf_#1{\underset{#1}{\mathrm{ess\,inf\, }}}
\def\argmax_#1{\underset{#1}{\mathrm{arg\,max\, }}}
\def\argmin_#1{\underset{#1}{\mathrm{arg\,min\, }}}

\newcommand{\Gx}{\mathbb{G}}
\newcommand{\Fx}{\mathbb{F} }

\newcommand{\Uadt}{{\cal U}_t^{ad}}
\newcommand{\Uad}{{\cal U}_0^{ad}}

\newcommand{\F}{\mathcal{F}}

\newcommand{\R}{\mathbb{R}}

\endlocaldefs

\begin{document}

\begin{frontmatter}
\title{Risk-Sensitive Credit Portfolio Optimization under Partial Information and Contagion Risk}
\runtitle{Risk-Sensitive Portfolio Optimization}

\begin{aug}
\author[A]{\fnms{Lijun} \snm{Bo}\ead[label=e1]{lijunbo@xidian.edu.cn}},
\author[B]{\fnms{Huafu} \snm{Liao}\ead[label=e2]{mathuaf@nus.edu.sg}}
\and
\author[C]{\fnms{Xiang} \snm{Yu}\ead[label=e3]{xiang.yu@polyu.edu.hk}}
\address[A]{School of Mathematics and Statistics, Xidian University, Xi'an 710071, P.R. China,
\printead{e1}}

\address[B]{Department of Mathematics, National University of Singapore, Singapore 119076, Singapore,
\printead{e2}}

\address[C]{Department of Applied Mathematics, The Hong Kong Polytechnic University, Hung Hom, Kowloon, Hong Kong,
\printead{e3}}

\runauthor{Bo, Liao and Yu}
\end{aug}

\begin{abstract}
This paper investigates the finite horizon risk-sensitive portfolio optimization in a regime-switching credit market with physical and information-induced default contagion. It is assumed that the underlying regime-switching process has countable states and is unobservable. The stochastic control problem is formulated under partial observations of asset prices and sequential default events. By establishing a martingale representation theorem based on incomplete and phasing out filtration, we connect the control problem to a quadratic BSDE with jumps, in which the driver term is non-standard and carries the conditional filter as an infinite-dimensional parameter. By proposing some truncation techniques and proving a uniform a priori estimates, we obtain the existence of a solution to the BSDE using the convergence of solutions associated to some truncated BSDEs. The verification theorem can be concluded with the aid of our BSDE results, which in turn yields the uniqueness of the solution to the BSDE.
\end{abstract}



\end{frontmatter}

\section{Introduction}\label{sec:intro}

Optimal portfolio allocation under risk-sensitive criteria has been an important topic in quantitative finance. The problem formulation can integrate the expected growth rate, the penalty term from the asymptotic variance as well as the risk sensitivity parameter into the dynamic decision making. To name but a few recent works on this topic, Bielecki and Pliska~\cite{BieleckiPliska1999} identify that the risk-sensitive portfolio optimization is related to a mean-variance optimization problem; Nagai and Peng~\cite{NagaiPeng2002} study an infinite time risk-sensitive portfolio optimization problem with an unobservable stochastic factor process; El-Karoui and Hamad\`{e}ne~\cite{ElKaroui03} study the risk-sensitive control and the associated game problems on stochastic functional games; Hansen, et al.~\cite{HansenSargent2006} reformulate it as a robust criteria in which perturbations are penalized by a relative entropy; Hansen and Sargent~\cite{HansenSargent07} solve a decision-making problem with hidden states and relate the prior distribution on the states to a risk-sensitive operator; Davis and LIeo~\cite{DavisLleo11,DavisLleo13} utilize the HJB equation approach to study the risk-sensitive portfolio optimization problem in the jump diffusion model with full information and without default contagion; Andruszkiewicz, et al.~\cite{ADL2016} consider the risk-sensitive asset management involving an observable regime switching process over finite states; Birge, et al.~\cite{BirgeBoCapponi2018} examine a risk-sensitive credit asset management problem with an observable stochastic factor; Bo, et al.~\cite{BoLiaoYu2019} recently investigate a risk-sensitive portfolio optimization problem with both default contagion and regime switching over countable states.

This paper aims to study the risk-sensitive portfolio optimization among multiple credit risky assets. Similar to \cite{BoLiaoYu2019}, the default contagion is considered in the sense that the default intensities of surviving names depend on the default events of all other assets as well as regime states. In particular, the regime switching process is described by a continuous time Markov chain with countable states and the default events of risky assets are depicted via some pure jump indicators. The joint impacts on the optimal portfolio by contagion risk and changes of market and credit regimes can be analyzed in an integrated fashion. One reason to consider possibly countable states is that the Markov chain is usually used to approximate the dynamics of stochastic factors. The standard discretization of sample space leads to countable states of Markov chain (see, e.g., \cite{AngTimmermann12}), therefore our theoretical results can support the numerical implementations of some credit portfolio optimization with stochastic factor processes.

As opposed to \cite{BoLiaoYu2019}, we further recast the problem into a more practical setting when the regime-switching process is not observable, in which the filtering procedure becomes necessary. Consequently, the contagion risk comes from two distinct sources: the ``physical" contagion that is from our way to model default intensity as a function depending on all other default indicators and the ``information-induced" contagion that is generated by our estimation of the regime transition probability of the incoming default using observations of past default events. Despite abundant existing work in portfolio optimization under a hidden Markov chain, see among \cite{PhamQuenez2001}, \cite{SassHaussmann04}, \cite{BauRieder2007}, \cite{BrangerKraft2014}, \cite{LimQuenez15}, \cite{BoCapponiSICON2017}, \cite{XiongZhou2007} and many others, this paper appears as the first one considering risk-sensitive control with both default contagion and partial observations based on countable regimes states. Comparing with \cite{BoLiaoYu2019}, the countable regime states results in an infinite-dimensional filter process and we confront a more complicated infinite-dimensional system of coupled nonlinear PDEs due to default contagion and the infinite-dimensional filter process in Proposition \ref{prop:dynamicspM}. We are lack of adequate tools to tackle this infinite-dimensional system by means of standard PDE theories such as operator method or fixed point method (see, e.g., \cite{Cerrai2001} and \cite{DelongK2008}). On the other hand, BSDE approach has become a powerful tool in financial applications with default risk or incomplete information; see Jiao, et al.~\cite{JiaoKhaPham} in the context of utility maximization under contagion risk and complete information, and Papanicolaou~\cite{Papanicolaou2019} on stochastic control under partial observations without default jumps. In the present paper, we choose to employ the BSDE method to tackle the risk-sensitive control problem and it is interesting to see that the associated BSDE in \eqref{eq:regulizeBSDE} has a non-standard driver term that deserves some careful investigations.

The mathematical contribution of this paper is twofold. Firstly, a new martingale representation theorem is established under partial and phasing-out information. Secondly, we extend the study of quadratic BSDE with jumps by considering a random driver induced from our control problem. More detailed explanations are summarized as below:
\begin{itemize}
\item[(i)] Regarding the aspect of partial observations, we are interested in the incomplete information filtration that possesses a phasing out feature due to sequential defaults of multiple assets. That is, the information of the Brownian motion will be terminated after the associated risky asset defaults. This assumption can better match with the real life situation that the investor can no longer perceive any information from the asset once it exits the market. We therefore focus on the filtration $\FxM$ defined in \eqref{eq:marketinforhatF} that is generated by stopped Brownian motions and the default indicator processes, and a new martingale representation theorem under $\FxM$, i.e., Theorem \ref{thm:hatFrep}, is needed. By applying the changing of measure and technical modifications of some arguments in Frey and Schmidt~\cite{FrSch12} together with the approximation scheme and Monotone Class Theorem, we can conclude Theorem \ref{thm:hatFrep}, which is an interesting new result.

\item[(ii)] There are many existing works on quadratic BSDE with jumps. Morlais~\cite{Morlais09} studies the existence of solution to the BSDE with jumps arising from an exponential utility maximization problem with a bounded terminal condition. Morlais~\cite{Morlais} extends the work when the jump measure satisfies the infinite-mass. Kazi-Tani, et al.~\cite{ChaoZhou2015} apply a fixed point method to study the quadratic BSDE with jumps given a small $L^{\infty}$-terminal condition. Antonelli and Mancini~\cite{Fabio2016} further refines the results of the previous work by considering a generator depending on all components and unbounded terminal conditions. All aforementioned work crucially rely on the same quadratic-exponential structure of the driver term, namely quadratic growth in the Brownian component and exponential growth with respect to the jump term, which entails a priori estimates of the solution. On the contrary, the random driver in our quadratic BSDE \eqref{eq:regulizeBSDE} does not satisfy this property, which results from the risk sensitive preference engaging contagion dependence and the filtering process, see Remark \ref{remarkdiff} for detailed explanations. Consequently, the existence of solution can not follow from the same analysis in the literature. This is the main motivation for us to conduct this research, which not only can contribute to the risk sensitive portfolio optimization under default contagion, but will also enrich the study of quadratic BSDE with jumps by allowing some non-standard random drivers.

Note that Ankirchner, et al.~\cite{SCA2010} consider a quadratic BSDE driven by Brownian motion and a compensated default process, and the quadratic-exponential structure is not postulated therein. Nevertheless, the arguments in \cite{SCA2010} also can not be adopted in our setting because \cite{SCA2010} only considers a single default jump and their BSDE can eventually be split into two BSDE problems without jumps, see Remark \ref{remarkdiff} for the detailed comparison. To overcome some new difficulties caused by the random driver, we follow a two-step procedure. In the first step, we propose some tailor-made truncations on the driver term to make it Lipschitz uniformly in time and in sample path such that the existence and uniqueness of the solution can easily follow. The challenging part is to derive a uniform a priori estimates for all truncated solutions, in which the bounded estimate of the jump solution of the truncated quadratic BSDE will become helpful when the random driver does not exhibit the standard structure. In the second step, we adopt and modify some approximation arguments in Kobylanski~\cite{Koby2000} to fit into our setting with jumps and verify that the limiting process from step one solves the original BSDE in an appropriate space. We believe that the analysis of BSDE \eqref{eq:regulizeBSDE} can be further extended to tackle more general random drivers that stem from other default contagion models.
\end{itemize}

The rest of the paper is organized as follows. Section \ref{sec:model} introduces the model of credit risky assets with regime-switching under partial information. Section \ref{sec:filer-MRT} focuses on the filter process and proves a new martingale representation theorem. Section~\ref{sec:risk-senpartialinf} relates the risk-sensitive portfolio optimization problem under partial information to a quadratic BSDE with jumps. Section \ref{sec:BSDE} is devoted to the proof of the existence of solution to the BSDE problem. In Section \ref{sec:optimum}, the verification theorem is concluded by using our BSDE results, which further implies the uniqueness of the solution to the BSDE problem. The technical proofs of some auxiliary results are reported in Appendix \ref{app:proof1}.

\section{The model}\label{sec:model}

We first introduce the market model consisting of credit risky assets with default contagion and regime-switching. Let $(\Omega,\mathcal{F},\mathbb{F},\Px)$ be a complete filtered probability space with the filtration $\Fx=(\F_t)_{t\geq0}$ satisfying the usual conditions. We consider $n$ defaultable risky assets and one riskless bond, whose dynamics are $\Fx$-adapted processes and are defined via three components:
\begin{itemize}
  \item {\it Hidden regime-switching process.} The hidden regime-switching process $I$ is described by a continuous time Markov chain with the generator matrix $Q=(q_{ij})_{1\leq i,j\leq m}$, where $2\leq m\leq+\infty$. The state space of the regime-switching process $I$, denoted by $S_I=\{1,2,\ldots,m\}$, may contain countably many states. It is assumed henceforth that the information of the regime-switching process $I$ is not observable by the investor.
  \item {\it Default indicator process.} Let $H=(H_i(t);~i=1,\ldots,n)_{t\geq0}$ denote the default indicator process with the state space ${S}_H=\{0,1\}^n$. It is assumed that the bivariate process $(I(t),H(t))_{t\geq0}$ is a Markov process with the state space $S_I\times S_H$, and moreover $(I(t))_{t\geq0}$ and $(H(t))_{t\geq0}$ do not jump simultaneously. With a stochastic rate ${\bf1}_{\{H_i(t)=0\}}\lambda_{i}(I(t),H(t))={\bf1}_{\{H_i(t)=0\}} \lambda_i(I(t),(H_1(t),\ldots,H_{i-1}(t),0,H_{i+1}(t),\ldots,H_n(t)))$, the default indicator process $H$ transits from a state
      \[
      H(t):=(H_1(t),\ldots,H_{i-1}(t),H_i(t),H_{i+1}(t),\ldots,H_n(t))
      \]
      in which the risky asset $i$ is alive ($H_i(t)=0$) to the neighbor state
      \[
      {H}^i(t):=(H_1(t),\ldots,H_{i-1}(t),1-H_i(t),H_{i+1}(t),\ldots,H_n(t))
      \]
      in which the asset $i$ has defaulted. The default contagion is allowed to occur among $n$ risky assets in view that the default intensity of the $i$-th asset depends on the default state $H_j(t)$ for all $j\neq i$ in the market on the event $\{H_i(t) = 0\}$. From its construction, simultaneous defaults are precluded because transitions from $H(t)$ can only occur to a state differing from $H(t)$ in exactly one of the entries (see \cite{BoCapponiSIAM2018}). The intensity function $\lambda_i(k,z)$ is assumed to be strictly positive for all $z\in S_H$. The default intensity of the $i$-th risky asset may change either if (i) a risky asset in the portfolio defaults (counterparty risk effect), or (ii) there are transitions in the macro-economic environment (regime switching). The default time of the $i$-th risky asset with the initial time $t\geq0$ is then given by
\begin{align}\label{eq:default-time}
\tau_i^t:=\inf\{s\geq t;\ H_i(s)=1\},\qquad i=1,\ldots,n.
\end{align}
For simplicity, we set $\tau_i:=\tau_i^0$. Our default model belongs to a rich class of interacting Markovian intensity models, introduced by Frey and Runggaldier~\cite{FreyRung2010}. The Dynkin's formula yields that the process of pure jumps
\begin{align}\label{eq:default-mart}
\Upsilon_i(t):=H_i(t)-\int_0^{t\wedge\tau_i}\lambda_i(I(s),H(s))ds,\qquad t\geq0
\end{align}
is a $(\Px,\Fx)$-martingale, $i=1,\ldots,n$. Let us also denote $\Upsilon=(\Upsilon_i(t);\ i=1,\ldots,n)_{t\geq0}^{\top}$.
\item {\it Pre-default price dynamics.}\quad The price process of the riskless bond $B(t)$ is given by $dB(t)= rB(t)dt$ with $B(0)=1$, where $r\geq0$ is the interest rate. Let $W=(W_i(t);\ i=1,\ldots,n)_{t\geq0}^{\top}$ be an $n$-dimensional Brownian motion. The pre-default price dynamics of $n$ risky assets are given by
\begin{align}\label{eq:P}
dP(t) = {\rm diag}(P(t)) \{(\mu(I(t))+\lambda(I(t),H(t))) dt + \sigma dW(t)\},
\end{align}
where $P(t)=(P_i(t);\ i=1,\ldots,n)^{\top}$. 
For each regime state $k\in S_I$, $\mu(k)$ is an $\R^n$-valued column vector, and $\lambda(k,z)=(\lambda_i(k,z);\ i=1,\ldots,n)^{\top}$ stands for the vector of {default intensities}. The volatility $\sigma={\rm diag}((\sigma_i)_{i=1,\ldots,n})$ is an $\R^{n\times n}$-valued constant diagonal matrix. Here we assume $\sigma_i>0$, $i=1,\ldots,n$, and the inverse of $\sigma$ is denoted by $\sigma^{-1}$.
\end{itemize}

Taking the default into consideration, we can write the price process $\tilde{P}_i(t)$ of the $i$-th defaultable asset by $\tilde{P}_i(t)=(1-H_i(t))P_i(t)$. Integration by parts yields that
\begin{align}\label{eq:tildeP}
d\tilde{P}(t) = {\rm diag}(\tilde{P}(t-))\{\mu(I(t))dt +  \sigma dW(t)-d\Upsilon(t)\}.
\end{align}

Recall that the information of the hidden regime-switching process $I$ is not accessible by the investor, who can only observe public prices of risky assets continuously and the default events of assets (i.e., the information generated by $\tilde{P}$ and $H$). It is our first task to formulate the model dynamics under partial information filtration. To this end, for an adapted process $X=(X(t))_{t\geq0}$, let $\F_t^X=\sigma(X(s);\ s\leq t)$ be the natural filtration generated by $X$. We introduce the auxiliary process ${W}^o=({W}_1^o(t),\ldots,{W}_n^o(t))_{t\geq0}^{\top}$ defined by
\begin{align}\label{eq:tildeWio}
{W}_i^o(t) &:=\sigma_i^{-1}\int_0^{t}(\mu_i(I(s))+\lambda_i(I(s),H(s)))ds+W_i(t),\quad t\geq0,
\end{align}
for $i=1,\ldots,n$. Let $W^{o,\tau}=(W_1^{o,\tau}(t),\ldots,W_n^{o,\tau}(t))_{t\geq0}^{\top}$ be the stopped process of $W^{o}$ by the default times $(\tau_1,\ldots\tau_n)$ in the sense that
\begin{align}\label{eq:tildeWi}
W_i^{o,\tau}(t):=W_i^o(t\wedge\tau_i), \quad t\geq0,\quad  \text{for $i=1,\ldots,n$.}
\end{align}
In view of \eqref{eq:P} and \eqref{eq:tildeP}, the available market information filtration ${\Fx}^{\rm M}:=(\FM_t)_{t\geq0}$ satisfies that
\begin{align}\label{eq:marketinforhatF}
\FM_t:=\mathcal{F}^{\tilde{P}}_t\vee\mathcal{F}^{H}_t={\F}_t^{{W^{o,\tau}}}\vee{\F}_t^H,\quad t\geq0,
\end{align}
where $({\F}_t^{{W^{o,\tau}}})_{t\geq0}$ and $({\F}_t^H)_{t\geq0}$ are the filtration generated by $W^{o,\tau}$ and $H$ respectively, i.e., ${\F}_t^{{W^{o,\tau}}}=\bigvee_{i=1}^n \F_t^{W_i^{o,\tau}}$ and ${\F}_t^H=\bigvee_{i=1}^n \F_t^{H_i}$ for $t\geq0$.

From this point onwards, the next assumption is imposed especially when the number of regime states is infinite, i.e., $m=+\infty$.
\begin{itemize}
  \item[({\bf H})] For $(i,k,z)\in\{1,\ldots,n\}\times S_I\times S_H$, there exist positive constants $\varepsilon$ and $C$ independent of $k$ such that $\varepsilon\leq|\lambda_i(k,z)|+|\mu_i(k)|\leq C$.
\end{itemize}
Note that if the number of regime states is finite, the assumption ({\bf H}) holds trivially by taking
\begin{align*}
\varepsilon:=\min_{(i,k,z)}\{\lambda_i(k,z)|+|\mu_i(k)|\}\ \ \text{and}\ \ C:=\max_{(i,k,z)}\{\lambda_i(k,z)|+|\mu_i(k)|\}.
\end{align*}

\section{Filter processes and martingale representation}\label{sec:filer-MRT}

The goal of this section is to establish a martingale representation theorem for the filter process of the hidden regime-switching process $I=(I(t))_{t\geq0}$ given the partial information $\FxM$ defined by \eqref{eq:marketinforhatF}. This result can simplify our risk-sensitive portfolio optimization problem, which will be elaborated in the next section.

For $k\in S_I$, we introduce the filter process of the hidden regime-switching process $I$ by
\begin{align}\label{eq:projectFMa}
p_k^{\rm M}(t)&:=\Px(I(t)=k|\FM_t),\quad t\geq0.
\end{align}
The state space of $p^{\rm M}=(p_k^{\rm M}(t);\ k\in S_I)^{\top}_{t\geq0}$ is denoted by $S_{p^{\rm M}}$. When $m<+\infty$, it is shown in Lemma B.1 in Capponi, et al.~\cite{capponi} that $S_{p^{\rm M}}=\{p\in(0,1)^m;~\sum_{i=1}^mp_i=1\}$. In our BSDE approach, it is not important if the boundary point in the infinite-dimensional state space $S_{p^{\rm M}}$ can be achieved or not.

Let us also introduce the enlarged filtration $\breve{\Fx}:={\Fx}^{{W^o}}\vee{\Fx}^H$. We first apply a well-known martingale representation (see, e.g., Proposition 7.1.3 in Bielecki and Rutkowski~\cite{Bielecki2002}) of the filter process under the filtration $\breve{\Fx}$. Consider $W^{\rm M}=(W_1^{\rm M}(t),\ldots,W_n^{\rm M}(t))_{t\geq0}^{\top}$ defined by
\begin{align}\label{eq:WM-M}
W_i^{\rm M}(t) &:=W_i^{o,\tau}(t)-\sigma_i^{-1}\int_0^{t\wedge\tau_i}({\mu}_i^{\rm M}(p^{\rm M}(s))+\lambda_i^{\rm M}(p^{\rm M}(s),H(s)))ds,~~ i=1,\ldots,n,
\end{align}
in which we define
\begin{align}\label{eq:filtersM}
\mu^{\rm M}(p):= \sum_{k\in S_I}\mu(k)p_k,\quad \lambda^{\rm M}(p,z):=\sum_{k\in S_I}\lambda(k,z)p_k,\quad (p,z)\in S_{p^{\rm M}}\times{S}_H.
\end{align}
Note that $\mu^{\rm M}(p^{\rm M}(t))$ and $\lambda^{\rm M}(p^{\rm M}(t),z)$ are conditional expectations of $\mu(I(t))$ and $\lambda(I(t),z)$ given the filtration $\FM_t$. The assumption ({\bf H}) guarantees that $\mu^{\rm M}(p)$ and $\lambda^{\rm M}(p,z)$ defined in \eqref{eq:filtersM} are finite. Therefore, it is not difficult to verify that, under ({\bf H}), the process $W^{\rm M}=(W_i^{\rm M}(t);\ i=1,\ldots,n)_{t\geq0}^{\top}$ is a continuous $(\Px,\FxM)$-martingale. Also, we can show that, for $i=1,\ldots,n$, the pure jump process defined by
\begin{align}\label{eq:MMi}
\Upsilon_i^{\rm M}(t)&:=H_i(t)-\int_0^t\lambda_i^{\rm M}(p^{\rm M}(s),H(s))ds,\quad t\geq0
\end{align}
is a $(\Px,\FxM)$-martingale.

First, we have the next auxiliary result.
\begin{lemma}\label{lem:stopmeasrbl}
For $t\geq0$ and $i=1,\ldots,n$, let us denote $\breve{\F}_t^i:=\F_t^{W_i^o}\vee\F_t^{H_i}$ and $\FMi_t:=\F_t^{W_i^{o,\tau}}\vee\F_t^{H_i}$. For any bounded $\R$-valued r.v. $\xi\in\breve{\F}^i_t$, we have $\xi{\bf 1}_{\{\tau_i\geq t\}}\in\FMi_t$.
\end{lemma}

\begin{proof}
Denote ${\cal L}$ the family of all bounded $\R$-valued r.v.s in the sense that
\begin{align*}
{\cal L}:=\{\xi\in\breve{B}_t^i;~\xi{\bf1}_{\{\tau_i\geq t\}}\in\FMi_t\},
\end{align*}
where $\breve{B}_t^i$ stands for all bounded $\R$-valued r.v.s that are $\breve{\F}_t^i$-measurable. The class ${\cal L}$ is nonempty as all constants are in ${\cal L}$. Moreover, it holds that
\begin{itemize}
    \item[(i)] Let $\xi_k\in\mathcal{L}$ for $k\geq1$ such that $\lim_{k\to\infty}\xi_k=\xi$, then $\xi{\bf1}_{\{\tau_i\geq t\}}=\lim_{k\to\infty}\xi_k{\bf1}_{\{\tau_i\geq t\}}\in\FMi_t$.
    \item[(ii)] Let $\xi_i\in{\cal L}$ with $i=1,2$. Then, for all $a,b\in\R$,
    $\{a\xi_1+b\xi_2\}{\bf1}_{\{\tau_i\geq t\}}=a\xi_1{\bf1}_{\{\tau_i\geq t\}}+b\xi_2{\bf1}_{\{\tau_i\geq t\}}\in\FMi_t$.
\end{itemize}
We define another class of r.v.s by
\begin{align}\label{eq:multiplicativeclass}
{\cal M}:=\left\{\prod_{\ell=1}^k{\bf1}_{\{[W_i^o(t_{\ell})]^{-1}(A_{\ell})\}};\ 0\leq t_1<\ldots<t_k\leq t,\ A_{\ell}\in{\cal B}(\R),\ \ell=1,\ldots,k\in\N\right\}.
\end{align}
It is not difficult to see that ${\cal M}$ is a multiplicative class, and it holds that $\F_t^{W_i^o}=\sigma({\cal M})$. Furthermore, each $\xi\in{\cal M}$ admits the form that
\begin{align*}
\xi=\prod_{\ell=1}^k{\bf1}_{\{[W_i^o(t_{\ell})]^{-1}(A_{\ell})\}},\ \mbox{where}\ 0\leq t_1<\ldots<t_k\leq t,\ A_{\ell}\in{\cal B}(\R),\ \ell=1,\ldots,k.
\end{align*}
Therefore, we obtain that
\begin{align*}
\xi{\bf1}_{\{\tau_i\geq t\}}=\prod_{\ell=1}^k{\bf1}_{\{[W_i^o(t_{\ell})]^{-1}(A_{\ell})\}}{\bf1}_{\{\tau_i\geq t\}}=\prod_{\ell=1}^k{\bf1}_{\{[W_i^{o,\tau}(t_{\ell})]^{-1}(A_{\ell})\}}{\bf1}_{\{\tau_i\geq t\}}\in \FMi_t.
\end{align*}
This implies that ${\cal M}\subset{\cal L}$. Monotone Class Theorem entails that ${\cal L}$ contains all bounded $\sigma({\cal M})$-measurable r.v.s. On the other hand, we have $\F_t^{H_i}\subset{\cal L}$ by definition. We next consider
\begin{align*}
\breve{{\cal M}}:=\left\{{\bf1}_{A}(\omega){\bf1}_{B}(\omega);\ A\in\F_t^{W_i^o},\ B\in\F_t^{H_i}\right\}.
\end{align*}
It holds that $\breve{{\cal M}}$ is a multiplicative class and $\breve{\F}_t^i=\sigma(\breve{{\cal M}})$. Moreover, for any $\eta\in\breve{{\cal M}}$, $\eta$ admits the form that $\eta={\bf1}_{A}{\bf1}_{B}$, where $A\in\F_t^{W_i^o}$ and $B\in\F_t^{H_i}$. It has been proved that both ${\bf1}_{A}$ and ${\bf1}_{B}$ are in ${\cal L}$, and hence
\begin{align*}
\eta{\bf1}_{\{\tau_i\geq t\}}={\bf1}_{A}{\bf1}_{B}{\bf1}_{\{\tau_i\geq t\}}=({\bf1}_{A}{\bf1}_{\{\tau_i\geq t\}})({\bf1}_{B}{\bf1}_{\{\tau_i\geq t\}})\in\FMi_t,
\end{align*}
which shows that $\eta\in{\cal L}$. By Monotone Class Theorem again, it holds that ${\cal L}$ contains all bounded $\breve{\F}_t^i$-measurable r.v.s.
\end{proof}

We next present the main result of this section.
\begin{theorem}\label{thm:hatFrep}
Let $T>0$ be a terminal horizon and $L=(L_t)_{t\in[0,T]}$ be a real-valued $(\Px,\FxM)$-square integrable martingale with bounded jumps. There exist $\FxM$-predictable and square integrable $\alpha^{\rm M}=(\alpha_1^{\rm M}(t),\ldots,\alpha_n^{\rm M}(t))_{t\in[0,T]}^{\top}$ and $\beta^{\rm M}=(\beta_1^{\rm M}(t),\ldots,\beta_n^{\rm M}(t))_{t\in[0,T]}^{\top}$ such that, for all $t\in[0,T]$,
\begin{align}\label{eq:L-mart-rep_FM}
L_t=L_0+\sum_{i=1}^n\int_0^t\alpha_i^{\rm M}(s)dW_i^{\rm M}(s)+\sum_{i=1}^n\int_0^t\beta_i^{\rm M}(s)d\Upsilon_i^{\rm M}(s).
\end{align}
Here, the $(\Px,\FxM)$-martingales $W^{\rm M}$ and $\Upsilon^{\rm M}$ are given by \eqref{eq:WM-M} and \eqref{eq:MMi}.
\end{theorem}
Note that the observable information $\FxM$ is generated by $W^{o,\tau}$ and $H$, where $W^{o,\tau}$ is a stopped Brownian motion under $\Px$. Our proof of the theorem can be outlined as two steps: Firstly, we prove a martingale representation w.r.t. $\FxM$ using an auxiliary probability measure $\Px^*$, under which the observed $W^{o,\tau}$ has zero drift and $H$ has the unit default intensity. Secondly, we change the measure and establish the martingale representation under the original probability measure $\Px$.

Fix $t\in[0,T]$ and let $u\in[t,T]$. We introduce
\begin{align}\label{eq:density-Gam}
\Gamma^t(u):=\sum_{i=1}^n\int_t^u(\lambda_i^{-1}(s-)-1)d\Upsilon_i(s)-\sum_{i=1}^n\sigma_i^{-1}\int_t^{u\wedge\tau_i^t}(\mu_i(s)+\lambda_i(s))dW_i(s),
\end{align}
where the simplified notations $\mu_i(t):=\mu_i(I(t))$ and $\lambda_i(t):=\lambda_i(I(t),H(t))$ are used. We then define
\begin{align}\label{eq:Pstar}
    \frac{d\Px^*}{d\Px}\big|_{\F_T}={\cal E}(\Gamma^0)_T,
\end{align}
where ${\cal E}$ denotes the Dol\'eans-Dade exponential and $\Gamma^0=(\Gamma^0(t))_{t\in[0,T]}$. The assumption ({\bf H}) guarantees that $\Px^*\sim\Px$ is a probability measure. Moreover, $W^o$ is an $\Fx$-Brownian motion under $\Px^*$, while the observed process $W^{o,\tau}$ is a stopped $\Fx$-Brownian motion. The $\Fx$-intensity of $H$ is $1$, that is, for $i=1,\ldots,n$, we have that
\begin{align}\label{eq:Mstar}
\Upsilon_i^*(t):=H_i(t)-\int_0^t(1-H_i(s))ds,\quad t\in[0,T]
\end{align}
is an $\Fx$-martingale of pure jumps (It is in fact also an $\FxM$-martingale). The next result serves as the first step to prove Theorem \ref{thm:hatFrep}.
\begin{lemma}\label{lem:martrepPstar}
Let $L=(L_t)_{t\in[0,T]}$ be a real-valued $(\Px^*,\FxM)$-square integrable martingale with bounded jumps. There exist $\FxM$-predictable processes $\alpha^{\rm M}=(\alpha_1^{\rm M}(t),\ldots,\alpha_n^{\rm M}(t))_{t\in[0,T]}^{\top}$ and $\beta^{\rm M}=(\beta_1^{\rm M}(t),\ldots,\beta_n^{\rm M}(t))_{t\in[0,T]}^{\top}$ such that, for all $t\in[0,T]$,
\begin{align}\label{eq:martrepPstar}
L_t=L_0+\sum_{i=1}^n\int_0^t\alpha_i^{\rm M}(s)dW_i^{o,\tau}(s)+\sum_{i=1}^n\int_0^t\beta_i^{\rm M}(s)d\Upsilon_i^*(s).
\end{align}
\end{lemma}

\begin{proof}
Let ${\cal L}$ be the family of all bounded $\FM_T$-measurable r.v.s that can be represented by stochastic integrals w.r.t. $W^{o,\tau}$ and $\Upsilon^*$, i.e., $\xi\in{\cal L}$ if and only if there exist $\FxM$-predictable processes $(\alpha,\beta)$ such that
\begin{align}\label{eq:MrepZ}
\xi=\Ex^*[\xi]+\sum_{i=1}^n\int_0^T\alpha_i(s)dW_i^{o,\tau}(s)+\sum_{i=1}^n\int_0^T\beta_i(s)d\Upsilon_i^*(s).
\end{align}
Here, $\Ex^*$ denotes the expectation under $\Px^*$.

It is easy to see that all constants are in ${\cal L}$ and ${\cal L}$ is a vector space. Moreover, let us consider nonnegative increasing r.v.s $(\xi_k)_{k\geq1}\subset{\cal L}$ such that $\lim_{k\rightarrow\infty}\xi_k=\xi$ a.s. and $\xi$ is bounded. Then, Bounded Convergence Theorem implies that $\xi_k\to\xi$, in $L^2(\Omega)$, as $k\to\infty$. Hence, for each $k\geq1$, there exist $\FxM$-predictable processes $(\alpha^{(k)},\beta^{(k)})$ such that $\xi_k$ admits \eqref{eq:MrepZ}. It follows that, for all distinct $k,l\geq1$,
\begin{align*}
    \xi_k-\xi_l&=\Ex^*[\xi_k-\xi_l]+\sum_{i=1}^n\int_0^T(\alpha_i^{(k)}(s)-\alpha_i^{(l)}(s))dW_i^{o,\tau}\nonumber\\
    &\quad+\sum_{i=1}^n\int_0^T(\beta_i^{(k)}(s)-\beta_i^{(l)}(s))d\Upsilon_i^*(s).
\end{align*}
Therefore, it holds that
\begin{align*}
4\Ex^*[|\xi_k-\xi_l|^2]
&\geq\int_0^T\Ex^*[|\alpha^{(k)}(s)-\alpha^{(l)}(s)|^2+|\beta^{(k)}(s)-\beta^{(l)}(s)|^2]ds.
\end{align*}
This implies that $(\alpha^{(k)},\beta^{(k)})_{k\geq1}$ is a Cauchy sequence in $L^2(\Omega\times[0,T])$, and there exist $\FxM$-predictable processes $(\alpha^*,\beta^*)$ such that $(\alpha^{(k)},\beta^{(k)})\to(\alpha^*,\beta^*)$ in $L^2(\Omega\times[0,T])$, as $k\to\infty$. Let us define
\begin{align*}
\tilde{\xi}:=\Ex^*[\xi]+\sum_{i=1}^n\int_0^T\alpha_i^*(s)dW_i^{o,\tau}(s)+\sum_{i=1}^n\int_0^T\beta_i^*(s) d\Upsilon_i^*(s).
\end{align*}
It follows that $\xi_k\rightarrow\tilde{\xi}$ in $L^2(\Omega)$, as $k\to\infty$. The uniqueness of $L^2$-limit gives that $\xi=\tilde{\xi}$ and hence $\xi\in{\cal L}$.

We next define a multiplicative class of r.v.s by
\begin{align}\label{eq:mulclass222}
{\cal M}:=\left\{\prod_{i=1}^n\xi_i;\  \xi_i\in\FMi_T\ \mbox{is bounded for}~i=1,\ldots,n\right\}.
\end{align}
It is easy to see that $\FM_T=\sigma({\cal M})$. Consider bounded r.v.s $\xi_i\in\FMi_T$, $i=1,\ldots,n$. As $\FMi_T\subset\breve{\F}_T^i$ for $i=1,\ldots,n$, the classical martingale representation under $\breve{\F}_T^i$ (see, e.g., Proposition 7.1.3 of \cite{Bielecki2002}) gives the existence of $\breve{\Fx}^i$-predictable processes $\breve{\alpha}_i=(\breve{\alpha}_i(t))_{t\in[0,T]}$ and $\breve{\beta}_i=(\breve{\beta}_i(t))_{t\in[0,T]}$ such that
\begin{align*}
\xi_i=\Ex^*[\xi_i]+\int_0^T\breve{\alpha}_i(s)dW_i^{o}(s)+\int_0^T\breve{\beta}_i(s)d\Upsilon_i^*(s).
\end{align*}
For $i=1,\ldots,n$, and $t\in[0,T]$, it holds that $W^{o,\tau}_i(t)$, $H_i(t)\in\breve{\F}^i_{t\wedge\tau_i}$, hence $\FMi_T\subset\breve{\F}^i_{T\wedge\tau_i}$. Then
\begin{align*}
\xi_i&=\Ex^*[\xi_i|\breve{\F}^i_{T\wedge\tau_i}]=\Ex^*[\xi_i]+\int_0^{T\wedge\tau_i}\breve{\alpha}_i(s)dW_i^{o}(s)+\int_0^{T\wedge\tau_i}\breve{\beta}_i(s)d\Upsilon_i^*(s)\notag\\
&=\Ex^*[\xi_i]+\int_0^{T\wedge\tau_i}\breve{\alpha}_i(s)dW_i^{o,\tau}(s)+\int_0^{T\wedge\tau_i}\breve{\beta}_i(s)d\Upsilon_i^*(s).
\end{align*}
By virtue of Lemma \ref{lem:stopmeasrbl}, we have that both $\alpha_i(t):=\breve{\alpha}_i(t){\bf1}_{\{\tau_i\geq t\}}$ and $\beta_i(t):=\breve{\beta}_i(t){\bf1}_{\{\tau_i\geq t\}}$ are $\FMi_t$-predictable for $t\in[0,T]$ as ${\bf1}_{\{\tau_i\geq t\}}$ is $\FMi_t$-predictable. Therefore, each $\xi_i\in\FMi_T$ enjoys the representation given by
\begin{align*}
\xi_i=\Ex^*[\xi_i]+\int_0^T\alpha_i(s)dW_i^{o,\tau}(s)+\int_0^T\beta_i(s)d\Upsilon_i^*(s),\quad i=1,\ldots,n.
\end{align*}
For $i=1,\ldots,n$ and $t\in[0,T]$, we define $\FxM$-predictable processes by
\begin{align*}
    \alpha_i^{\rm M
    }(t):=\prod_{k\neq i}\xi_k^{\rm M}(t-)\alpha_i(t),\qquad \beta_i^{\rm M
    }(t):=\prod_{k\neq i}\xi_k^{\rm M}(t-)\beta_i(t),
\end{align*}
where
\begin{align*}
    \xi_i^{\rm M
    }(t) &:=\Ex^*[\xi_i]+\int_0^t\alpha_i(s)dW_i^{o,\tau}(s)+\int_0^t\beta_i(s)d\Upsilon_i^*(s).
\end{align*}
It\^o's formula gives that
\begin{align}\label{eq:multimartrep000}
\prod_{i=1}^n\xi_i=\Ex^*\left[\prod_{i=1}^n\xi_i\right]+\sum_{i=1}^n\int_0^T\alpha_i^{\rm M}(s) dW_i^{o,\tau}(s)+\sum_{i=1}^n\int_0^T\beta_i^{\rm M}(s)d\Upsilon_i^*(s).
\end{align}
The representation \eqref{eq:multimartrep000} then implies that ${\cal M}\subset{\cal L}$ and Monotone Class Theorem yields that ${\cal L}$ contains all bounded $\FM_T$-measurable r.v.s. Note that the jumps of $\Upsilon^{*}$ are bounded. We can hence apply the localization techniques to $L$ and obtain the desired martingale representation under $\Px^*$ as stated in \eqref{eq:martrepPstar}.
\end{proof}

We then continue to complete the proof of Theorem~\ref{thm:hatFrep}.

\noindent\textsc{Proof of Theorem~\ref{thm:hatFrep}.}\quad For fixed $t\in[0,T]$ and any $u\in[t,T]$, we define
\begin{align}\label{eq:GamM}
    \Gamma^{{\rm M},t}(u)&:=\sum_{i=1}^n\int_t^u(\lambda^{\rm M}_i(s-)^{-1}-1)d\Upsilon_i^{\rm M}(s)-\sum_{i=1}^n\sigma_i^{-1}\int_t^u(\mu_i^{\rm M}(s)+\lambda_i^{\rm M}(s))dW_i^{\rm M}(s).
\end{align}
In view of the assumption ({\bf H}), the process $\psi(u):={\cal E}(\Gamma^{{\rm M},t})_u$, $u\in[t,T]$, is an $\FxM$-martingale that satisfies the representation
\begin{align*}
d\psi(u)=\psi(u-)\left\{\sum_{i=1}^n(\lambda^{\rm M}_i(u-)^{-1}-1)d\Upsilon_i^{\rm M}(u)-\sum_{i=1}^n\sigma_i^{-1}(\mu_i^{\rm M}(u)+\lambda_i^{\rm M}(u))dW_i^{\rm M}(u)\right\}.
\end{align*}

Consider an arbitrary bounded r.v. $\xi\in\FM_T$. The process $\zeta^{\rm M, *}(t):=\Ex^*[\psi(T)^{-1}\xi|\FM_t]$ for $t\in[0,T]$ is a square integrable $(\Px^*,\FxM)$-martingale by ({\bf H}). By Lemma \ref{lem:martrepPstar}, there exist $\FxM$-predictable processes $\alpha^{\rm M
    }=(\alpha_1^{\rm M
    }(t),\ldots,\alpha_n^{\rm M
    }(t))_{t\in[0,T]}^{\top}$ and $\beta^{\rm M
    }=(\beta_1^{\rm M
    }(t),\ldots,\beta_n^{\rm M
    }(t))_{t\in[0,T]}^{\top}$ such that
\begin{align*}
\zeta^{\rm M, *}(T)=\psi(T)^{-1}\xi=\Ex^*[\psi(T)^{-1}\xi]+\sum_{i=1}^n\int_0^T\alpha_i^{\rm M}(s) dW_i^{o,\tau}(s)+\sum_{i=1}^n\int_0^T\beta_i^{\rm M}(s)d\Upsilon_i^*(s).
\end{align*}
Therefore, we deduce that
\begin{align}\label{eq:martrepadd}
\xi=\psi(T)\Ex^*[\psi(T)^{-1}\xi]+\psi(T)\sum_{i=1}^n\int_0^T\alpha_i^{\rm M}(s) dW_i^{o,\tau}(s)+\psi(T)\sum_{i=1}^n\int_0^T\beta_i^{\rm M}(s)d\Upsilon_i^*(s).
\end{align}
On the other hand, we first have that
\begin{align}\label{eq:etaTEstar}
    \psi(T)\Ex^*[\psi(T)^{-1}\xi]&=\Ex^*[\psi(T)^{-1}\xi]+\Ex^*[\psi(T)^{-1}\xi]\sum_{i=1}^n\int_0^T\psi(s-)(\lambda^{\rm M}_i(s-)^{-1}-1)d\Upsilon_i^{\rm M}(s)\nonumber\\
   &\quad-\Ex^*[\psi(T)^{-1}\xi]\sum_{i=1}^n\int_0^T\psi(s)\sigma_i^{-1}(\mu_i^{\rm M}(s)+\lambda_i^{\rm M}(s))dW_i^{\rm M}(s).
\end{align}
Integration by parts yields that
\begin{align}\label{eq:parts1}
    &\psi(T)\sum_{i=1}^n\int_0^T\alpha_i^{\rm M}(s) dW_i^{o,\tau}(s)=\sum_{i=1}^n\int_0^T\psi(s)\alpha_i^{\rm M}(s) dW_i^{\rm M}(s)\nonumber\\
    &\qquad+\sum_{j=1}^n\int_0^T\psi(s-)\left(\sum_{i=1}^n\int_0^s\alpha_i^{\rm M}(u) dW_i^{o,\tau}(u)\right)(\lambda^{\rm M}_j(s-)^{-1}-1)d\Upsilon_j^{\rm M}(s)\nonumber\\
    &\qquad-\sum_{j=1}^n\int_0^T\psi(s)\left(\sum_{i=1}^n\int_0^s\alpha_i^{\rm M}(u) dW_i^{o,\tau}(u)\right)\sigma_j^{-1}(\mu_j^{\rm M}(s)+\lambda_j^{\rm M}(s))dW_j^{\rm M}(s),
\end{align}
and
\begin{align}\label{eq:parts2}
    &\psi(T)\sum_{i=1}^n\int_0^T\beta_i^{\rm M}(s) d\Upsilon_i^{*}(s)=\sum_{i=1}^n\int_0^T\psi(s-)\beta_i^{\rm M}(s) \lambda^{\rm M}_j(s-)^{-1} d\Upsilon_i^{\rm M}(s)\nonumber\\
    &\qquad+\sum_{j=1}^n\int_0^T\psi(s-)\left(\sum_{i=1}^n\int_0^{s-}\beta_i^{\rm M}(u) d\Upsilon_i^{*}(u)\right)(\lambda^{\rm M}_j(s-)^{-1}-1)d\Upsilon_j^{\rm M}(s)\nonumber\\
    &\qquad-\sum_{j=1}^n\int_0^T\psi(s)\left(\sum_{i=1}^n\int_0^s\beta_i^{\rm M}(u) d\Upsilon_i^{*}(u)\right)\sigma_j^{-1}(\mu_j^{\rm M}(s)+\lambda_j^{\rm M}(s))dW_j^{\rm M}(s).
\end{align}
By \eqref{eq:martrepadd}-\eqref{eq:parts2}, we deduce that any bounded r.v. $\xi\in\FM_T$ admits the representation as a stochastic integral w.r.t $\Px$-martingales $W^{\rm M}$ and $\Upsilon^{\rm M}$. As the jumps of $\Upsilon^{*}$ are bounded, the localization technique can be applied to $L$ and the desired martingale representation under $\Px$ in \eqref{eq:L-mart-rep_FM} follows. \hfill$\Box$

As a by-product of Theorem~\ref{thm:hatFrep}, the dynamics of the filter $p_k^{\rm M}$ can be explicitly characterized. This result is useful by itself and the proof is deferred to Appendix \ref{app:proof1}.
\begin{proposition}\label{prop:dynamicspM}
Let $k\in S_I$ and $t\in[0,T]$. Under the assumption {\rm({\bf H})}, the filter process $p_k^{\rm M}$ defined in \eqref{eq:projectFMa} admits that
\begin{align}\label{eq:dynamicspM000}
&dp_k^{\rm M}(t)=\sum_{j\in S_I}q_{jk}p_j^{\rm M}(t)dt+p_k^{\rm M}(t-)\sum_{i=1}^n\left\{\frac{\lambda_i(k,H(t-))}{\sum_{l\in S_I}\lambda_i\big(l,H(t-)\big)p_l^{\rm M}(t-)}-1\right\}d\Upsilon_i^{\rm M}(t)\\
&+p_k^{\rm M}(t)\sum_{i=1}^n\left\{\sigma_i^{-1}(\mu_i(k)+\lambda_i(k,H(t)))-\sum_{l\in S_I}p_l^{\rm M}(t)\sigma_i^{-1}(\mu_i(l)+\lambda_i(l,H(t)))\right\}dW_i^{\rm M}(t).\nonumber
\end{align}
Here, the $(\Px,\FxM)$-martingales $W^{\rm M}$ and $\Upsilon^{\rm M}$ are given by \eqref{eq:WM-M} and \eqref{eq:MMi}.
\end{proposition}

Note that in the price dynamics \eqref{eq:tildeP}, the volatility matrix $\sigma$ is assumed to be diagonal, i.e.,  all defaultable assets are driven by independent Brownian motions. This assumption can actually be relaxed as shown in the next remark.
\begin{remark}
Consider the price dynamics of the $i$-th defaultable asset given by
\begin{align}\label{eq:tildePisigmaij}
d\tilde{P}_i(t) = \tilde{P}_i(t-)\left\{\mu_i(I(t))dt +  \sum_{j=1}^n\sigma_{ij} dW_j(t)-d\Upsilon_i(t)\right\},\quad i=1,\ldots,n,
\end{align}
where the volatility matrix $\sigma=(\sigma_{ij})\in\R^{n\times n}$ is non-diagonal. We next transform \eqref{eq:tildePisigmaij} into the one with a diagonal volatility matrix, but noises are no longer independent. More precisely, define $\tilde{W}_i(t):=\tilde{\sigma}_i^{-1}\sum_{k=1}^n\sigma_{ik}W_k(t)$ for $t\in[0,T]$, where $\tilde{\sigma}_i:=\sqrt{\sum_{k=1}^n\sigma^2_{ik}}$ for $i=1,\ldots,n$. Then, for $i=1,\ldots,n$, $\tilde{W}_i=(\tilde{W}_i(t))_{t\in[0,T]}$ is a Brownian motion satisfying the correlation $\langle\tilde{W}_i,\tilde{W}_j\rangle_t=\tilde{\sigma}_i^{-1}\tilde{\sigma}_j^{-1}\sum_{k=1}^n\sigma_{ik}\sigma_{jk}t$ for $i\neq j$. 
The price process \eqref{eq:tildePisigmaij} can be written that
\begin{align}\label{eq:tildePtldesigma}
d\tilde{P}(t) = {\rm diag}(\tilde{P}(t-))\{\mu(I(t))dt +  \tilde\sigma d\tilde W(t)-d\Upsilon(t)\},
\end{align}
where $\tilde\sigma:={\rm diag}(\tilde\sigma_1,\ldots,\tilde\sigma_n)$ is still diagonal and $\tilde W=(\tilde{W}_1,\ldots,\tilde{W}_n)^{\top}$ is an $n$-dimensional correlated Brownian motion.
That is, we can still consider the price dynamics \eqref{eq:tildeP} with correlated Brownian motions $(W_1,\ldots,W_n)$. Note that we can still define $W^o$ and $W^{o,\tau}$ as in \eqref{eq:tildeWio} and \eqref{eq:tildeWi} that for $i=1,\ldots,n$,
\begin{align*}
{W}_i^o(t) &:=\sigma_i^{-1}\int_0^{t}(\mu_i(I(s))+\lambda_i(I(s),H(s)))ds+W_i(t),~ W_i^{o,\tau}(t):=W_i^o(t\wedge\tau_i),~~ t\geq0.
\end{align*}
By the approximation argument and Monotone Class Theorem, Lemma \ref{lem:stopmeasrbl} still holds. However, it will be difficult to prove Lemma \ref{lem:martrepPstar} and Theorem \ref{thm:hatFrep} when $(W_1,\ldots,W_n)$ are not independent.
Indeed, recall that the proof of Lemma \ref{lem:stopmeasrbl} is based on the filtration generated by the price process and the default event of every asset $i$ (i.e., the sub-filtration $\FMi_t:=\F_t^{W_i^{o,\tau}}\vee\F_t^{H_i}$ for $t\geq0$). When $(W_1,\ldots,W_n)$ are independent, we first establish the martingale representation result under each sub-filtration $\FMi_T$. That is, any bounded r.v.s $\xi_i\in\FMi_T$, $i=1,\ldots,n$, admits the representation that
\begin{align*}
\xi_i=\Ex^*[\xi_i]+\int_0^T\alpha_i(s)dW_i^{o,\tau}(s)+\int_0^T\beta_i(s)d\Upsilon_i^*(s),
\end{align*}
where $\alpha_i$ and $\beta_i$ are $(\FMi_t)_{t\in[0,T]}$-predictable. Then, integration by parts can be applied to yield a general representation result under the filtration $\FM_T$, as the underlying driving martingales $(W_i^{o,\tau},\Upsilon_i^*)$ are orthogonal for $i=1,\ldots,n$, and hence Lemma  \ref{lem:martrepPstar} can be proved by the approximation scheme and Monotone Class Theorem.

On the other hand, if $(W_1,\ldots,W_n)$ are {\it not} independent, the orthogonality of these martingales does not hold. But we can still make the same conclusion using an alternative argument. For $i=1,\ldots,n$, under each $\FMi$, it first follows from the same techniques used in Lemma \ref{lem:stopmeasrbl}, Theorem \ref{thm:hatFrep}, and Lemma \ref{lem:martrepPstar}  with independent $(W_1,\ldots,W_n)$ that for any real-valued $\Fx^{{\rm M}i}=(\FMi_t)_{t\in[0,T]}$-square integrable $(\Px,\Fx^{{\rm M}i})$-martingale $L=(L_t)_{t\in[0,T]}$ with bounded jumps, there exist $\Fx^{{\rm M}i}$-predictable and square integrable processes $\alpha^{\rm M}_i$ and $\beta^{\rm M}_i$ such that
\begin{align}\label{eq:L-mart-rep_FM}
L_t=L_0+\int_0^t\alpha_i^{\rm M}(s)dW_i^{\rm M}(s)+\int_0^t\beta_i^{\rm M}(s)d\Upsilon_i^{\rm M}(s),\quad t\in[0,T].
\end{align}

We next prove Theorem  \ref{thm:hatFrep} using Jacod-Yor Theorem (see, e.g., Theorem IV.57 in \cite{Protter2005} or Theorem III.4.29 in \cite{JacodShir03}). To this end, let us consider a filtered probability space $(\Omega,\mathcal G,\mathbb G,P)$. Let $\mathcal H^2$ be the space of $(P,\Gx)$-special semimartingales with finite $\mathcal H^2$-norm. The $\mathcal H^2$-norm for a special semimartingale with canonical decompostion $X=N+A$\footnote{$N$ (resp. $A$) is a local $P$-martingale (resp. a predictable process of finite variation under $P$).} is defined by
\begin{align*}
  \|X\|_{\mathcal H^2}:=\left\|[N,N]^{1/2}_{T}\right\|_{L^2}+\left\|\int_0^{T}\left|dA_s\right|\right\|_{L^2}.
\end{align*}
Let ${\cal A}\subset{\cal H}^2$, which contains constant martingales. Denote by ${\cal S}({\cal A})$ the stable subspace of stochastic integrals generated by ${\cal A}$, and ${\cal M}({\cal A})$ the space of probability measures making all elements of ${\cal A}$ square integrable martingales. We consider the space ${\cal A}=\{W^{\rm M}_1,\ldots,W^{\rm M}_n,\Upsilon_1^{\rm M},\ldots,\Upsilon_n^{\rm M}\}$ and $\mathcal G=\FM_T$. It is easy to see that $\Px\in\mathcal M({\cal A})$. By Theorem IV.57 in Protter~\cite{Protter2005}, to show the martingale representation property is equivalent to show that $\mathbb P$ is an extremal point of ${\cal M}({\cal A})$, i.e., for any given probability measures $\Qx,\mathbb{K}\in{\cal M}({\cal A})$ satisfying
\begin{align}\label{QKP}
  \lambda\mathbb Q+(1-\lambda)\mathbb K=\mathbb P\quad\text{for some}~\lambda\in[0,1],
\end{align}
it holds that $\mathbb Q=\mathbb K=\mathbb P$. For $i=1,\ldots,n$, let us consider
\begin{align*}
{\cal G}_i= \FMi_T,\quad  {\cal A}_i=\left\{W^{\rm M}_i,\Upsilon_i^{\rm M}\right\}.
\end{align*}
Let $\mathbb P_i$, $\mathbb Q_i$ and $\mathbb K_i$ be the restriction of $\mathbb P$, $\mathbb Q$ and $\mathbb K$ on $\mathcal G_i$, respectively. Consequently, $\mathbb P_i$, $\mathbb Q_i$ and $\mathbb K_i\in\mathcal M(\mathcal A_i)$ for $i=1,\ldots,n$, and $\mathbb P_i$ is an extremal point of ${\cal M}({\cal A}_i)$. On the other hand, it follows from \eqref{QKP} that
\begin{align*}
  \lambda\mathbb Q_i+(1-\lambda)\mathbb K_i=\mathbb P_i\quad\text{for some}~\lambda\in[0,1].
\end{align*}
As $\mathbb P_i$, $\mathbb Q_i$ and $\mathbb K_i$ are the restriction of $\mathbb P$, $\mathbb Q$ and $\mathbb K$ on $\mathcal G_i$, it holds that $\mathbb Q_i=\mathbb K_i=\mathbb P_i$ for $i=1,\ldots,n$. Recall that ${\F}_T^{\rm M}=\bigvee_{i=1}^n \FMi_T$ and $\mathbb Q=\mathbb K=\mathbb P$ on $\FMi_T$ for $i=1,\ldots,n$, we have that $\mathbb Q=\mathbb K=\mathbb P$ on $\mathcal G$, which verifies Theorem \ref{thm:hatFrep} when $(W_1,\ldots,W_n)$ are {\it not} independent.
\end{remark}

\section{Risk-sensitive control under partial information}\label{sec:risk-senpartialinf}

We start to formulate the risk-sensitive portfolio optimization under the partial information $\FxM$. Let us first introduce the preliminary value function and transform it into an equivalent objective functional using the martingale representation result in Section~\ref{sec:filer-MRT} and changing of measure. This formulation, together with the appropriate set of admissible trading strategies, can link the control problem to a non-standard quadratic BSDE with jumps.

Let ${\pi}=({\pi}_i(t);\ i=1,\ldots,n)_{t\in[0,T]}^{\top}$ be an $\FxM$-predictable process, which represents the vector of proportions of wealth invested in $n$ defaultable assets $\tilde{P}$ under partial observations. The resulting wealth process $X^{{\pi}}=(X^{{\pi}}(t))_{t\in[0,T]}$ evolves as
\begin{align}\label{eq:wealth}
dX^{{\pi}}(t)
=&X^{{\pi}}(t-){\pi}(t)^\top\{(\mu(I(t))-re_n)dt+\sigma dW(t)-d\Upsilon(t)\}+rX^{{\pi}}(t)dt,
\end{align}
where $e_n=(1,1,\ldots,1)^{\top}$ is the $n$-dimensional identity column vector. As the price of the $i$-th asset jumps to zero when it defaults by \eqref{eq:tildeP}, the corresponding fraction of wealth held by the investor in this asset stays at zero after it defaults. It consequently follows that ${\pi}_i(t)=(1-H_i(t-)){\pi}_i(t)$ for $i=1,\ldots,n$.

We next introduce the admissible set of all candidate dynamic investment strategies in our framework.
\begin{definition}\label{def:addmisibleUad}
For $t\in[0,T]$, $\Uadt$ denotes the set of admissible controls  ${\pi}(u)=({\pi}_i(u);\ i=1,\ldots,n)^{\top}$, $u\in[t,T]$, which are $\FxM$-predictable processes such that SDE~\eqref{eq:wealth} admits a unique positive strong solution with $X^{{\pi}}(t)=x\in\R_+$ and $({\cal E}(\Lambda^{\pi,t})_u)_{u\in[t,T]}$ is a true $(\Px^*,\FxM)$-martingale, where $\Px^*$ is given by \eqref{eq:Pstar} and $\Lambda^{\pi,t}$ is defined later by \eqref{eq:Lambdapit}. It also follows that the process $\pi$ should take values in $U:=(-\infty,1)^n$. \end{definition}

\begin{remark}
The constraint on admissible investment strategies with the martingale property is by no means restrictive. It will be shown in Section \ref{sec:optimum} that the first-order condition leads to the optimal solution $\pi^*\in\Uadt$ as $({\cal E}(\Lambda^{\pi^*,t})_u)_{u\in[t,T]}$ can be verified to be a $(\Px^*,\FxM)$-martingale. This additional constraint on admissibility can facilitate our future transformation of the original control problem into a simplified form.
\end{remark}

For ${\pi}\in\Uadt$, the wealth process can be rewritten equivalently by
\begin{align}\label{eq:wealth2}
X^{{\pi}}(T) =&X^{{\pi}}(t)\exp\Bigg\{\int_t^T[r+{\pi}(s)^\top(\mu(I(s))-re_n)]ds+\int_t^T{\pi}(s)^\top\sigma dW(s)\nonumber\\
&-\frac{1}{2}\int_t^T{\pi}(s)^\top\sigma\sigma^\top{\pi}(s) ds+\sum_{i=1}^n\int_t^T\ln(1-{\pi}_i(s))d\Upsilon_i(s)\\
&+\sum_{i=1}^n\int_t^T\lambda_i(I(s),H(s))(1-H_i(s))[{\pi}_i(s)+\ln(1-{\pi}_i(s))]ds\Bigg\}.\notag
\end{align}
Given ${\pi}\in\Uad$ and $(X^{{\pi}}(0),H(0))=(x,z)\in\R_+\times{S}_H$, the risk-sensitive objective functional is defined by
\begin{align}\label{eq:J0}
\tilde{J}({\pi};x,z) := -\frac{2}{\theta}\ln\Ex\left[\exp\left(-\frac{\theta}{2}\ln X^{{\pi}}(T)\right)\right].
\end{align}
The investor seeks to maximize $\tilde{J}$ over all admissible strategies ${\pi}\in{\cal U}_0^{ad}$. We only focus on the case when $\theta\in(0,\infty)$, which corresponds to a risk sensitive attitude. 
For $(X^{{\pi}}(0),H(0))=(x,z)\in\R_+\times{S}_H$, the value function of the control problem is given by
\begin{align}\label{eq:value}
\tilde{V}(x,z) &:=\sup_{{\pi}\in\Uad}\left\{-\frac{2}{\theta}\ln\Ex\left[\exp\left(-\frac{\theta}{2}\ln X^{{\pi}}(T)\right)\right]\right\}\nonumber\\
&=\sup_{{\pi}\in\Uad}\left\{-\frac{2}{\theta}\ln\Ex\left[\left(X_0^{{\pi}}\right)^{-\frac{\theta}{2}}\left(\frac{X^{{\pi}}(T)}{X^{\pi}(0)}\right)^{-\frac{\theta}{2}}\right]\right\}\nonumber\\
&=\ln x-\frac{2}{\theta}\inf_{{\pi}\in\Uad}\left\{\ln\Ex\left[\left(\frac{X^{{\pi}}(T)}{X^{\pi}(0)}\right)^{-\frac{\theta}{2}}\right]\right\}\nonumber\\
&=\ln x-\frac{2}{\theta}\ln\left\{\inf_{{\pi}\in\Uad}\Ex\left[\left(\frac{X^{{\pi}}(T)}{X^{{\pi}}(0)}\right)^{-\frac{\theta}{2}}\right]\right\}.
\end{align}
The control problem is then transformed to $\inf_{{\pi}\in\Uad}\Ex[({X^{{\pi}}(T)}/X^{{\pi}}(0))^{-\frac{\theta}{2}}]$. Hence, for $(t,p,z)\in[0,T]\times S_{p^{\rm M}}\times S_H$, it is equivalent to study the dynamic minimization problem
\begin{align}\label{eq:Jvaluefcn}
V(t,p,z):=\inf_{{\pi}\in\Uadt}J({\pi};t,p,z)
:=\inf_{{\pi}\in\Uadt}\Ex_{t,p,z}\left[\left(\frac{X^{{\pi}}(T)}{X^{{\pi}}(t)}\right)^{-\frac{\theta}{2}}\right],
\end{align}
where $\Ex_{t,p,z}[\cdot]:=\Ex[\cdot|p^{\rm M}(t)=p,H(t)=z]$ and $\frac{X^{{\pi}}(T)}{X^{{\pi}}(t)}$ can be expressed by \eqref{eq:wealth2}.

We next rewrite the objective functional $J$ in \eqref{eq:Jvaluefcn} under $\Px^*$. First, it is easy to see that \eqref{eq:wealth2} is equivalent to
\begin{align}\label{thetaX}
\left(\frac{X^{{\pi}}(T)}{X^{{\pi}}(t)}\right)^{-\frac\theta2}&=\exp\Bigg\{-\frac\theta2\int_t^Tr(1-{\pi}(s)^\top e_n)ds-\frac\theta2\int_t^T{\pi}(s)^\top\sigma dW^{o,\tau}(s)\notag\\
&\quad+\frac{\theta}{4}\int_t^T{\pi}(s)^\top\sigma\sigma^\top{\pi}(s)ds-\frac\theta2\sum_{i=1}^n\int_t^T\ln(1-{\pi}_i(s))dH_i(s)\Bigg\},
\end{align}
where the last equality holds by virtue of ${\pi}_i(t)=(1-H_i(t-)){\pi}_i(t)$. We note that all terms in~\eqref{thetaX} are $\FxM$-adapted. By \eqref{eq:Jvaluefcn}, the objective functional is reformulated to
\begin{align}\label{eq:Junob}
J({\pi};t,q,z)=\Ex_{t,p,z}\left[\left(\frac{X^{{\pi}}(T)}{X^{\pi}(t)}\right)^{-\frac{\theta}{2}}\right]
=\Ex_{t,p,z}^{*}\left[\eta^{-1}(t,T)\left(\frac{X^{{\pi}}(T)}{X^{{\pi}}(t)}\right)^{-\frac{\theta}{2}}\right].
\end{align}
Here, the density process is defined by $\eta(t,u):={\cal E}(\Gamma^t)_u$ with $\Gamma^t$ given in \eqref{eq:density-Gam} and $u\geq t$, and $\Ex^*$ denotes the expectation operator under $\Px^*$ given in \eqref{eq:Pstar}.  Note that $\eta(t,T)$ is not necessarily $\FxM$-adapted due to the presence of $I$ in $\eta(t,T)$. In order to transform the objective functional $J$ in a fully observable form, let us introduce
\begin{align}\label{eq:hateta}
\eta^{\rm M}(t,u):=\Ex[\eta(t,u)|\FM_u],\ \ u\in[t,T].
\end{align}

\begin{lemma}\label{lem:dyanmicsetaM}
Let the assumption {\rm({\bf H})} hold. We have that
\begin{align}\label{eq:hatetaexpo}
\eta^{\rm M}(t,u)&={\cal E}\left(\phi^t\right)_u,\ \ \ u\in[t,T],
\end{align}
where we define
\begin{align*}
\phi^t(\cdot):=&\ \sum_{i=1}^n\int_t^{\cdot}(\lambda_i^{\rm M}(p^{\rm M}(s-),H(s-))^{-1}-1)d\Upsilon_i^{\rm M}(s)\nonumber\\
&\ -\sum_{i=1}^n\int_t^{\cdot}\sigma_i^{-1}(1-H_i(s))(\mu_i^{\rm M}(p^{\rm M}(s))+\lambda_i^{\rm M}(p^{\rm M}(s),H(s))dW_i^{\rm M}(s).
\end{align*}
\end{lemma}

\begin{proof}
It follows by definition that, for $u\in[t,T]$,
\begin{align*}
d\eta(t, u)&=\eta(t, u-)\Bigg\{\sum_{i=1}^n(\lambda_i(I(u-),H(u-))^{-1}-1)d\Upsilon_i(u)\nonumber\\
&\quad-\sum_{i=1}^n\sigma_i^{-1}(1-H_i(u))(\mu_i(I(u))+\lambda_i(I(u),H(u)))dW_i(u)\Bigg\}.
\end{align*}
As in the proof of Proposition~\ref{prop:dynamicspM}, we still choose $W_i^{o,\tau}$ to be the test process for $i=1,\ldots,n$. Noting that $W_i^{o,\tau}$ is a stopped $\Fx$-Brownian motion under $\Px^*$, we obtain that $\eta^{\rm M}=(\eta^{\rm M}(t,u))_{u\in[t,T]}$ and $(\eta W_i^{o,\tau})^{\rm M}=(\Ex[\eta(t,u) W_i^{o,\tau}(u)|\F_u^{\rm M}])_{u\in[t,T]}$ are both square-integrable $\FxM$-martingales under $\Px$. In light of Theorem~\ref{thm:hatFrep}, there exist $\FxM$-predictable processes $\alpha^{\rm M}=(\alpha_1^{\rm M}(t),\ldots,\alpha_n^{\rm M}(t))_{t\in[0,T]}^{\top}$ and $\beta^{\rm M}=(\beta_1^{\rm M}(t),\ldots,\beta_n^{\rm M}(t))_{t\in[0,T]}^{\top}$ such that, for $u\in[t,T]$,
\begin{align}\label{eq:hatetarep}
\eta^{\rm M}(t,u)=1+\sum_{i=1}^n\int_t^u\alpha_i^{\rm M}(s)dW^{\rm M}_i(s)+\sum_{i=1}^n\int_t^u\beta_i^{\rm M}(s)d\Upsilon_i^{\rm M}(s).
\end{align}
On the other hand, integration by parts gives that
\begin{align*}
&\eta^{\rm M}(t, u)W_i^{o,\tau}(u)=W_i^{o,\tau}(t)+\int_t^uW_i^{o,\tau}(s)d\eta^{\rm M}(t, s)+\int_t^u\eta^{\rm M}(t,s)dW_i^{\rm M}(s)\nonumber\\
&\quad+\sigma_i^{-1}\int_t^u\eta^{\rm M}(t,s)(1-H_i(s))(\mu_i^{\rm M}(s)+\lambda_i^{\rm M}(s))ds+\int_t^u(1-H_i(s))\alpha_i^{\rm M}(s)ds.
\end{align*}
Note that the $\FxM$-adapted finite variation part in the canonical decomposition of $(\eta W_i^{o,\tau})^{\rm M}$ vanishes. Using the equality $(\eta W_i^{o,\tau})^{\rm M}=\eta^{\rm M}W_i^{o,\tau}$ and comparing their finite variation parts, we deduce that
\begin{align}\label{eq:alphaMeta}
\alpha_i^{\rm M}(s)=-\sigma_i^{-1}\eta^{\rm M}(t,s)(\mu_i^{\rm M}(s)+\lambda^{\rm M}_i(s)),\quad t\leq s\leq\tau_i^t.
\end{align}

We next choose a test process $\phi_i(t):=H_i(t)-t\wedge\tau_i$ for $t\in[0,T]$ to identify $\beta^{\rm M}$ in \eqref{eq:hatetarep}. By Girsanov's theorem, $\eta\phi_i$ is a $(\Px,\Fx)$-martingale. Then, the $\FxM$-adapted finite variation part of $(\eta\phi_i)^{\rm M}$ vanishes. Moreover, integration by parts yields that
\begin{align*}
&\eta^{\rm M}(t,u)\phi_i(u)
=\phi_i(t)+\int_t^u\phi_i(s-)d\eta^{\rm M}(t,s)+\int_t^u(\eta^{\rm M}(t,s-)+\beta_i^{\rm M}(s-))d\Upsilon_i^{\rm M}(s)
\notag\\
&\quad+\sigma_i^{-1}\int_t^u\eta^{\rm M}(t,s)(1-H_i(s))(\lambda^{\rm M}_i(s)-1)ds+\int_t^u(1-H_i(s))\lambda^{\rm M}_i(s)\beta_i^{\rm M}(s)ds.
\end{align*}
Comparing the finite variation parts of processes $(\eta\phi_i)^{\rm M}=(\Ex[\eta(t,u)\phi_i(u)|\F_u^{\rm M}])_{u\in[t,T]}$ and $\eta^{\rm M}\phi_i=(\eta^{\rm M}(t,u)\phi_i(u))_{u\in[t,T]}$, we have that
\begin{align}\label{eq:betaMeta}
\beta_i^{\rm M}(s)=\eta^{\rm M}(t,s-)(\lambda_i^{\rm M}(s-)^{-1}-1),\quad t\leq s\leq\tau_i^t.
\end{align}
The proof is completed by plugging $\alpha^{\rm M}$ in \eqref{eq:alphaMeta} and $\beta^{\rm M}$ in \eqref{eq:betaMeta} back into \eqref{eq:hatetarep}.
\end{proof}

We next give the reformulation of the objective functional $J$ in \eqref{eq:Junob} under partial information $\FxM$. The proof is deferred to Appendix \ref{app:proof1}.
\begin{lemma}\label{lem:Xtpi2}
Let the assumption {\rm({\bf H})} hold and $\Px^*$ be the probability measure defined in \eqref{eq:Pstar}. Then, for $(\pi;t,p,z)\in\Uadt\times[0,T]\times S_{p^{\rm M}}\times S_H$, it holds that
\begin{align}\label{eq:Qtu0}
J(\pi;t,p,z)=\Ex_{t,p,z}\left[\left(\frac{X^{\pi}(T)}{X^{\pi}(t)}\right)^{-\frac{\theta}{2}}\right]=\Ex^{*}_{t,p,z}\left[e^{Q^{\pi,t}(T)}\right].
\end{align}
Here, the $\FxM$-adapted process $Q^{\pi,t}(u)$ for $u\in[t,T]$ is defined by
\begin{align}\label{eq:Qtu}
Q^{\pi,t}(u)&:=-\frac{r\theta}2(u-t)+\sum_{i=1}^n\int_t^u\left\{\sigma_i^{-1}(\mu_i^{\rm M}(s)+\lambda_i^{\rm M}(s))-\frac{\theta\sigma_i}{2}\pi_i(s)\right\}dW_i^{o,\tau}(s)\notag\\
&\quad-\sum_{i=1}^n\int_t^u\left\{\frac{\theta}{2}\ln(1-\pi_i(s))-\ln(\lambda_i^{\rm M}(s-))\right\}d\Upsilon^*_i(s)\notag\\
&\quad+\sum_{i=1}^n\int_t^{u\wedge\tau_i^t}\left\{1-\lambda_i^{\rm M}(s)+\ln(\lambda_i^{\rm M}(s))-\frac{1}{2}\sigma_i^{-2}(\mu^{\rm M}_i(s)+\lambda^{\rm M}_i(s))^2\right\}ds\notag\\
&\quad+\sum_{i=1}^n\int_t^{u\wedge\tau_i^t}\left\{\frac{r\theta}{2}\pi_i(s)+\frac{\theta\sigma_i^2}{4}\pi_i^2(s)-\frac{\theta}{2}\ln(1-\pi_i(s))\right\}ds,
\end{align}
where $\Upsilon^*=(\Upsilon_1^*(t),\ldots,\Upsilon_n^*(t))_{t\in[0,T]}^{\top}$ is defined by \eqref{eq:Mstar}.
\end{lemma}

We can now introduce a quadratic BSDE with jumps associated to the control problem \eqref{eq:Jvaluefcn}. Let $(t,p,z)\in[0,T]\times S_{p^{\rm M}}\times S_H$, and $(p^{\rm M}(t),H(t))=(p,z)$. Consider the following BSDE defined on the filtered probability space $(\Omega,\F,\FxM,\Px^*)$ with $\Px^*$ given in \eqref{eq:Pstar} that
\begin{align}\label{eq:formalBSDE}
\left\{\begin{aligned}
dY(u)&=f(p^{\rm M}(u),H(u),Z(u),V(u))du+Z(u)^{\top}dW^{o,\tau}(u)+V(u)^{\top}d\Upsilon^*(u),~u\in[t,T);\\[0.4em]
Y(T)&=0,
\end{aligned}\right.
\end{align}
where, for $(p,z,\xi,v)\in S_{p^{\rm M}}\times S_{H}\times\R^n\times\R^n$, the driver term of BSDE is given by
\begin{align}\label{eq:deff}
f(p,z,\xi,v):=\sup_{\pi\in(-\infty,1)^n}h(\pi;p,z,\xi,v),
\end{align}
in which $h(\pi;p,z,\xi,v)$ is given by
\begin{align}\label{eq:defh}
h(\pi;p,z,\xi,v)&:=h_L(p,z,\xi,v)+\sum_{i=1}^nh_i(\pi_i;p,z,\xi_{i},v_{i}).
\end{align}
Here, $h_L(p,z,\xi,v)$ is a linear strategy-independent function in $(\xi,v)$, which is defined by
\begin{align}\label{eq:defh2}
h_L(p,z,\xi,v)&:=-\sum_{i=1}^n(1-z_i)\xi_i\sigma_i^{-1}(\mu^M_i(p)+\lambda^M_i(p,z))+\sum_{i=1}^n(1-z_i)v_i+\frac{r\theta}2,
\end{align}
and for $i=1,\ldots,n$,
\begin{align}\label{eq:defh2}
h_i(\pi_i;p,z,\xi_i,v_i)&:=(1-z_i)\Bigg\{-\frac\theta4\sigma_i^2\pi_i^2+\frac\theta2\left(\mu_i^{\rm M}(p)+\lambda^{\rm M}_i(p,z)-r\right)\pi_i-\frac12\left|\frac\theta2\sigma_i\pi_i-\xi_i\right|^2\nonumber\\
&\quad+\lambda_i^{\rm M}(p,z)-\lambda_i^{\rm M}(p,z)(1-\pi_i)^{-\frac\theta2}e^{v_i}\Bigg\}.
\end{align}
The functions $\mu^{\rm M}(p)$ and $\lambda^{\rm M}(p,z)$ are given in \eqref{eq:filtersM}. From this point onwards, we will write the first component $Y(u)$ of the solution of the BSDE~\eqref{eq:formalBSDE} as $Y(u;t,p,z)$ to emphasize its dependence on the initial data $(p,z)$ at time $t$.

The preliminary relationship between the value function and the solution of BSDE~\eqref{eq:formalBSDE} is built in the first verification result on the optimality as below.
\begin{lemma}\label{lem:BSDE-valuefcn}
Let the assumption {\rm({\bf H})} hold and $(Y,Z,V)$ be a solution of BSDE~\eqref{eq:formalBSDE} given the initial data $(p^{\rm M}(t),H(t))=(p,z)\in S_{p^{\rm M}}\times S_H$ at time $t$. Then, for any $\pi\in\Uadt$, it holds that $J(\pi;t,p,z)\geq e^{Y(t;t,p,z)}$. Moreover, if there exists a process $\pi^*\in\Uadt$ such that $d\Px^*\otimes du$-a.e.
\begin{align}\label{eq:optimal}
h(\pi^*(u);p^{\rm M}(u-),H(u-),Z(u),V(u))=f(p^{\rm M}(u-),H(u-),Z(u),V(u)),
\end{align}
for $u\in[t,T]$, and $\pi^*$ is an optimal strategy for the risk sensitive control problem \eqref{eq:value}.
\end{lemma}

\begin{proof}
By Lemma \ref{lem:Xtpi2}, we have that, for $\pi\in\Uadt$,
\begin{align}
J(\pi;t,p,z)=\Ex_{t,p,z}\left[\left(\frac{X^{\pi}(T)}{X^{\pi}(t)}\right)^{-\frac{\theta}{2}}\right]=\Ex^{*}_{t,p,z}\left[e^{Q^{\pi,t}(T)}\right],
\end{align}
where $Q^{\pi,t}$ is given by \eqref{eq:Qtu}. For $u\in[t,T]$, let us define
\begin{align}\label{eq:Lambdapit}
\Lambda^{\pi,t}(u)&:=\sum_{i=1}^n\int_t^u\left\{\sigma^{-1}_i\big(\mu_i^{\rm M}(s)+\lambda^{\rm M}_i(s)\big)-\frac{\theta\sigma_i}2\pi_i(s)+Z_i(s)\right\}dW_i^{o,\tau}(s)\notag\\
&\quad+\sum_{i=1}^n\int_t^u\left\{(1-\pi(s))^{-\frac\theta2}\lambda^{\rm M}_i(s-)e^{V_i(s)}-1\right\}d\Upsilon_i^*(s).
\end{align}
As $(Y,Z,V)$ solves BSDE~\eqref{eq:formalBSDE}, a direct calculation yields that
\begin{align*}
J(\pi;t,p,z)e^{-Y(t;t,p,z)}&=\Ex^{*}_{t,p,z}\left[e^{Q^{\pi,t}(T)-Y(t;t,p,z)}\right]\notag\\
&=\Ex^*_{t,p,z}\left[{\cal E}(\Lambda^{\pi,t})_T\exp\left(\int_t^T(f(u)-h(\pi(u);u))du\right)\right].
\end{align*}
Here, we have used the simplified notations $f(u):=f(p^{\rm M}(u-),H(u-),Z(u),V(u))$ and $h(\pi(u);u):=h(\pi(u);p^{\rm M}(u-),H(u-),Z(u),V(u))$. By the definition of $f$ in \eqref{eq:deff}, it is easy to see that $f(u)-h(\pi(u);u)\geq0$ for all $u\in[t,T]$. Therefore, for all $s\in[t,T]$,
\begin{align}\label{control Lambda}
e^{Q^{\pi,t}(s)}e^{Y(s;t,p,z)-Y(t;t,p,z)}&={\cal E}(\Lambda^{\pi,t})_s\exp\left(\int_t^s(f(u)-h(\pi(u);u))du\right)
\geq{\cal E}(\Lambda^{\pi,t})_s.
\end{align}
Note that, for all admissible strategies $\pi\in\Uadt$, the process $({\cal E}(\Lambda^{\pi,t})_s)_{s\in[t,T]}$ is a $(\Px^*,\FxM)$-martingale by Definition~\ref{def:addmisibleUad}. This implies that, for any $\pi\in\Uadt$,
\begin{align}\label{eq:JpiY}
J(\pi;t,p,z)e^{-Y(t;t,p,z)}&=\Ex^{*}_{t,p,z}\left[e^{Q^{\pi,t}(T)-Y(t;t,p,z)}\right]\notag\\
&=\Ex^*_{t,p,z}\left[{\cal E}(\Lambda^{\pi,t})_T\exp\left(\int_t^T(f(u)-h(\pi(u);u))du\right)\right]\notag\\
&\geq\Ex^*_{t,p,z}\left[{\cal E}(\Lambda^{\pi,t})_T\right]=1.
\end{align}
On the other hand, if \eqref{eq:optimal} holds, then $f(u)=h(\pi^*(u);u)=0$ for $u\in[t,T]$, a.s.. This further entails that the inequality \eqref{eq:JpiY} holds as an equality. Hence, for all $\pi\in\Uadt$, we get that
\begin{align*}
J(\pi;t,p,z)\geq e^{Y(t;t,p,z)}=J(\pi^*;t,p,z),
\end{align*}
which confirms that $\pi^*\in\Uadt$ is an optimal strategy.
\end{proof}

\section{Quadratic BSDE with jumps}\label{sec:BSDE}

This section focuses on the existence of solutions to BSDE~\eqref{eq:formalBSDE} under the partial information probability space $(\Omega,\F,\FxM,\Px^*)$ with $\Px^*$ given by \eqref{eq:Pstar}. To this end, let us first introduce the next regularized form of BSDE~\eqref{eq:formalBSDE} that
\begin{align}\label{eq:regulizeBSDE}
\left\{\begin{aligned}
d\tilde{Y}(u)&=\tilde{f}(p^{\rm M}(u),H(u),\tilde{Z}(u),\tilde{V}(u))du+\tilde{Z}(u)^{\top}dW^{o,\tau}(u)+\tilde{V}(u)^{\top}d\Upsilon^*(u),~ u\in[t,T);\\[0.4em]
\tilde{Y}(T)&=\int_{t}^Tf(p^{\rm M}(u),H(u),0,0)du.
\end{aligned}\right.
\end{align}
Here, $\tilde{f}(p,z,\xi,v):=f(p,z,\xi,v)-f(p,z,0,0)$ and hence $\tilde{f}(p,z,0,0)=0$ for all $(p,z)\in S_{p^{\rm M}}\times S_H$. Note that the triplet $(Y,Z,V)$ solves~\eqref{eq:formalBSDE} on $[t,T]$ if and only if $(Y-\int_t^\cdot f(p^{\rm M}(u),H(u),0,0)du,Z,V)$ solves \eqref{eq:regulizeBSDE} on $[t,T]$. Therefore, it suffices to prove the existence of $\Fx^{\rm M}$-solutions of BSDE~\eqref{eq:regulizeBSDE} with the random terminal condition.

\begin{remark}\label{remarkdiff}
We stress that $W^{o,\tau}=(W_1^{o}(t\wedge\tau_1),\ldots,W_n^{o}(t\wedge\tau_n))_{t\in[0,T]}^{\top}$ is a martingale under $(\Omega,\F,\FxM,\Px^*)$, therefore the stopped feature by $(\tau_1,\ldots,\tau_n)$ is actually hidden in the proof of the existence of solution $(\tilde{Y}, \tilde{Z}, \tilde{V})$ to BSDE \eqref{eq:regulizeBSDE}. The main challenges to analyze BSDE \eqref{eq:regulizeBSDE} come from its random driver term ${\rm G}(t,\omega,\xi,v):=\tilde{f}(p^{\rm M}(\omega,t),H(\omega,t),\xi,v)$ with $(t,\omega,\xi,v)\in[0,T]\times\Omega\times\R^n\times\R^n$. By the definition of $f(p,z,\xi,v)$ in \eqref{eq:deff}-\eqref{eq:defh2}, it is clear to see that $\tilde{f}(p,z,\xi,v)$ is quadratic in $\xi\in\R^n$ and it is exponentially nonlinear in $v\in\R^n$.
Some standard arguments to obtain a priori estimates in the literature of quadratic BSDEs with jumps, which usually enjoy a quadratic-exponential structure as in Assumption 3.1 of Kazi-Tani, et al.~\cite{ChaoZhou2015} (see also the assumption (H) in \cite{Fabio2016}), can not be applied to BSDE \eqref{eq:regulizeBSDE}.

Note that the quadratic-exponential structure is not enforced in \cite{SCA2010}, which instead consider a class of locally Lipschitz assumption of the driver in their one-dimensional BSDE with respect to the jump solution variable $u\in\R$. However, the assumption (P1) in Ankirchner, et al.~\cite{SCA2010} assumes that the random driver $f(s,\omega,z,u):[0,T]\times\Omega\times\R^d\times\R$ satisfies a special decomposition form in terms of a single default indicator, i.e.,
\begin{align}\label{eq:decom-driver}
  f(s,\omega,z,u)=(l(s,z)+j(s,u))(1-D_{s-}(\omega))+m(s,z)D_{s-}(\omega),
\end{align}
where $D_t:=\idc_{\{\tau_1\leq t\}}$ is the single default indicator and the default time $\tau_1$ is the {\it single jump} in their BSDE. In the decomposition form \eqref{eq:decom-driver}, it can be observed that $m(s,z)$ corresponds to the driver of the post-default case, while $l(s,z)+j(s,u)$ corresponds to the driver of the pre-default case. Moreover, they also assume that $l(\cdot,z),m(\cdot,z)$ and $j(\cdot,u)$ are predictable w.r.t. the filtration generated by a Brownian motion $W$, and there exists a constant $L\in\R_+$ such that, for all $z,z'\in\R^d$,
\begin{align}\label{eq:lm}
|l(s,z)-l(s,z')|+|m(s,z)-m(s,z')|\leq L(1+|z|+|z'|)|z-z'|,
\end{align}
and the jump function $j\geq0$ also satisfies the Lipschitz continuity on $(-K,\infty)$ for any $K>0$. The above assumptions allow them to split the BSDE into two BSDEs driven by the Brownian motion $W$ without jumps. As opposed to a single jump in \cite{SCA2010}, our paper studies a sequential multiple defaults with default contagion and (common) unobservable regime-switching on an infinite sate space (note that a single default does not raise any contagion issue). It is clear that assumptions \eqref{eq:decom-driver} and \eqref{eq:lm} are violated by our random driver ${\rm G}(t,\omega,\xi,v)$.

In summary, some existing analysis can not be applied directly to show the existence of solutions to BSDE~\eqref{eq:regulizeBSDE} with the non-standard random driver ${\rm G}(t,\omega,\xi,v)$. We therefore apply some tailor-made truncation techniques and then show that the solutions of truncated BSDEs will eventually converge to the solution of BSDE~\eqref{eq:regulizeBSDE}.
\end{remark}

\subsection{Formulation of truncated BSDEs}
Let us start to introduce the truncated BSDE under $(\Omega,\F,\FxM,\Px^*)$ as follows: for any $N\geq1$,
\begin{align}\label{tuncatedBSDE}
\left\{\begin{aligned}
d\tilde{Y}^N(u)&=\tilde{f}^N(u,\tilde{Z}^N(u),\tilde{V}^N(u))du+\tilde{Z}^N(u)^{\top}dW^{o,\tau}(u)+\tilde{V}^N(u)^{\top}d\Upsilon^*(u),~ u\in[t,T);\\[0.4em]
\tilde{Y}^N(T)&=\int_t^Tf^N(u,0,0)du.
\end{aligned}\right.
\end{align}
For $(\omega,u,\xi,v)\in\Omega\times[t,T]\times\R^n\times\R^n$, the truncated random driver $\tilde{f}^N$ is defined by
\begin{align}\label{eq:randomdirver}
\tilde{f}^N(\omega,u,\xi,v)&:=f^N(\omega,u,\xi,v)-f^N(\omega,u,0,0),
\end{align}
where
\begin{align}\label{eq:fNhN}
f^N(\omega,u,\xi,v)&:=h_L(p^{\rm M}(\omega,u),H(\omega,u),\xi)\nonumber\\
&\quad+\sum_{i=1}^n(1-H_i(\omega,u))\sup_{\pi_i\in(-\infty,1)}h^N_i(\pi_i;p^{\rm M}(\omega,u),H(\omega,u),\xi,v);\nonumber\\
h^N_i(\pi_i;p,z,\xi_i,v_i)&:=-\frac\theta4\sigma_i^2\pi_i^2+\frac\theta2\left(\mu_i^{\rm M}(p)+\lambda^{\rm M}_i(p,z)-r\right)\pi_i-\frac12\left|\frac\theta2\sigma_i\pi_i-\xi_i\right|^2\rho_N(\xi_i)\nonumber\\
&\quad+\lambda_i^{\rm M}(p,z)-\lambda_i^{\rm M}(p,z)(1-\pi_i)^{-\frac\theta2}\hat{\rho}_N(e^{v_i}).
\end{align}
Here, for $N\geq1$, $\rho_N:\R\to\R_+$ is a chosen truncation function whose first-order derivative is bounded by $1$, such that $\rho_N(x)=1$ if $|x|\leq N$, $\rho_N(x)=0$ if $|x|\geq N+2$, and $0\leq\rho_N(x)\leq1$ if $N\leq|x|\leq N+2$. Meanwhile $\hat{\rho}_N:\R_+\to\R_+$ is chosen as an increasing $C^1$-function whose first-order derivative is bounded by $1$, such that $\hat{\rho}_N(x)=x$, if $0\leq x\leq N$, $\hat{\rho}_N(x)=N+1$, if $x\geq N+2$, and $N\leq\hat{\rho}(x)\leq N+1$, if $N\leq x\leq N+2$.

We will show that for each $N\geq1$, the truncated random driver $\tilde{f}^N(\omega,u,\xi,v)$ is Lipschtiz in $(\xi,v)\in\R^n\times\R^n$ uniformly in $(\omega,u)\in \Omega\times[t,T]$. To this end, we first present the next auxiliary result, whose proof is given in Appendix \ref{app:proof1}.
\begin{lemma}\label{lem:hNRN}
Let the assumption {\rm({\bf H})} hold and $(p,z,\xi_i,v_i)\in S_{p^{\rm M}}\times S_H\times\R\times\R$ for $i=1,\ldots,n$. For each $N\geq1$, there exists a constant $R_N>0$, only depending on $N$, such that
\begin{align}\label{eq:boundedsup}
\sup_{\pi_i\in(-\infty,1)}h^N_i(\pi_i;p,z,\xi_i,v_i)=\sup_{\pi_i\in[-R_N,1)}h^N_i(\pi_i;p,z,\xi_i,v_i).
\end{align}
\end{lemma}
The next result helps to derive a priori estimate for the solution of the truncated BSDE \eqref{tuncatedBSDE}.
\begin{lemma}\label{lem:LipdriverN}
Let the assumption {\rm({\bf H})} hold. For each $N\geq1$, the (random) driver $\tilde{f}^N(\omega,u,\xi,v)$ defined by \eqref{eq:randomdirver} is Lipschtizian continuous in $(\xi,v)\in\R^n\times\R^n$ uniformly on $(\omega,u)\in\Omega\times[t,T]$.
\end{lemma}

\begin{proof} By virtue of \eqref{eq:randomdirver} and \eqref{eq:fNhN} and Lemma~\ref{lem:hNRN}, it suffices to prove that for each $i=1,\ldots,n$, $\bar{h}^N_i(p,z,\xi_i,v_i):=\sup_{\pi_i\in[-R_N,1)}h^N_i(\pi_i;p,z,\xi_i,v_i)$ is  Lipschtizian continuous in $(\xi_i,v_i)\in\R\times\R$ uniformly on $(p,z)\in S_{p^{\rm M}}\times S_H$.  For each $(p,z,\xi_i,v_i)\in S_{p^{\rm M}}\times S_{H}\times\R\times\R$, thanks to the first-order condition, the critical point $\pi^*_i=\pi^*_i(p,z,\xi_i,v_i)$ satisfies that
\begin{align}\label{eq:oneordereqn}
&\lambda_i^{\rm M}(p,z)(1-{\pi^*_i})^{-\frac\theta2-1}\hat{\rho}_N(e^{v_i})\nonumber\\
&\qquad=-\left(1+\frac{\theta}{2}\rho_N(\xi_i)\right)\sigma_i^2{\pi^*_i}+\mu_i^{\rm M}(p)+\lambda^{\rm M}_i(p,z)-r+\sigma_i\xi_i\rho_N(\xi_i).
\end{align}
With the aid of Lemma~\ref{lem:hNRN} and the strict convexity of $\pi_i\to h_i^N(\pi_i;p,z,\xi_i,v_i)$, we get that $\pi^*_i\in[-R_N,1)$. Moreover, in view of \eqref{eq:oneordereqn}, it follows that the positive term
\begin{align}\label{eq:hvbound1}
&(1-{\pi^*_i})^{-\frac\theta2}\hat{\rho}_N(e^{v_i})\nonumber\\
&\quad=\frac{1-{\pi^*_i}}{\lambda_i^{\rm M}(p,z)}\left[-\left(1+\frac{\theta}{2}\rho_N(\xi_i)\right)\sigma_i^2{\pi^*_i}+\mu_i^{\rm M}(p)+\lambda^{\rm M}_i(p,z)-r+\sigma_i\xi_i\rho_N(\xi_i)\right]\leq R_{N,1},
\end{align}
where the constant $R_{N,1}>0$ satisfies that
\begin{align*}
R_{N,1}\geq\frac{1+R_N}{\varepsilon}\max_{i=1,\ldots,n}\left[\left(1+\frac{\theta}{2}\right)\sigma_i^2R_N
+2C+r+\sigma_i(N+2)\right],
\end{align*}
where we recall that the constant $C>0$ is given in the assumption ({\bf H}). The Implicit Function Theorem yields that
\begin{align*}
\frac{\partial}{\partial v_i}\bar{h}^N_i(p,z,\xi_i,v_i)&=\frac{\partial}{\partial v_i}h^N_i(\pi^*_i(p,z,\xi_i,v_i);p,z,\xi_i,v_i)=\frac{\partial}{\partial v_i}h^N_i(\pi_i;p,z,\xi_i,v_i)\Big|_{\pi_i=\pi^*_i(p,z,\xi_i,v_i)}\nonumber\\
&\quad+\frac{\partial\pi^*_i}{\partial v_i}(p,z,\xi_i,v_i)\frac{\partial}{\partial\pi_i}h^N_i(\pi_i;p,z,\xi_i,v_i)\Big|_{\pi_i=\pi^*_i(p,z,\xi_i,v_i)}\nonumber\\
&=\frac{\partial}{\partial v_i}h^N_i(\pi_i;p,z,\xi_i,v_i)\Big|_{\pi_i=\pi^*_i(p,z,\xi_i,v_i)}\nonumber\\
&=-\lambda_i^{\rm M}(p,z)(1-{\pi^*_i})^{-\frac\theta2}e^{v_i}\hat{\rho}_N'(e^{v_i}),
\end{align*}
in which we applied the first-order condition \eqref{eq:oneordereqn} for $\pi_i^*$ in the last equality. Note that the increasing function $\hat{\rho}_N$ enjoys the property that
\begin{align}\label{eq:truncatedhatrhooverthon}
\frac{x\hat{\rho}'_N(x)}{\hat{\rho}_N(x)}=\left\{
  \begin{array}{cl}
    1, & {\rm if}\ x\in(0,N];\\
    \in[0,\frac{N+2}{N}], & {\rm if}\ x\in[N,N+2];\\
    0, & {\rm if}\ x\geq N+2.
  \end{array}
\right.
\end{align}
Taking into account the assumption ({\bf H}) and \eqref{eq:hvbound1}, we arrive at
\begin{align}\label{eq:derivibound}
\left|\frac{\partial}{\partial v_i}\bar{h}^N_i(p,z,\xi_i,v_i)\right|&=\lambda_i^{\rm M}(p,z)(1-{\pi^*_i})^{-\frac\theta2}\hat{\rho}_N(e^{v_i})\frac{e^{v_i}\hat{\rho}_N'(e^{v_i})}{\hat{\rho}_N(e^{v_i})}\leq R_{N,2},
\end{align}
where $R_{N,2}:=C\frac{N+2}{N}R_{N,1}$ is a positive constant that only depends on $N$. On the other hand, we have that
\begin{align*}
\frac{\partial}{\partial \xi_i}\bar{h}^N_i(p,z,\xi_i,v_i)&=\frac{\partial}{\partial \xi_i}h^N_i(\pi_i;p,z,\xi_i,v_i)\Big|_{\pi_i=\pi_i^*(p,z,\xi_i,v_i)}\nonumber\\
&=\left(\frac{\theta}{2}\sigma_i\pi_i^*-\xi_i\right)\rho_N(\xi_i)
-\frac12\left|\frac\theta2\sigma_i\pi_i^*-\xi_i\right|^2\rho_N'(\xi_i).
\end{align*}
It then holds that
\begin{align}\label{eq:derixiibound}
&\left|\frac{\partial}{\partial \xi_i}\bar{h}^N_i(p,z,\xi_i,v_i)\right|\nonumber\\
&~\leq \frac{\theta}{2}\sigma_i(R_N\vee1)+|\xi_i|\rho_N(\xi_i){\bf1}_{|\xi_i|\leq N+2}+\frac{\theta^2}{4}\sigma_i^2(R_N\vee1)^2+|\xi_i|^2|\rho_N'(\xi_i)|{\bf1}_{|\xi_i|\leq N+2}
\leq R_{N,3},
\end{align}
where $R_{N,3}:=\max_{i=1,\ldots,n}[\frac{\theta}{2}\sigma_i(R_N\vee1)+\frac{\theta^2}{4}\sigma_i^2(R_N\vee1)^2+N+2+(N+2)^2]$ is a positive constant that only depends on $N$. Combining \eqref{eq:derivibound} and \eqref{eq:derixiibound}, we obtain the desired result.
\end{proof}

By \eqref{eq:fNhN}, it is easy to see that $f^N(u,0,0)=f(p^{\rm M}(u),H(u),0,0)$ for $u\in[t,T]$. Hence, the terminal condition of the truncated BSDE \eqref{tuncatedBSDE} coincides with the one of the regularized BSDE~\eqref{eq:regulizeBSDE}, i.e.,
\begin{align}\label{def-zeta}
\tilde{Y}^N(T)=\tilde{Y}(T)=:\zeta\text{ for all }N\geq1.
\end{align}
The next auxiliary result further asserts that this random terminal condition is in fact bounded and its proof is presented in Appendix \ref{app:proof1}.
\begin{lemma}\label{lem:boundedzeta}
Let the assumption {\rm({\bf H})} hold. Then, for fixed $t\in[0,T]$, the random terminal value $\zeta=\int_t^Tf(p^{\rm M}(u),H(u),0,0)du$ is bounded.
\end{lemma}

Building upon the martingale representation result in Theorem \ref{thm:hatFrep}, Lemma~\ref{lem:LipdriverN} and Lemma~\ref{lem:boundedzeta}, we next prove that there exists a unique solution of the truncated BSDE~\eqref{tuncatedBSDE} under the assumption {\rm({\bf H})}. In accordance with conventional notations, let us first introduce the following spaces of processes: for fixed $t\in[0,T]$,
\begin{itemize}
\item ${\cal S}_t^p$ for $1\leq p<+\infty$: the space of $\FxM$-adapted r.c.l.l. real-valued processes $Y=(Y(u))_{u\in[t,T]}$ s.t. $\Ex^*[\sup_{u\in[t,T]}|Y(u)|^p]<+\infty$.
\item ${\cal S}_t^{\infty}$: the space of $\FxM$-adapted r.c.l.l. real-valued processes $Y=(Y(u))_{u\in[t,T]}$ s.t. $\|Y\|_{t,\infty}:=\esssup_{(u,\omega)\in[t,T]\times\Omega}|Y(u,\omega)|<\infty$.
\item $L_t^2$: the space of $\FxM$-predictable $\R^n$-valued processes $X=(X(u))_{u\in[t,T]}$ s.t.\\ $\sum_{i=1}^n\Ex^*[\int_t^{T\wedge\tau_i^t}|X_i(u)|^2du]<\infty$.
\item ${\Hxx}_{t,{\rm BMO}}^2$: the space of $\FxM$-predictable $\R^n$-valued processes $Z=(Z(u))_{u\in[t,T]}$ s.t. $\|Z\|^2_{t,{\rm BMO}}:=\sup_{\zeta\in{\cal T}_{[t,T]}}\sum_{i=1}^n\Ex^*[\int_{\zeta}^T(1-H_i(u))|Z_i(u)|^2du|{\F}^{\rm M}_{\zeta}]<\infty$. Here, ${\cal T}_{[t,T]}$ denotes the set of all $\FxM$-stopping times taking values on $[t,T]$.
\end{itemize}

\begin{lemma}\label{lem:existYN}
Let the assumption {\rm({\bf H})} hold. Then, for each $N\geq1$, the truncated BSDE~\eqref{tuncatedBSDE}~admits the unique solution $(\tilde{Y}^N,\tilde{Z}^N,\tilde{V}^N)\in{\cal S}_t^2\times L_t^2\times L_t^2$.
\end{lemma}

\begin{proof}
We can modify some arguments in Carbone, et a.~\cite{CarboneFerrario2008} to fit into our framework. By Lemma~\ref{lem:LipdriverN}, the driver $\tilde{f}^N$ of BSDE \eqref{tuncatedBSDE} is uniformly Lipschitz. Moreover, the predictable quadratic variation process of $K(s):=(W^{o,\tau}(s),\Upsilon^*(s))$ with $s\in[t,T]$ is given by
\begin{align*}
\langle K,K\rangle(s)=\int_0^sk(u)k(u)^\top du,
\end{align*}
where $k(u)=\text{diag}(1-H(u),1-H(u))\in\R^{2n\times 2n}$. Theorem 3.1 in \cite{CarboneFerrario2008} implies that there exist a unique $(\tilde{Y}^N,\tilde{Z}^N,\tilde{V}^N)\in{\cal S}_t^2\times L_t^2\times L_t^2$ and a square integrable $(\Px^*,\FxM)$-martingale $U=(U(u))_{u\in[t,T]}$ satisfying $[U,W^{o,\tau}_i](u)=[U,\Upsilon^*_i](u)=0$ for $u\in[t,T]$, $i=1,\ldots,n$, such that
\begin{align}\label{eq:tildeYZVU}
\tilde{Y}^N(T)-\tilde{Y}^N(s)&=\int_s^T\tilde{f}^N(u,\tilde{Z}^N(u),\tilde{V}^N(u))du+\int_s^T\tilde{Z}^N(u)^{\top}dW^{o,\tau}(u)\notag\\
&\quad+\int_s^T\tilde{V}^N(u)^{\top}d\Upsilon^*(u)+U(T)-U(s),\quad s\in[t,T),
\end{align}
with $\tilde{Y}^N(T)=\int_t^Tf^N(u,0,0)du$. By the martingale representation result in Lemma~\ref{lem:martrepPstar}, there exist $\alpha\in L_t^2$ and $\beta\in L_t^2$ such that, for $s\in[t,T]$,
\begin{align}\label{}
U(s)=U(t)+\sum_{i=1}^n\int_t^s\alpha_i(u)dW_i^{o,\tau}(u)+\sum_{i=1}^n\int_t^s\beta_i(u)d\Upsilon_i^*(u).
\end{align}
A direct calculation yields that, for $s\in[t,T]$,
\begin{align*}
[U,U](s)&=\sum_{i=1}^n\int_t^s\alpha_i(u)d[U,W_i^{o,\tau}](u)+\sum_{i=1}^n\int_t^s\beta_i(u)d[U,\Upsilon^*_i](u)=0.
\end{align*}
This gives that $U(T)-U(s)=0$ for all $s\in[t,T]$, and it follows from \eqref{eq:tildeYZVU} that $(\tilde{Y}^N,\tilde{Z}^N,\tilde{V}^N)\in{\cal S}_t^2\times L_t^2\times L_t^2$ is the unique solution of BSDE~\eqref{tuncatedBSDE}.
\end{proof}

\subsection{A priori estimates and comparison result of truncated solutions}

In this section, we establish a priori estimates and a comparison result of the solution to the truncated BSDE \eqref{tuncatedBSDE} under the assumption ({\bf H}).

We start with a simple estimation depending on $N$.
\begin{lemma}\label{lem:prioriestimate1}
For any $N\geq1$, let $(\tilde{Y}^N,\tilde{Z}^N,\tilde{V}^N)\in{\cal S}_t^2\times L_t^2\times L_t^2$ be the solution of \eqref{tuncatedBSDE}. There exists a constant $R_{T,N}>0$, which depends on $N$ and the bound of $|\zeta|$, such that
\begin{align}\label{eq:Sinfty}
\|\tilde{Y}^N\|_{t,\infty}\leq R_{T,N},\quad \tilde{V}^N(u)\leq R_{T,N}, \quad d\Px^*\otimes du\text{-a.e.}
\end{align}
\end{lemma}

\begin{proof} By applying It\^{o}'s formula to $e^{\beta u}\left|\tilde{Y}^N(u)\right|^2$ with a constant $\beta$ to be determined,  we get that, for any $u\in[t,T]$,
\begin{align}\label{eq:eYN2}
&e^{\beta T}\zeta-e^{\beta u}\left|\tilde{Y}^N(u)\right|^2\notag\\
=&\int_u^T\beta e^{\beta s}\left|\tilde{Y}^N(s)\right|^2ds+2\int_u^Te^{\beta s}\tilde{Y}^N(s)\tilde{f}^N(s,\tilde{Z}^N(s),\tilde{V}^N(s))ds\notag\\
&+2\int_u^Te^{\beta s}\tilde{Y}^N(s)\tilde{Z}^N(s)^\top dW^{o,\tau}(s)-2\sum_{i=1}^n\int_u^{T\wedge\tau_i^u}e^{\beta u}\tilde{Y}^N(s)\tilde{V}^N_i(s)ds\\
&+\sum_{i=1}^n\int_u^Te^{\beta s}\left(|\tilde{Y}^N(s)+\tilde{V}^N_i(s)|^2-|\tilde{Y}^N(s)|^2\right)dH_i(s)+\sum_{i=1}^n\int_u^{T\wedge\tau_i^u}e^{\beta s}\left|\tilde{Z}^N_i(s)\right|^2ds.\notag
\end{align}
Rearranging terms on both sides of \eqref{eq:eYN2}, we can get that
\begin{align}\label{eq:ebeaY3}
&e^{\beta u}\left|\tilde{Y}^N(u)\right|^2+\int_u^T\beta e^{\beta s}\left|\tilde{Y}^N(s)\right|^2ds
+\sum_{i=1}^n\int_u^{T\wedge\tau_i^u}e^{\beta s}\left|\tilde{Z}^N_i(s)\right|^2ds\notag\\
=&e^{\beta T}\zeta-2\int_u^Te^{\beta s}\tilde{Y}^N(s)\tilde{f}^N(s,\tilde{Z}^N(s),\tilde{V}^N(s))ds-2\int_u^Te^{\beta s}\tilde{Y}^N(s)\tilde{Z}^N(s)^\top dW^{o,\tau}(s)\nonumber\\
&-\sum_{i=1}^n\int_u^Te^{\beta s}\left(2\tilde{Y}^N(s)\tilde{V}^N_i(s)+|\tilde{V}_i(s)|^2\right)d\Upsilon^*_i(s)-\sum_{i=1}^n\int_u^{T\wedge\tau_i^u}e^{\beta s}\left|\tilde{V}_i^N(s)\right|^2ds.
\end{align}
Taking into account \eqref{eq:defh} and \eqref{eq:fNhN}, we have that the random driver $\tilde{f}^N(u,\xi,v)$ satisfies that $\tilde{f}^N(u,\xi,v)=\tilde{f}^N(u,(1-H(u))\xi,(1-H(u))v)$. By Lemma~\ref{lem:LipdriverN}, there exists a constant $L_N>0$ depending only on $N$ such that, for all $\epsilon>0$,
\begin{align}\label{eq:ebeaY4}
&\left|2\int_u^Te^{\beta s}\tilde{Y}^N(s)\tilde{f}^N(s,\tilde{Z}^N(s),\tilde{V}^N(s))ds\right|\nonumber\\
&\qquad\leq2L_N\sum_{i=1}^n\int_u^{T\wedge\tau_i^u}e^{\beta s}\left|\tilde{Y}^N(s)\right|\left(|\tilde{Z}^N_i(s)|+|\tilde{V}^N_i(s)|\right)ds\nonumber\\
&\qquad\leq n\epsilon^{-1}L_N\int_u^Te^{\beta s}\left|\tilde{Y}^N(s)\right|^2ds+2\epsilon L_N\sum_{i=1}^n\int_u^{T\wedge\tau_i^u}e^{\beta s}\left(|\tilde{Z}^N_i(s)|^2+|\tilde{V}^N_i(s)|^2\right)ds.
\end{align}
By taking $\epsilon=(4L_N)^{-1}$ and $\beta=n\epsilon^{-1}L_N$, we obtain from \eqref{eq:ebeaY3} and \eqref{eq:ebeaY4} that $e^{\beta u}|\tilde{Y}^N(u)|^2\leq \Ex[e^{\beta T}|\zeta|^2|\mathcal{F}^{\rm M}_u]$, a.s. for $u\in[t,T]$. Thanks to Lemma~\ref{lem:boundedzeta}, it follows that $\|\tilde{Y}^N\|_{t,\infty}\leq e^{\beta T}\|\zeta\|_{0,\infty}$, which proves the first term in \eqref{eq:Sinfty}.

On the other hand, in view of $\Delta \tilde{Y}^N(u)=\tilde{V}^N(u)^\top\Delta \Upsilon^*(u)$, we obtain $|\tilde{V}^N(u)^\top\Delta \Upsilon^*(u)|\leq 2\|\tilde{Y}^N\|_{t,\infty}$. The fact that $\Delta \Upsilon^*_i(u)\in\{0,1\}$ for all $i=1,\ldots,n$ leads to that $\tilde{V}^N(u)^\top\Delta \Upsilon^*(u)=\hat{V}^N(u)^\top\Delta \Upsilon^*(u)$. For $i=1,\ldots,n$, let us define
\begin{align}\label{boundVN}
\hat{V}_i^N(u):=\tilde{V}_i^N(u)\wedge(2\|\tilde{Y}^N\|_{t,\infty})\vee(-2\|\tilde{Y}^N\|_{t,\infty}).
\end{align}
Thus, the stochastic integral $(\tilde{V}^N-\hat{V}^N)\cdot \Upsilon^*$ is a continuous martingale of finite variation, which implies that $(\tilde{V}^N-\hat{V}^N)\cdot \Upsilon^*\equiv0$. Therefore, it follows from $[(\tilde{V}^N-\hat{V}^N)\cdot \Upsilon^*]\equiv0$ that
\begin{align}\label{eq:tildeVN}
(1-H(u))\tilde{V}^N(u)=(1-H(u))\hat{V}^N(u),\ \ d\Px^*\otimes du\text{-a.e.}
\end{align}
Here, for any $\alpha\in\R^n$, $(1-H(u))\alpha:=((1-H_1(u))\alpha_1,\ldots,(1-H_n(u))\alpha_n)^{\top}$. Therefore, $(\tilde{Y}^N,\tilde{Z}^N,\hat{V}^N)$ also solves the BSDE \eqref{tuncatedBSDE} in view of \eqref{eq:tildeVN}. As $\hat{V}^N\in L_t^2$, the uniqueness of solution in Lemma \ref{lem:existYN} entails that $\tilde{V}^N(u)=\hat{V}^N(u)$, $d\Px^*\otimes du$-a.e., which completes the proof of \eqref{eq:Sinfty}.
\end{proof}

The next result improves the estimation by establishing a uniform bound of $(\tilde{Y}^N,\tilde{Z}^N,\tilde{V}^N)_{N\geq1}$, which is independent of $N$. In particular, the BMO property plays an important role in the proof of the verification theorem.
\begin{lemma}\label{lem:prioriestimate2}
For any $N\geq1$, let $(\tilde{Y}^N,\tilde{Z}^N,\tilde{V}^N)\in{\cal S}_t^2\times L_t^2\times L_t^2$ be the solution of \eqref{tuncatedBSDE}. There exists some constant $C_T>0$, which only depends on the bound of $|\zeta|$ defined by \eqref{def-zeta}, such that
\begin{align}\label{eq:estiBMO}
\max\left\{\big\|\tilde{Z}^N\big\|_{t,{\rm BMO}},\ \|\tilde{Y}^N\|_{t,\infty}\right\}\leq C_T,\quad \tilde{V}^N(u)\leq C_T, \quad d\Px^*\otimes du\text{-a.e.}
\end{align}
\end{lemma}

\begin{proof} The key step of the proof is to construct an equivalent probability measure under which $\tilde{Y}^N=(\tilde{Y}^N(t))_{t\in[0,T]}$ is an $\Fx^{\rm M}$-martingale. By Lemma~\ref{lem:boundedzeta}, the boundedness property of $\tilde{Y}^N$ follows by the martingale property of $\tilde{Y}^N=(\tilde{Y}^N(t))_{t\in[0,T]}$ under the new probability measure and the fact that $\tilde{Y}^N(T)=\zeta$ is bounded. It follows from Lemma \ref{lem:prioriestimate1} that, there exists an $\Fx^{\rm M}$-predictable $\R^n$-valued (bounded) process $\hat{V}^N$ defined in \eqref{boundVN} such that $\Px^*\otimes du$-a.e., $(1-H(u-))\tilde{V}^N(u)=(1-H(u-))\hat{V}^N(u)$.

To construct the aforementioned equivalent probability measure, for $i=1,\ldots,n$, let us define
\begin{align*}
\tilde{Z}^{N,i}(u):=(\tilde{Z}^N_1(u),\ldots,\tilde{Z}^N_i(u),0,\ldots,0),\quad \hat{V}^{N,i}(u)&=(\hat{V}^N_1(u),\ldots,\hat{V}^N_i(u),0,\ldots,0).
\end{align*}
We also set $\tilde{Z}^{N,0}(u)=\hat{V}^{N,0}(u)=0$. Consider the following processes that
\begin{align}\label{eq:gammaetai}
\gamma_i(u):=\left\{\begin{array}{cl}
\frac{\tilde{f}^N(u,\tilde{Z}^{N,i}(u),\tilde{V}^N(u))-\tilde{f}^N(u,\tilde{Z}^{N,i-1}(u),\tilde{V}^N(u))}{\tilde{Z}^N_i(u)}, & {\rm if}\ (1-H_i(u-))\tilde{Z}^N_i(u)\neq0; \\[0.4em]
0, & {\rm if}\ (1-H_i(u-))\tilde{Z}^N_i(u)=0,
\end{array}\right.
\end{align}
and
\begin{align}\label{eq:gammaetai2}
\eta_i(u):=\left\{\begin{array}{cl}
\frac{\tilde{f}^N(u,0,\hat{V}^{N,i}(u))-\tilde{f}^N(u,0,\hat{V}^{N,i-1}(u))}{\hat{V}^N_i(u)}, & {\rm if}\ (1-H_i(u-))\hat{V}^N_i(u)\neq0; \\[0.4em]
0, & {\rm if}\ (1-H_i(u-))\hat{V}^N_i(u)=0,
\end{array}\right.
\end{align}
for $i=1,\ldots,n$. Note that $\tilde{f}^N(u,0,0)=0$. Then, for $t\in[0,T]$, we have that, $d\Px^*\otimes du$-a.e.
\begin{align}\label{eq:equiZgaMeta}
\int_t^T\tilde{Z}^N(u)^\top\gamma(u)du+\int_t^T\hat{V}^N(u)^\top\eta(u)du=\int_t^T\tilde{f}^N(u,\tilde{Z}^N(u),\hat{V}^N(u))du.
\end{align}
On the other hand, Lemma~\ref{lem:boundedzeta} yields that the $\R^n$-valued process $\gamma=(\gamma(t))_{t\in[0,T]}$ is bounded. Moreover, Lemma \ref{lem:prioriestimate1} states that the $\Fx^{\rm M}$-predictable $\R^n$-valued process $\hat{V}^N$ is bounded by some constant $C_{T,N}>0$ depending on $T$ and $N$. We next prove that there exists some positive constant $\delta_{T,N}$ depending on $N$ such that
\begin{align}\label{eq:novi}
-1+\delta_{T,N}\leq-\eta_i(u)\leq L_N,\quad \text{a.e.},~~ i=1,\ldots,n,
\end{align}
where $L_N>0$ is the Lipchitiz coefficient of the driver $\tilde{f}^N$ (see Lemma~\ref{lem:LipdriverN}).
In fact, if $H_i(u-)=1$, then $\eta_i(u)=0$. It suffices to assume that $H_i(u-)=0$. For $\tilde{V}^N_i(u)\neq0$, we have from \eqref{eq:truncatedhatrhooverthon} that
\begin{align*}
&\frac{\tilde{f}^N(u,0,\hat{V}^{N,i}(u))-\tilde{f}^N(u,0,\hat{V}^{N,i-1}(u))}{\hat{V}^N_i(u)}\nonumber\\
&\qquad=\int_0^1\frac{\partial}{\partial v_i}\tilde{f}^N(u,0,s\hat{V}^{N,i}(u)+(1-s)\hat{V}^{N,i-1}(u))ds\notag\\
&\qquad=1-\int_0^1(1-{\pi^*_i}(u))^{-\frac\theta2}\hat{\rho}_N(e^{s\hat{V}_i^{N}(u)})\frac{e^{s\hat{V}_i^{N}(u)}
\hat{\rho}'_N(e^{u\hat{V}_i^{N}(u)})}{\hat{\rho}_N(e^{s\hat{V}_i^{N}(u)})}ds\notag\\
&\qquad\leq1-(1+R_N)^{-\frac\theta2}\int_0^{1\wedge R_{T,N}^{-1}\ln{N}}\hat{\rho}_N(e^{s\hat{V}_i^{N}(u)})\frac{e^{s\hat{V}_i^{N}(u)}\hat{\rho}'_N(e^{s\hat{V}_i^{N}(u)})}{\hat{\rho}_N(e^{s\hat{V}_i^{N}(u)})}ds\notag\\
&\qquad=1-(1+R_N)^{-\frac\theta2}\int_0^{1\wedge R_{T,N}^{-1}\ln{N}}\hat{\rho}_N(e^{s\hat{V}_i^{N}(u)})ds\notag\\
&\qquad\leq1-\frac{(1+R_N)^{-\frac\theta2}}{R_{T,N}}\left\{1-e^{-(R_{T,N}\wedge \ln{N})}\right\}=:1-\delta_{T,N}.
\end{align*}
Here, the positive constants $R_N$ and $R_{T,N}$ are given in Lemma \ref{lem:hNRN} and Lemma \ref{lem:prioriestimate1} respectively.

We next define the probability measure $\Qx\sim\Px^*$ by
\begin{align}\label{eq:Q}
\frac{d\Qx}{d\Px^*}\Big|_{\F_s^{\rm M}}={\cal E}\left(-\int_0^\cdot\gamma(u)^\top dW^{o,\tau}(u)-\int_0^\cdot\eta(u)^\top d\Upsilon^*(u)\right)_s.
\end{align}
In view of \eqref{eq:novi} and the boundedness of $\gamma=(\gamma(s))_{s\in[0,T]}$, we have that $\hat{W}^{o,\tau}=(\hat{W}^{o,\tau}(s))_{s\in[0,T]}$ and $\hat{\Upsilon}^*=(\hat{\Upsilon}^*(s))_{s\in[0,T]}$ are both $(\Qx,\Fx^{\rm M})$-martingales, where we define
\begin{align}\label{eq:WhatM}
\hat{W}^{o,\tau}(s):=W^{o,\tau}(s)+\int_0^s\gamma(u)du,\quad \hat{\Upsilon}^*(s):=\Upsilon^*(s)+\int_0^s\eta(u)du,\quad s\in[0,T].
\end{align}
It follows from \eqref{tuncatedBSDE} and \eqref{eq:equiZgaMeta} that, for $u\in[t,T]$,
\begin{align}\label{eq:underQ}
\tilde{Y}^N(u)-\tilde{Y}^N(T)=-\int_u^T\tilde{Z}^N(s)^\top d\hat{W}^{o,\tau}(s)-\int_u^T \tilde{V}^N(s)^\top d\hat{\Upsilon}^*(s),~~\text{$\Qx$-a.e.}
\end{align}
Let $\theta_k^t\geq t$ be a localizing sequence as $\Fx^{\rm M}$ stopping times satisfying $\lim_{k\to\infty}\theta_k^t=T$, a.e. By \eqref{eq:underQ}, it holds that $\tilde{Y}^N(u)=\Ex^{\Qx}[\tilde{Y}^N(T\wedge\tau_k)\big|\F^{\rm M}_u]$ for all $k\geq1$. Lemma~\ref{lem:prioriestimate1} and Bounded Convergence Theorem lead to that $\tilde{Y}^N(u)=\Ex^{\Qx}[\zeta|\F^{\rm M}_u]$ for all $u\in[t,T]$. This, together with Lemma~\ref{lem:boundedzeta}, implies the uniform bound of $\tilde{Y}^N$, i.e., $\|\tilde{Y}^N\|_{t,\infty}\leq\|\zeta\|_{0,\infty}$.

We again construct $\hat{V}^N(u)$ as in \eqref{boundVN}, which gives that $|\hat{V}^N(u)|\leq2\|\tilde{Y}^N\|_{t,\infty}$. We consequently have that $\|\hat{V}^N\|_{t,\infty}\leq2\|\zeta\|_{0,\infty}$ by the argument above. Following the same proof of Lemma \ref{lem:prioriestimate1}, the uniqueness of the solution to BSDE \eqref{tuncatedBSDE} entails the second estimation in \eqref{eq:estiBMO}.

We next apply It\^{o}'s formula to $e^{\beta\tilde{Y}^N(u)}$ on $u\in[t,T]$, where $\beta$ is a constant to be determined, and get that
\begin{align}\label{eq:itoexpYN}
&e^{\beta \zeta}-e^{\beta \tilde{Y}^N(u)}\nonumber\\
&\quad=\sum_{i=1}^n\int_u^T\{e^{\beta(\tilde{Y}^N(s-)+\hat{V}^N_i(s))}-e^{\beta\tilde{Y}^N(s-)}\}dH_i(s)
-\sum_{i=1}^n\int_u^{T\wedge\tau_i^u}\beta e^{\beta\tilde{Y}^N(s-)}\hat{V}^N_i(s)ds\notag\\
&\qquad+\int_u^T\beta e^{\beta\tilde{Y}^N(s)}\tilde{f}^N(s,\tilde{Z}^N(s),\hat{V}^N(s))ds+\int_u^T\beta e^{\beta\tilde{Y}^N(s)}\tilde{Z}^N(s)^\top dW^{o,\tau}(s)\notag\\
&\qquad+\frac{\beta^2}{2}\sum_{i=1}^n\int_u^{T\wedge\tau_i^u}e^{\beta\tilde{Y}^N(s)}\left|\tilde{Z}^N_i(s)\right|^2ds.
\end{align}
Note that $\|(1-H)\hat{V}^N\|_{t,\infty}\leq 2\|\zeta\|_{0,\infty}$. Then, for all $N\geq1$ and $s\in[0,T]$, we claim here that there exist positive constants $R_4$ and $R_5$ independent of $(N,s)$ such that
\begin{align}\label{eq:fquadratic}
\left|\tilde{f}^N\big(s,Z^N(s),\hat{V}^N(s))\right|\leq R_4+R_5\sum_{i=1}^n(1-H_i(s))\left|\tilde{Z}^N_i(s)\right|^2.
\end{align}
To see this, note that the following estimates are independent of $N$:
\begin{align*}
  -\left|\frac\theta2\sigma_i\pi_i\right|^2-\xi^2_i\leq-\frac12\left|\frac\theta2\sigma_i\pi_i-\xi_i\right|^2\rho_N(\xi_i)\leq0,
\end{align*}
and
\begin{align*}
  0&\geq-\lambda_i^{\rm M}(p,z)(1-\pi_i)^{-\frac\theta2}\hat{\rho}_N(e^{v_i})\geq-\lambda_i^{\rm M}(p,z)(1-\pi_i)^{-\frac\theta2}e^{v_i}\\
  &\geq -\frac{\big|\lambda_i^{\rm M}(p,z)\big|^2(1-\pi_i)^{-\theta}+e^{2v_i}}2.
\end{align*}
It then follows that
\begin{align}\label{ineq-1}
 -\xi^2_i-\frac12e^{2v_i}+h^{(1)}_i(\pi_i;p,z)\leq h^N_i(\pi_i;p,z,\xi_i,v_i)\leq h^{(2)}_i(\pi_i;p,z),
\end{align}
where the lower and upper bound functions are given by
\begin{align*}
  h^{(1)}_i(\pi_i;p,z)&:=-\frac\theta2\sigma_i^2\pi_i^2+\frac\theta2\left(\mu_i^{\rm M}(p)+\lambda^{\rm M}_i(p,z)-r\right)\pi_i+\lambda_i^{\rm M}(p,z)\nonumber\\
  &\quad-\frac12\left|\lambda_i^{\rm M}(p,z)\right|^2(1-\pi_i)^{-\theta},\\
h^{(2)}_i(\pi_i;p,z)&:=-\frac\theta4\sigma_i^2\pi_i^2+\frac\theta2\left(\mu_i^{\rm M}(p)+\lambda^{\rm M}_i(p,z)-r\right)\pi_i+\lambda_i^{\rm M}(p,z).
\end{align*}
Note that $h^{(1)}_i(\pi_i;p,z)$ and $h^{(2)}_i(\pi_i;p,z)$ are independent of $(N,\xi_i,v_i)$. Consequently, under the assumption ({\bf H}), there exists a constant $C$ independent of $N$,  such that
\begin{align}\label{ineq-2}
  \sup_{\pi_i\in(-\infty,1)}|h^{(1)}_i(\pi_i;p,z)|+\sup_{\pi_i\in(-\infty,1)}|h^{(2)}_i(\pi_i;p,z)|\leq C.
\end{align}
By \eqref{ineq-1} and \eqref{ineq-2}, we have that
\begin{align*}
  &\left|\sum_{i=1}^n(1-H_i(\omega,u))\sup_{\pi_i\in(-\infty,1)}h^N_i(\pi_i;p^{\rm M}(\omega,u),H(\omega,u),\xi,v)\right|\\
  &\qquad\leq C_1\sum_{i=1}^n(1-H_i(\omega,u))\left(\xi_i^2+\sum_{i=1}^ne^{v_i}+1\right).
\end{align*}
Similarly, we have the estimate of $h_L$ that
\begin{align}\label{h-L-estimate}
  \left|h_L(p,z,\xi,v)\right|\leq C_2\sum_{i=1}^n(1-H_i(\omega,u))\left(\xi_i^2+|v_i|+1\right),
\end{align}
where $C_2$ is independent of N. Plugging \eqref{ineq-2} and \eqref{h-L-estimate} into \eqref{eq:fNhN}, we obtain
\begin{align*}
  \left|f^N(\omega,u,\xi,v)\right|\leq C_3\sum_{i=1}^n(1-H_i(\omega,u))\left(\xi_i^2+|v_i|+\sum_{i=1}^ne^{v_i}+1\right),
\end{align*}
in which $C_3$ is hence independent of $N$. As a result, we get that
\begin{align}\label{eq:R45}
  &\left|\tilde{f}^N\big(s,Z^N(s),\hat{V}^N(s))\right|=\left|f^N\big(\omega,s,Z^N(s),\hat{V}^N(s))+f^N\big(\omega,s,0,0)\right|\nonumber\\
  &\qquad\leq C_3\sum_{i=1}^n(1-H_i(\omega,u))\left(\left|Z^N_i(s)\right|^2+\left|\hat{V}^N_i(s)\right|+\sum_{i=1}^ne^{\hat{V}^N_i(s)}+1\right)\nonumber\\
  &\quad\qquad+C_3(n+1)\sum_{i=1}^n(1-H_i(\omega,u)).
\end{align}
Therefore, the existence of $R_4$ and $R_5$ in the claim \eqref{eq:fquadratic} follows from \eqref{eq:R45} and the fact that $\|(1-H)\hat{V}^N\|_{t,\infty}\leq 2\|\zeta\|_{0,\infty}$.

Plugging \eqref{eq:fquadratic} into \eqref{eq:itoexpYN} and taking the conditional expectation under $\F^{\rm M}_u$, we attain that
\begin{align}\label{eq:ZNBMO}
&\left(\frac{\beta^2}{2}-R_5\beta\right)\sum_{i=1}^n\Ex^*\left[\int_u^{T\wedge\tau_i^u}e^{\beta\tilde{Y}^N(s)}\left|\tilde{Z}^N_i(s)\right|^2ds\Big|\F_u^{\rm M}\right]\leq \Ex^*\left[e^{\beta\zeta}\big|\F_u^{\rm M}\right]-e^{\beta\tilde{Y}^N(u)}\notag\\
&~\quad+R_4\beta\Ex^*\left[\int_u^Te^{\beta\tilde{Y}^N(s)}ds\Big|\F_u^{\rm M}\right]-\sum_{i=1}^n\Ex^*\left[\int_u^{T\wedge\tau_i^u}\{e^{\beta(\tilde{Y}^N(s)+\hat{V}^N_i(s))}-e^{\beta\tilde{Y}^N(s)}\}ds\Big|\F_u^{\rm M}\right]\notag\\
&~\quad+\sum_{i=1}^n\Ex^*\left[\int_u^{T\wedge\tau_i^u}\beta e^{\beta\tilde{Y}^N(s)}\hat{V}^N_i(s)ds\Big|\F_u^{\rm M}\right],\quad u\in[t,T].
\end{align}
For any constant $R_0>0$ independent of $N$, there exists a constant $\beta_0>0$ such that $\frac{\beta_0^2}{2}-R_5\beta_0=R_0$. Note that each term in r.h.s. of \eqref{eq:ZNBMO} is bounded by a positive constant, uniformly in $N$, say $R_6$. We then arrive at
\begin{align*}
\sum_{i=1}^n\Ex^*\left[\int_u^{T\wedge\tau_i^u}e^{-\beta_0\|\zeta\|_{0,\infty}}\left|\tilde{Z}^N_i(s)\right|^2ds\Big|\F_u^{\rm M}\right]&\leq\sum_{i=1}^n\Ex^*\left[\int_u^{T\wedge\tau_i^u}e^{\beta_0\tilde{Y}^N(s)}\left|\tilde{Z}^N_i(s)\right|^2ds\Big|\F_u^{\rm M}\right]\\
&\leq R_0^{-1}R_6,\quad \text{a.e.}
\end{align*}
This implies that
\begin{align*}
\sum_{i=1}^n\Ex^*\left[\int_u^{T}\left|\tilde{Z}^N_i(s)\right|^2ds\Big|\F_u^{\rm M}\right]\leq e^{\beta_0\|\zeta\|_{0,\infty}}R_0^{-1}R_6,\quad \text{a.e.},
\end{align*}
which concludes the desired estimation \eqref{eq:estiBMO}.
\end{proof}

We also state here a comparison result for the truncated BSDE that will be used in later sections. Its proof is deferred to Appendix \ref{app:proof1}.
\begin{lemma}\label{lem:monoYN}
For any $N\geq1$, let $(\tilde{Y}^N,\tilde{Z}^N,\tilde{V}^N)\in{\cal S}_t^2\times L_t^2\times L_t^2$ be the solution of \eqref{tuncatedBSDE}. There exists a constant $N_0>0$ such that, for $u\in[t,T]$, $\tilde{Y}^N(u)$ is increasing for all $N\geq N_0$, $\Px^*$-a.s..
\end{lemma}

\subsection{Convergence of solutions of truncated BSDEs}
Aiming to prove the existence of solution to the original BSDE~\eqref{eq:regulizeBSDE}, we continue to show that the solutions associated to truncated BSDEs \eqref{tuncatedBSDE} converge as $N\rightarrow \infty$ and the limit process is the desired solution of BSDE~\eqref{eq:regulizeBSDE} in an appropriate space.

For any compact set ${\cal C}\subset\R^n$, we choose $N$ large enough such that $e^{|y|}\leq N$ for all $y\in{\cal C}$. By virtue of \eqref{eq:fNhN}, we have that, $\Px$-a.s., $f^N(u,\xi,v)=f(u,\xi,v)$ for all $u\in[t,T]$ and $\xi,v\in{\cal C}$. This implies the locally uniform (almost surely) convergence of $f^N$ to $f$, i.e., it holds that $\sup_{(u,\xi,v)\in[t,T]\times{\cal C}^2}|f^N(u,\xi,v)-f(u,\xi,v)|\to 0$, $N\to\infty$, a.s. We first have the next convergence result of the truncated solutions $(\tilde{Y}^N,\tilde{Z}^N,\tilde{V}^N)$ given in Lemma~\ref{lem:existYN}. Thanks to Lemma \ref{lem:prioriestimate2}, it is known that $\tilde{V}^N$ is $d\Px^*\otimes du$-a.e. bounded by a constant $C_T$ for all $N\geq 1$.
\begin{lemma}\label{lem:converRes}
There exist an $\Fx^{\rm M}$-adapted process $\tilde{Y}=(\tilde{Y}(u))_{u\in[t,T]}$ and processes $(\tilde{Z},\tilde{V})\in L_t^2\times L_t^2$ such that, for $u\in[t,T]$, $\tilde{Y}^N(u)\to\tilde{Y}(u)$, $\Px^*$-a.s., $\tilde{Z}^N\to\tilde{Z}$ weakly in $L_t^2$, and $\tilde{V}^N\to\tilde{V}$ weakly in $L_t^2$, as $N\to\infty$.
\end{lemma}

\begin{proof} By Lemma~\ref{lem:monoYN}, we have that $N\to\tilde{Y}^N(u)$ is increasing, $\Px^*$-a.e. for $u\in[t,T]$. Lemma~\ref{lem:prioriestimate2} gives that $\tilde{Y}^N=(\tilde{Y}^N(u))_{u\in[t,T]}$ is uniformly bounded in ${\cal S}^\infty_t$. Then, there exists an $\Fx^{\rm M}$-adapted process $\tilde{Y}=(\tilde{Y}(u))_{u\in[t,T]}$ such that, for $u\in[t,T]$, $\tilde{Y}^N(u)\to\tilde{Y}(u)$, as $N\to\infty$, $\Px^*$-a.e..
It follows from Lemma~\ref{lem:prioriestimate2} that the sequence of $\Fx^{\rm M}$-predictable solutions $\tilde{Z}^N=(\tilde{Z}^N(u))_{u\in[t,T]}$ for $N\geq1$ is bounded in $L_t^2$. Hence, there exists a process $\tilde{Z}=(\tilde{Z}(u))_{u\in[t,T]}\in L_t^2$ such that $\tilde{Z}^N\to\tilde{Z}$ weakly in $L_t^2$. Moreover, by Lemma~\ref{lem:existYN}, the sequence of $\int_t^\cdot \tilde{V}^N(u)^\top d\Upsilon^*(u)$ for $N\geq1$ is bounded in $L_t^2$. Thanks to the martingale representation theorem in Protter~\cite{Protter2005} and the weak compactness of $L^2$, there exists a process $\tilde{V}=(\tilde{V}(u))_{u\in[t,T]}\in L_t^2$ such that $\tilde{V}^N\to\tilde{V}$ (up to a subsequence) weakly in $L_t^2$ as $N\to\infty$. We claim  that $\tilde{V}$ is predictable. Indeed, by using Mazur's lemma, we deduce the existence of a sequence of convex combinations of $\tilde{V}^N$ for $N\geq1$, which converges to $\tilde{V}$ pointwise. Because every convex combination of $\tilde{V}^N$ is predictable, $\tilde{V}$ is also predictable.
\end{proof}

Let us continue to prove the strong convergence result of the truncated solutions $(\tilde{Y}^N,\tilde{Z}^N,\tilde{V}^N)$ for $N\geq1$ given in Lemma~\ref{lem:existYN} to the limit process $(\tilde{Y},\tilde{Z},\tilde{V})$ given in Lemma~\ref{lem:converRes}.
\begin{lemma}\label{lem:strongconvergenceZN}
The sequence $(\tilde{Z}^N)_{N\geq 1}$ converges to $\tilde{Z}$ in $L_t^2$ as $N\to\infty$.
\end{lemma}

\begin{proof} To ease the notation in the rest of the proof, we set $\tilde{f}^N(u):=\tilde{f}^N(u,\tilde{Z}^N(u),\tilde{V}^N(u))$ for $u\in[t,T]$. Let $N_2\geq N_1\geq1$ be two integers and $\phi:\R\to\R_+$ be a smooth function that will be determined later. For $Y^{\rm d}(u):=\tilde{Y}^{N_2}(u)-\tilde{Y}^{N_1}(u)\geq0$, a.e., using Lemma~\ref{lem:existYN} and It\^o's formula, we have that
\begin{align}\label{eq:itopsi}
&\phi(0)-\phi(Y^{\rm d}(t))\nonumber\\
&\quad=\int_t^T\phi'(Y^{\rm d}(u))(\tilde{f}^{N_2}(u)-\tilde{f}^{N_1}(u))du+\int_t^T\phi'(Y^{\rm d}(u))(\tilde{Z}^{N_2}(u)-\tilde{Z}^{N_1}(u))^\top dW^{o,\tau}(u)\notag\\
&\qquad-\sum_{i=1}^n\int_t^{T\wedge\tau_i^t}\phi'(Y^{\rm d}(u))(\tilde{V}^{N_2}_i(u)-\tilde{V}^{N_1}_i(u))du\nonumber\\
&\qquad+\frac12\sum_{i=1}^n\int_t^{T\wedge\tau_i^t}\phi{''}(Y^{\rm d}(u))\left|\tilde{Z}^{N_2}_i(u)-\tilde{Z}^{N_1}_i(u)\right|^2du\nonumber\\
&\qquad+\sum_{i=1}^n\int_t^T\{\phi(Y^{\rm d}(u-)+\tilde{V}^{N_2}_i(u)-\tilde{V}^{N_1}_i(u))-\phi(Y^{\rm d}(u-))\}dH_i(u).
\end{align}
In view of \eqref{eq:randomdirver} and Lemma~\ref{lem:prioriestimate2}, for all $u\in[t,T]$, there exist positive constants $R_i$ with $i=1,2,3$ which are independent of $N$ and $u$ such that, a.e.
\begin{align}\label{eq:deltafquadratic}
&\left|\tilde{f}^{N_2}(u)-\tilde{f}^{N_1}(u)\right|\leq R_1+R_2\sum_{i=1}^n(1-H_i(u))\left\{\left|\tilde{Z}^{N_1}_i(u)\right|^2+\left|\tilde{Z}^{N_2}_i(u)\right|^2\right\}\\
&\quad\leq R_1+R_3\sum_{i=1}^n(1-H_i(u))\left\{\left|\tilde{Z}^{N_1}_i(u)-\tilde{Z}^{N_2}_i(u)\right|^2+\left|\tilde{Z}^{N_1}_i(u)-\tilde{Z}_i(u)\right|^2
+\left|\tilde{Z}_i(u)\right|^2\right\}.\notag
\end{align}
We choose $\phi(x)=e^{\beta x}-\beta x-1$ for $x\in\R$, where $\beta$ is a positive constant satisfying $\beta>4R_3$. Then $\phi$ enjoys the properties that $\phi(x)\geq0$ for all $x\in\R$, $\phi(0)=\phi'(0)=0$, $\phi'(x)\geq0$ for $x\in\R_+$, and $\phi''(x)-4R_3\phi'(x)=(\beta^2-4R_3\beta)e^{\beta x}+4R_3\beta>0$ for all $x\in\R$. Plugging \eqref{eq:deltafquadratic} into \eqref{eq:itopsi} and manipulating terms on both sides, we obtain that
\begin{align}\label{eq:ZN-Zestimate1}
&\frac12\sum_{i=1}^n\int_t^{T\wedge\tau_i^t}\phi{''}(Y^{\rm d}(u))\left|\tilde{Z}^{N_1}_i(u)-\tilde{Z}^{N_2}_i(u)\right|^2du\nonumber\\
&\qquad-R_3\sum_{i=1}^n\int_0^{T\wedge\tau_i^t}\phi'(Y^{\rm d}(u))\left|\tilde{Z}^{N_1}_i(u)-\tilde{Z}^{N_2}_i(u)\right|^2du\notag\\
&\quad\leq\phi(0)-\phi(Y^{\rm d}(t))+R_3\sum_{i=1}^n\int_t^{T\wedge\tau_i^t}\phi'(Y^{\rm d}(u))\left|\tilde{Z}^{N_1}_i(u)-\tilde{Z}_i(u)\right|^2du\nonumber\\
&\qquad+R_1\int_t^T\phi'(Y^{\rm d}(u))du
+R_3\sum_{i=1}^n\int_t^{T\wedge\tau_i^t}\phi'(Y^{\rm d}(u))\left|\tilde{Z}_i(u)\right|^2du\nonumber\\
&\qquad-\int_t^T\phi'(Y^{\rm d}(u))(\tilde{Z}^{N_2}(u)-\tilde{Z}^{N_1}(u))^\top dW^{o,\tau}(u)\notag\\
&\qquad+\sum_{i=1}^n\int_t^{T\wedge\tau_i^t}\phi'(Y^{\rm d}(u))(\tilde{V}^{N_2}_i(u)-\tilde{V}^{N_1}_i(u))du\\
&\qquad-\sum_{i=1}^n\int_t^T\{\phi(Y^{\rm d}(u-)+\tilde{V}^{N_2}_i(u)-\tilde{V}^{N_1}_i(u))-\phi(Y^{\rm d}(u-))\}dH_i(u).\notag
\end{align}
On the other hand, it follows from Lemma~\ref{lem:converRes} that $\tilde{Z}^{N_2}$ converges weakly to $\tilde{Z}$ in $L_t^2$ as $N_2\to\infty$. We next prove that, for $i=1,\ldots,n$, as $N_2\to\infty$,
\begin{align}\label{eq:weakconv}
&\sqrt{\left(\frac12\phi{''}-R_3\phi'\right)(Y^{\rm d}(u))}(1-H_i)(\tilde{Z}^{N_1}_i-\tilde{Z}^{N_2}_i)\\
&\qquad\to\sqrt{\left(\frac12\phi{''}-R_3\phi'\right)(\tilde{Y}-\tilde{Y}^{N_1})}(1-H_i)(\tilde{Z}^{N_1}_i-\tilde{Z}_i),~{\rm weakly\ in}\ L^2([t,T]\times\Omega;\Px^*).\nonumber
\end{align}
Thanks to the fact that $(\tilde{Y}^N)_{N\geq1}$ and $\tilde{Y}$ are bounded, we have that, for $u\in[t,T]$,
\begin{align*}
\delta Y^{N_2}(u):=\left(\frac12\phi{''}-R_3\phi'\right)^{\frac12}(\tilde{Y}^{N_2}(u)-\tilde{Y}^{N_1}(u))
-\left(\frac12\phi{''}-R_3\phi'\right)^{\frac12}(\tilde{Y}(u)-\tilde{Y}^{N_1}(u))
\end{align*}
is also bounded and tends to $0$ as $N_2\rightarrow\infty$. Moreover, the weak convergence of $(\tilde{Z}^N)_{N\geq1}$ in $L_t^2$ implies that they are uniformly bounded in $L_t^2$ by the Resonance Theorem, which can also be deduced from $(\tilde{Z}^{N})_{N\geq1}\in{\Hxx}_{t,{\rm BMO}}^2$ by Lemma \ref{lem:prioriestimate2}. Cauchy-Schwartz inequality then gives that, for all $X\in L^2([t,T]\times\Omega;\Px^*)$,
\begin{align*}
\lim_{N_2\to\infty}\Ex^*\left[\int_t^{T\wedge\tau_i^t}\delta Y^{N_2}(u)(\tilde{Z}^{N_1}_i(u)-\tilde{Z}^{N_2}_i(u))X(u)du\right]=0.
\end{align*}
Hence, it holds that
\begin{align*}
&\lim_{N_2\to\infty}\Ex^*\left[\int_t^{T\wedge\tau_i^t}\left(\frac12\phi{''}-R_3\phi'\right)^{\frac12}(Y^{\rm d}(u))(\tilde{Z}^{N_1}_i(u)-\tilde{Z}^{N_2}_i(u))X(u)du\right]\notag\\
&\quad=\lim_{N_2\to\infty}\Ex^*\left[\int_t^{T\wedge\tau_i^t}\left(\frac12\phi{''}-R_3\phi'\right)^{\frac12}(Y(u)-Y^{N_1}(u))
(\tilde{Z}^{N_1}_i(u)-\tilde{Z}^{N_2}_i(u))X(u)du\right]\notag\\
&\quad\qquad+\lim_{N_2\to\infty}\Ex^*\left[\int_t^{T\wedge\tau_i^t}\delta Y^{N_2}(u)(\tilde{Z}^{N_1}_i(u)-\tilde{Z}^{N_2}_i(u))X(u)du\right]\notag\\
&\quad=\Ex^*\left[\int_t^{T\wedge\tau_i^t}\left(\frac12\phi{''}-R_3\phi'\right)^{\frac12}(Y(u)-Y^{N_1}(u))(\tilde{Z}^{N_1}_i(u)-\tilde{Z}_i(u))X(u)du\right],
\end{align*}
which proves \eqref{eq:weakconv}. By using the property of convex functional and weak convergence (see Theorem 1.4 in \cite{Figueiredo1989}), as $N_2\to\infty$, we deduce that the l.h.s. of \eqref{eq:ZN-Zestimate1} satisfies that
\begin{align}\label{eq:lhsN2}
&\liminf_{N_2\to\infty}\sum_{i=1}^n\Ex^*\left[\int_t^{T\wedge\tau_i^t}\left(\frac12\phi{''}-R_3\phi'\right)(Y^{\rm d}(u))\left|\tilde{Z}^{N_1}_i(u)-\tilde{Z}^{N_2}_i(u)\right|^2du\right]\notag\\
&\qquad\geq\sum_{i=1}^n\Ex^*\left[\int_t^{T\wedge\tau_i^t}\left(\frac12\phi{''}-R_3\phi'\right)(\tilde{Y}(u)-\tilde{Y}^{N_1}(u))\left|\tilde{Z}^{N_1}_i(u)-\tilde{Z}_i(u)\right|^2du\right].
\end{align}

For the jump term in the r.h.s. of \eqref{eq:ZN-Zestimate1}, as $\phi(x)\geq0$ for all $x\in\R$, we get that
\begin{align}\label{eq:expterm}
&\sum_{i=1}^n\Ex^*\left[\int_t^{T\wedge\tau_i^t}\phi'(Y^{\rm d}(u))(\tilde{V}^{N_2}_i(u)-\tilde{V}^{N_1}_i(u)\big)du\right]\notag\\
&\qquad\qquad-\sum_{i=1}^n\Ex^*\left[\int_t^T\left(\phi(Y^{\rm d}(u-)+\tilde{V}^{N_2}_i(u)-\tilde{V}^{N_1}_i(u))-\phi(Y^{\rm d}(u-))\right)dH_i(u)\right]\notag\\
&\qquad=-\sum_{i=1}^n\Ex^*\left[\int_t^{T\wedge\tau_i^t}e^{\beta{Y}^{\rm d}(u)}\phi(\tilde{V}^{N_2}_i(u)-\tilde{V}^{N_1}_i(u))du\right]\leq0.
\end{align}
Thanks to \eqref{eq:lhsN2}, \eqref{eq:expterm} and Dominated Convergence Theorem, it follows from \eqref{eq:ZN-Zestimate1} that
\begin{align*}
&\sum_{i=1}^n\Ex^*\left[\int_t^{T\wedge\tau_i^t}\left(\frac12\phi{''}-R_3\phi'\right)(\tilde{Y}(u)-\tilde{Y}^{N_1}(u))\left|\tilde{Z}^{N_1}_i(u)-\tilde{Z}_i(u)\right|^2du\right]\notag\\
&\quad\leq R_3\sum_{i=1}^n\Ex^*\left[\int_t^{T\wedge\tau_i^t}\phi'(\tilde{Y}(u)-\tilde{Y}^{N_1}(u))\left|\tilde{Z}^{N_1}_i(u)-\tilde{Z}_i(u)\right|^2du\right]\notag\\
&\qquad+R_3\sum_{i=1}^n\Ex^*\left[\int_t^T\phi'(\tilde{Y}(u)-\tilde{Y}^{N_1}(u))\left|\tilde{Z}_i(u)\right|^2du\right]\nonumber\\
&\qquad+R_1\Ex^*\left[\int_t^T\phi'(\tilde{Y}(u)-\tilde{Y}^{N_1}(u))du\right].
\end{align*}
Thanks to Lemma~\ref{lem:prioriestimate2} and Lemma~\ref{lem:converRes}, we have that $\|\tilde{Y}\|_{t,\infty}\leq\|\zeta\|_{0,\infty}$. By choosing $R_4:=\frac12(\beta^2-4R_3\beta)e^{-2\beta|\zeta|_{\infty}}>0$, we obtain that
\begin{align}\label{eq:ZN-Zestimate2}
&R_4\sum_{i=1}^n\Ex^*\left[\int_t^{T\wedge\tau_i^t}\left|\tilde{Z}^{N_1}_i(u)-\tilde{Z}_i(u)\right|^2du\right]\nonumber\\
\leq&\frac12\sum_{i=1}^n\Ex^*\left[\int_t^{T\wedge\tau_i^t}\{\phi{''}-4R_3\phi'\}(\tilde{Y}(u)-\tilde{Y}^{N_1}(u))
\left|\tilde{Z}^{N_1}_i(u)-\tilde{Z}_i(u)\right|^2du\right]\notag\\
\leq& R_3\sum_{i=1}^n\Ex^*\left[\int_t^{T\wedge\tau_i^t}\phi'(\tilde{Y}(u)-\tilde{Y}^{N_1}(u))\left|\tilde{Z}_i(u)\right|^2du\right]\notag\\
&+R_1\Ex^*\left[\int_t^{T\wedge\tau_i^t}\phi'(\tilde{Y}(u)-\tilde{Y}^{N_1}(u))du\right].
\end{align}
Note that $\phi'(0)=0$ and that for each $u\in[t,T]$, $\tilde{Y}^N(u)\uparrow\tilde{Y}(u)$ as $N\rightarrow\infty$. Dominated Convergence Theorem gives that the r.h.s. of \eqref{eq:ZN-Zestimate2} tends to zero as $N_1\to\infty$. Then, the estimate \eqref{eq:ZN-Zestimate2} implies that
\begin{align*}
\lim_{N_1\to\infty}\sum_{i=1}^n\Ex^*\left[\int_t^{T\wedge\tau_i^t}\left|\tilde{Z}^{N_1}_i(u)-\tilde{Z}_i(u)\right|^2du\right]=0,
\end{align*}
which completes the proof.
\end{proof}

\begin{lemma}\label{lem:VNstrongcon}
The sequence $(\tilde{V}^N)_{N\geq 1}$ converges to $\tilde{V}$ in $L_t^2$ as $N\to\infty$. Therefore, $\tilde{V}$ is also $d\Px^*\otimes du$-a.e. bounded by some constant $C_T$.
\end{lemma}

\begin{proof}
Let us take $\phi(x)=x^2$ for $x\in\R$. Then \eqref{eq:itopsi} can be reduced to
\begin{align*}
-\Ex\left[\left|Y^{\rm d}(t)\right|^2\right]&=2\Ex^*\left[\int_t^{T}Y^{\rm d}(u)(\tilde{f}^{N_2}(u)-\tilde{f}^{N_1}(u))du\right]\nonumber\\
&\quad-2\sum_{i=1}^n\Ex^*\left[\int_t^{T\wedge\tau_i^t}Y^{\rm d}(u)(\tilde{V}^{N_2}_i(u)-\tilde{V}^{N_1}_i(u))du\right]\nonumber\\
&\quad+\sum_{i=1}^n\Ex^*\left[\int_t^{T\wedge\tau_i^t}\left|\tilde{Z}^{N_2}_i(u)-\tilde{Z}^{N_1}_i(u)\right|^2du\right]\nonumber\\
&\quad+\sum_{i=1}^n\Ex^*\left[\int_t^{T\wedge\tau_i^t}\left(|Y^{\rm d}(u-)+\tilde{V}^{N_2}_i(u)-\tilde{V}^{N_1}_i(u)|^2-|Y^{\rm d}(u-)|^2\right)du\right].
\end{align*}
It follows from \eqref{eq:deltafquadratic} that
\begin{align}\label{eq:VN-V}
&\sum_{i=1}^n\Ex^*\left[\int_t^{T\wedge\tau_i^t}\left|\tilde{V}^{N_2}_i(u)-\tilde{V}^{N_1}_i(u)\right|^2du\right]\notag\\
&\quad\leq2R_2\sum_{i=1}^n\Ex^*\left[\int_t^{T\wedge\tau_i^t}\left|Y^{\rm d}(u)\right|\left(|\tilde{Z}^{N_1}_i(u)|^2+|\tilde{Z}^{N_2}_i(u)|^2\right)du\right]\\
&\qquad-\Ex^*\left[\left|Y^{\rm d}(t)\right|^2\right]+2R_1\Ex^*\left[\int_t^T\left|Y^{\rm d}(u)\right|du\right]\nonumber\\
&\qquad-\sum_{i=1}^n\Ex^*\left[\int_t^{T\wedge\tau_i^t}\left|\tilde{Z}^{N_2}_i(u)-\tilde{Z}^{N_1}_i(u)\right|^2du\right].\notag
\end{align}
Moreover, for $i=1,\ldots,n$, we also have that
\begin{align}\label{eq:intZN2}
&\Ex^*\left[\int_t^{T\wedge\tau_i^t}\left|Y^{\rm d}(u)\right|\left|\tilde{Z}^{N_2}_i(u)\right|^2du\right]\notag\\
\leq &2\Ex^*\left[\int_t^{T\wedge\tau_i^t}\left|Y^{\rm d}(u)\right|\left|\tilde{Z}^{N_2}_i(s)-\tilde{Z}_i(u)\right|^2du\right]
+2\Ex^*\left[\int_t^{T\wedge\tau_i^t}\left|Y^{\rm d}(u)\right|\left|\tilde{Z}_i(u)\right|^2du\right]\\
\leq & 4\|\zeta\|_{0,\infty}\Ex^*\left[\int_t^{T\wedge\tau_i^t}\left|\tilde{Z}^{N_2}_i(u)-\tilde{Z}_i(u)\right|^2du\right]
+2\Ex^*\left[\int_t^{T\wedge\tau_i^t}\left|Y^{\rm d}(u)\right|\left|\tilde{Z}_i(u)\right|^2du\right].\notag
\end{align}
We can derive from \eqref{eq:VN-V} and \eqref{eq:intZN2} that
\begin{align*}
&\sum_{i=1}^n\Ex^*\left[\int_t^{T\wedge\tau_i^t}\left|\tilde{V}^{N_2}_i(u)-\tilde{V}^{N_1}_i(u)\right|^2du\right]\nonumber\\
&\quad\quad\leq2R_1\Ex^*\left[\int_t^T\left|Y^{\rm d}(u)\right|du\right]+2R_2\sum_{i=1}^n\Ex\left[\int_t^{T\wedge\tau_i^t}\left|Y^{\rm d}(u)\right|\left|\tilde{Z}^{N_1}_i(u)\right|^2du\right]\notag\\
&\quad\qquad+4R_2\sum_{i=1}^n\Ex^*\left[\int_t^{T\wedge\tau_i^t}\left|Y^{\rm d}(u)\right|\left|\tilde{Z}_i(u)\right|^2du\right]\nonumber\\
&\quad\qquad+8R_2\|\zeta\|_{0,\infty}\Ex^*\left[\int_t^{T\wedge\tau_i^t}\left|\tilde{Z}^{N_2}_i(u)-\tilde{Z}_i(u)\right|^2du\right].
\end{align*}
Letting $N_2\to\infty$ and using Dominated Convergence Theorem and Lemma~\ref{lem:strongconvergenceZN}, we obtain that
\begin{align*}
&\liminf_{N_2\rightarrow\infty}\sum_{i=1}^n\Ex^*\left[\int_t^{T\wedge\tau_i^t}\left|\tilde{V}_i^{N_2}(u)-\tilde{V}^{N_1}_i(u)\right|^2du\right]\\
\leq& 2R_1\Ex^*\left[\int_t^T\left|\tilde{Y}(u)-\tilde{Y}^{N_1}(u)\right|du\right]
+2R_2\sum_{i=1}^n\Ex^*\left[\int_t^{T\wedge\tau_i^t}\left|\tilde{Y}(u)-\tilde{Y}^{N_1}(u)\right|\left|\tilde{Z}^{N_1}_i(u)\right|^2du\right]\nonumber\\
&+4R_2\sum_{i=1}^n\Ex^*\left[\int_t^{T\wedge\tau_i^t}\left|\tilde{Y}(u)-\tilde{Y}^{N_1}(u)\right|\left|\tilde{Z}_i(u)\right|^2du\right]\nonumber\\
\leq & 2R_1\Ex^*\left[\int_t^T\left|\tilde{Y}(u)-\tilde{Y}^{N_1}(u)\right|du\right]
+8R_2\sum_{i=1}^n\Ex^*\left[\int_t^{T\wedge\tau_i^t}\left|\tilde{Y}(u)-\tilde{Y}^{N_1}(u)\right|\left|\tilde{Z}_i(u)\right|^2du\right]\nonumber\\
&+8R_2\|\zeta\|_{0,\infty}\sum_{i=1}^n\Ex^*\left[\int_t^{T\wedge\tau_i^t}\left|\tilde{Z}^{N_1}_i(u)-\tilde{Z}_i(u)\right|^2du\right].
\end{align*}
Thanks to the property of convex functional and weak convergence (see, e.g., Theorem 1.4 in \cite{Figueiredo1989}), one can get that
\begin{align}\label{eq:VN-V4}
&\sum_{i=1}^n\Ex^*\left[\int_t^{T\wedge\tau_i^t}\left|\tilde{V}_i(u)-\tilde{V}^{N_1}_i(u)\right|^2du\right]\notag\\
\leq &2R_1\Ex^*\left[\int_t^T\left|\tilde{Y}(u)-\tilde{Y}^{N_1}(u)\right|du\right]+8R_2\sum_{i=1}^n\Ex^*\left[\int_t^{T\wedge\tau_i^t}\left|\tilde{Y}(u)-\tilde{Y}^{N_1}(u)\right|\left|\tilde{Z}_i(u)\right|^2du\right]\notag\\
&+8R_2\|\zeta\|_{0,\infty}\sum_{i=1}^n\Ex^*\left[\int_t^{T\wedge\tau_i^t}\left|\tilde{Z}^{N_1}_i(u)-\tilde{Z}_i(u)\right|^2du\right].
\end{align}
The desired convergence that $\tilde{V}^N\rightarrow \tilde{V}$ in $L_t^2$ can be derived by Dominated Convergence Theorem and Lemma~\ref{lem:strongconvergenceZN} as $N_1\to\infty$. The boundedness of $\tilde{V}$ is consequent on the uniform boundedness of $\tilde{V}^N$, $N\geq1$.
\end{proof}

We finally present the main result of this section on the existence of a solution to the original BSDE~\eqref{eq:regulizeBSDE}.
\begin{theorem}\label{thm:limitsolutionBSDE}
Let $(\tilde{Y},\tilde{Z},\tilde{V})$ be the limiting process given in Lemma~\ref{lem:converRes}. Then, $(\tilde{Y},\tilde{Z},\tilde{V})\in{\cal S}_t^{\infty}\times{\Hxx}_{t,{\rm BMO}}^2\times L_t^2$ is a solution of BSDE~\eqref{eq:regulizeBSDE}.
\end{theorem}

\begin{proof}
We first prove that $\tilde{Y}^N$ converges to $\tilde{Y}$ in the uniform norm as $N\to\infty$, a.s. In fact, for the fixed $t\in[0,T]$ and any $u\in[t,T]$, we first have that
\begin{align}\label{eq:YNcauchy}
\sup_{u\in[t,T]}\left|\tilde{Y}^{N_1}(u)-\tilde{Y}^{N_2}(u)\right|&\leq \int_t^T\left|\tilde{f}^{N_1}(s)-\tilde{f}^{N_2}(s)\right|ds\nonumber\\
&\quad+\sup_{u\in[t,T]}\left|\int_u^T(\tilde{Z}^{N_1}(s)-\tilde{Z}^{N_2}(s))^\top dW^{o,\tau}(s)\right|\notag\\
&\quad+\sup_{u\in[t,T]}\left|\int_u^T(\tilde{V}^{N_1}(s)-\tilde{V}^{N_2}(s))^\top d\Upsilon^*(s)\right|.
\end{align}
Taking into account Lemma~\ref{lem:strongconvergenceZN} and Lemma 2.5 in \cite{Koby2000}, we obtain that, for each $i=1,\ldots,n$, there exists a subsequence $\{N_l\}$ such that
\begin{align}\label{eq:convZK}
(1-H)\tilde{Z}^{N_l}\to (1-H)\tilde{Z},\ d\Px^*\otimes du\text{-a.e.},\ \ {\rm and}\ \hat{Z}=(\hat{Z}_1,\ldots,\hat{Z}_n)\in L_t^2,
\end{align}
where $\hat{Z}_i(u):=\sup_{l\geq1}|(1-H_i(u))\tilde{Z}_i^{N_l}(u)|$ for $u\in[t,T]$. Moreover, Lemma~\ref{lem:VNstrongcon} implies that for some subsequence $\{N_{l_k}\}\subset\{N_l\}$, it holds that $(1-H)\tilde{V}^{N_{l_k}}\to(1-H)\tilde{V}$, as $k\to\infty$, $d\Px^*\otimes du$-a.e.. To ease the notation, the subsequence is still denoted by $\{N\}$. By the definition of $\tilde{f}^N$ and the fact that the random function $\tilde{f}$ is a.s. continuous in its domain, we have that
\begin{align}\label{eq:cmtpp}
\lim_{N\to\infty}\tilde{f}^N(u,\tilde{Z}^N(u),\tilde{V}^N(u))du=\tilde{f}(u,\tilde{Z}(u),\tilde{V}(u)),\quad d\Px^*\otimes du\text{-a.e.}
\end{align}
In light of \eqref{eq:randomdirver} and Lemma~\ref{lem:prioriestimate2}, for all $u\in[t,T]$, there exist constants $R_1,R_2>0$ independent of $N$ and $u$ such that
\begin{align*}
\left|\tilde{f}^N(u,\tilde{Z}^N(u),\tilde{V}^N(u))\right|&\leq R_1+R_2\sum_{i=1}^n(1-H_i(u))\left|\tilde{Z}^N_i(u)\right|^2\nonumber\\
&\leq R_1+R_2\sum_{i=1}^n(1-H_i(u))\left|\hat{Z}_i(u)\right|^2.
\end{align*}
Note that $\hat{Z}\in L_t^2$. Together with above inequality and \eqref{eq:cmtpp}, Dominated Convergence Theorem gives that
\begin{align}\label{eq:convfN}
\lim_{N\rightarrow\infty}\Ex\left[\int_t^T\left|\tilde{f}^N(u,\tilde{Z}^N(u),\tilde{V}^N(u))-\tilde{f}(u,\tilde{Z}(s),\tilde{V}(u))\right|du\right]=0.
\end{align}

The BDG inequality then implies the existence of constants $R_3,R_4>0$ independent of $N$ such that
\begin{align*}
&\Ex^*\left[\sup_{u\in[t,T]}\left|\int_u^T(\tilde{Z}^N(s)-\tilde{Z}(s))^\top dW^{o,\tau}(s)\right|^2\right]\\
&\qquad\leq 2\Ex^*\left[\left|\int_t^T(\tilde{Z}^N(s)-\tilde{Z}(s))^\top dW^{o,\tau}(s)\right|^2\right]\nonumber\\
&\qquad\quad+2\Ex^*\left[\sup_{u\in[t,T]}\left|\int_t^u (\tilde{Z}^N(s)-\tilde{Z}(s))^\top dW^{o,\tau}(s)\right|^2\right]\\
&\qquad\leq R_3\sum_{i=1}^n\Ex^*\left[\int_t^{T\wedge\tau_i^t}\left|\tilde{Z}^N_i(s)-\tilde{Z}_i(s)\right|^2ds\right].
\end{align*}
In a similar fashion, we also attain that
\begin{align*}
&\Ex^*\left[\sup_{u\in[t,T]}\left|\int_u^T(\tilde{V}^N(s)-\tilde{V}(s))^\top d\Upsilon^{*}(s)\right|^2\right]\leq R_4\sum_{i=1}^n\Ex^*\left[\int_t^{T\wedge\tau_i^t}\left|\tilde{V}^N_i(s)-\tilde{V}_i(s)\right|^2ds\right].
\end{align*}
Because of Lemma~\ref{lem:strongconvergenceZN} and Lemma~\ref{lem:VNstrongcon}, we have that
\begin{align}
&\lim_{N\rightarrow\infty}\Ex^*\left[\sup_{u\in[t,T]}\left|\int_u^T(\tilde{Z}^N(s)-\tilde{Z}(s))^\top dW^{o,\tau}(s)\right|^2\right]\notag\\
&\qquad=\lim_{N\rightarrow\infty}\Ex^*\left[\sup_{u\in[t,T]}\left|\int_u^T(\tilde{V}^N(s)-\tilde{V}(s))^\top d\Upsilon^{*}(s)\right|^2\right]=0.
\end{align}
Consequently, there exists a subsequence (still denoted by $N$) such that \eqref{eq:convfN} holds and
\begin{align}\label{eq:fNlim}
&\lim_{N\to\infty}\sup_{u\in[t,T]}\left|\int_t^T(\tilde{Z}^N(s)-\tilde{Z}(s))^\top dW^{o,\tau}(s)\right|=0,\ \text{a.e.},\\
\label{eq:VNlim}
&\lim_{N\to\infty}\sup_{u\in[t,T]}\left|\int_t^T(\tilde{V}^N(s)-\tilde{V}(s))^\top d\Upsilon^*(s)\right|=0,\ \text{a.e..}
\end{align}
We deduce by \eqref{eq:YNcauchy}, \eqref{eq:fNlim} and \eqref{eq:VNlim} that $(\tilde{Y}^N)_{N\geq1}$ is a Cauchy sequence a.e. under the uniform norm, and its limiting process coincides with $\tilde{Y}$ by Lemma~\ref{lem:converRes}. Thus, $\lim_{N\to\infty}\sup_{u\in[t,T]}|\tilde{Y}^N(u)-\tilde{Y}(u)|=0$, a.e.. By taking the limit on both sides of the equation, we obtain
\begin{align*}
\zeta-Y^N(t)&=\int_t^T\tilde{f}^N(u,\tilde{Z}^N(u),\tilde{V}^N(u))du+\int_t^T\tilde{Z}^N(u)^\top dW^{o,\tau}(u)\nonumber\\
&\quad+\int_t^T \tilde{V}^N(u)^\top d\Upsilon^*(u),
\end{align*}
and applying the established convergence results in \eqref{eq:convfN}, \eqref{eq:fNlim} and \eqref{eq:VNlim}, we can conclude that $(\tilde{Y},\tilde{Z},\tilde{V})\in{\cal S}_t^{\infty}\times{\Hxx}_{t,{\rm BMO}}^2\times L_t^2$ is indeed a solution of BSDE~\eqref{eq:regulizeBSDE}.
\end{proof}

\section{Optimal investment strategy}\label{sec:optimum}

At last, we characterize the optimal control strategy using the verification result in Lemma~\ref{lem:BSDE-valuefcn}, our newly established BSDE results and some properties of BMO martingales. It is noted that if $(\tilde{Y},\tilde{Z},\tilde{V})\in{\cal S}_t^{\infty}\times{\Hxx}_{t,{\rm BMO}}^2\times L_t^2$ is the solution of BSDE~\eqref{eq:regulizeBSDE} given in Theorem~\ref{thm:limitsolutionBSDE}, then $(\tilde{Y}+\int_t^\cdot f(p^{\rm M}(s),H(s),0,0)ds,\tilde{Z},\tilde{V})$ solves the original BSDE~\eqref{eq:formalBSDE}.
We also recall that by Lemma \ref{lem:VNstrongcon}, $\tilde{V}$ is $d\Px^*\otimes du$-a.e. bounded by some constant $C_T$.

The next theorem gives the existence of an optimal investment strategy for the original risk sensitive portfolio optimization problem.
\begin{theorem}\label{thm:pisatr}
Let the assumption {\rm({\bf H})} hold and let $(\tilde{Y},\tilde{Z},\tilde{V})\in{\cal S}_t^{\infty}\times{\Hxx}_{t,{\rm BMO}}^2\times L_t^2$ be a solution of BSDE~\eqref{eq:regulizeBSDE} in Theorem~\ref{thm:limitsolutionBSDE}. Define that
\begin{align}\label{defpi*}
\pi^*(u):=\argmax_{\pi\in U}h(\pi;p^{\rm M}(u-),H(u-),\tilde{Z}(u),\tilde{V}(u)),\quad u\in[t,T],
\end{align}
where the function $h(\pi;p,z,\xi,v)$ is given by~\eqref{eq:defh}. Then, we have $\pi^*\in\Uadt$ and $\pi^*$ is an optimal investment strategy for the risk sensitive control problem \eqref{eq:value}.
\end{theorem}

\begin{proof}
The main body of the proof is to show that the first assertion $\pi^*\in\Uadt$ holds. According to Definition \ref{def:addmisibleUad}, it remains to verify that  $({\cal E}(\Lambda^{\pi^*,t})_u)_{u\in[t,T]}$ is a true $(\Px^*,\FxM)$-martingale. In view of \eqref{defpi*}, it clearly holds that
\begin{align*}
h(\pi^*(u);p^{\rm M}(u-),H(u-),\tilde{Z}(u),\tilde{V}(u))\geq h(0;p^{\rm M}(u-),H(u-),\tilde{Z}(u),\tilde{V}(u)),~u\in[0,T].
\end{align*}
Similar to the proof of Lemma~\ref{lem:hNRN}, we can manipulate the r.h.s of the above inequality and attain the existence of constants $R_1,R_2>0$ depending on the essential upper bound of $\tilde{V}$ such that
\begin{align}\label{eq:optimalpicontrol}
\left|\pi^*(u)\right|^2\leq R_1|(1-H(u-))\tilde{Z}(u)|^2+R_2,\quad u\in[t,T].
\end{align}
For $u\in[t,T]$, let us define
\begin{align}\label{eq:Lambda1}
\Lambda_1^{\pi^*,t}(u):=\sum_{i=1}^n\int_t^u\left\{\sigma^{-1}_i(\mu_i^{\rm M}(s)+\lambda^{\rm M}_i(s))-\frac{\theta\sigma_i}2\pi_i^*(s)+\tilde{Z}_i(s)\right\}dW_i^{o,\tau}(s).
\end{align}
Thanks to the fact that $\tilde{Z}\in{\Hxx}_{t,{\rm BMO}}^2$ and \eqref{eq:optimalpicontrol}, it follows that $\Lambda_1^{\pi^*,t}=(\Lambda_1^{\pi^*,t}(u))_{u\in[t,T]}$ is a continuous BMO $(\Px^*,\FxM)$-martingale. By Theorem 3.4 in Kazamaki~\cite{Kazamaki1994}, there exists $\rho>1$ such that
\begin{align}\label{eq:BMOUI}
\Ex_{t,p,z}^*\left[{\cal E}(\Lambda_1^{\pi^*,t})_T^\rho\right]<+\infty.
\end{align}
On the other hand, the first-order condition gives that, for $i=1,\ldots,n$,
\begin{align}\label{eq:firstordercondition}
&\mu_i^{\rm M}(u-)+\lambda^{\rm M}_i(u-)-r+\sigma_i(1-H_i(u-))\tilde{Z}_i(u)\notag\\
&\qquad\quad=\left(1+\frac{\theta}2\right)\sigma^2_i{\pi}^*_i(u)+\lambda^{\rm M}_i(u-)(1-\pi^*_i(u))^{-\frac\theta2-1}e^{\tilde{V}_i(u)}.
\end{align}

We next prove the existence of constants $R_3,R_4>0$ depending on the essential upper bound of $\tilde{V}$ such that, for $i=1,\ldots,n$,
\begin{align}\label{eq:optimalpicontrol1}
\lambda^{\rm M}_i(u-)(1-\pi^*_i(u))^{-\frac\theta2-1}e^{\tilde{V}_i(u)}\leq R_3\left|(1-H_i(u-))\tilde{Z}_i(u)\right|+R_4.
\end{align}
In fact, for $i=1,\ldots,n$, if $\pi^*_i(u)\leq0$, the l.h.s. of \eqref{eq:optimalpicontrol1} is bounded by the constant $R_{\lambda}e^{|\tilde{V}_i|_{t,\infty}}$, where the positive constant $R_{\lambda}:=\max_{(i,k,z)\in\{1,\ldots,n\}\times S_I\times S_H}\lambda_i(k,z)$ is finite thanks to the assumption ({\bf H}). If $\pi^*_i(u)\in(0,1)$, it follows from \eqref{eq:firstordercondition} that
\begin{align*}
\lambda^{\rm M}_i(u-)(1-\pi^*_i(u))^{-\frac\theta2-1}e^{\tilde{V}_i(u)}&\leq\left(1+\frac{\theta}2\right)\sigma^2_i{\pi}^*_i(u)+\lambda^{\rm M}_i(u-)(1-\pi^*_i(u))^{-\frac\theta2-1}e^{\tilde{V}_i(u)}\notag\\
&=\mu_i^{\rm M}(u-)+\lambda^{\rm M}_i(u-)-r+\sigma_i(1-H_i(u-))\tilde{Z}_i(u).
\end{align*}
This shows \eqref{eq:optimalpicontrol1} again by the assumption ({\bf H}).

To continue, the estimate \eqref{eq:optimalpicontrol1} in turn entails the existence of constants $R_5,R_6>0$ such that, for $i=1,\ldots,n$,
\begin{align}\label{eq:optimalpicontrol2}
\left|\lambda^{\rm M}_i(u-)\right|^2(1-\pi^*_i(u))^{-\theta}e^{2\tilde{V}_i(u)}\leq R_5(1-H_i(u-))|\tilde{Z}_i(u)|^2+R_6.
\end{align}
For $u\in[t,T]$, we define
\begin{align}\label{eq:Lambda2i}
\Lambda_2^{\pi^*,t}(u):=\sum_{i=1}^n\Lambda_{2,i}^{\pi^*,t}(u):=\sum_{i=1}^n\int_t^u\{(1-\pi_i^*(s))^{-\frac\theta2}\lambda^{\rm M}_i(s-)e^{\tilde{V}_i(s)}-1\}d\Upsilon_i^*(s).
\end{align}
Moreover, we also define a probability measure $\Px^{(0)}\sim\Px^*$ via $\frac{d\Px^{(0)}}{d\Px^*}|_{\F_T^{\rm M}}={\cal E}(\Lambda_1^{\pi^*,0})_T$. Then, for $i=1,\ldots,n$, $H_i$ admits the $\Px^{(0)}$-intensity given by $1$. It holds that
\begin{align}
{\cal E}(\Lambda_{2,1}^{\pi^*,t})_u&=\exp\left(\int_t^u\{1-(1-\pi_1^*(s))^{-\frac\theta2}\lambda^{\rm M}_1(s)e^{\tilde{V}_1(s)}\}ds\right)\prod_{s\leq u}(1+\Delta\Lambda^{\pi^*,t}_{2,1}(s))\notag\\
&\leq e^{T-t}\left\{1+\int_t^T(1-\pi_1^*(s))^{-\frac\theta2}\lambda^{\rm M}_1(s-)e^{\tilde{V}_1(s)}dH_1(s)\right\},\quad u\in[t,T].
\end{align}
Let $R_T>0$ be a constant depending on $T$ that may refer to different values from line to line. Then, it follows from \eqref{eq:BMOUI} and \eqref{eq:optimalpicontrol2} that, for $(t,p,z)\in[0,T]\times S_{p^{\rm M}}\times S_H$,
\begin{align}\label{eq:Lambda212ndmoment}
\Ex^{(0)}_{t,p,z}\left[{\cal E}(\Lambda_{2,1}^{\pi^*,t})_u^2\right]&\leq R_T\Ex^{(0)}_{t,p,z}\left[1+\int_t^T(1-\pi^*_1(s))^{-\theta}\left|\lambda^{\rm M}_1(u-)\right|^2e^{2\tilde{V}_1(s)}dH_1(s)\right]\notag\\
&\leq R_T\left\{1+\Ex^{*}_{t,p,z}\left[{\cal E}(\Lambda_1^{\pi^*,t})_T\int_t^{T\wedge\tau_1^t}|\tilde{Z}_1(u)|^2du\right]\right\}\notag\\
&\leq R_T\left\{\Ex^*_{t,p,z}\left[{\cal E}(\Lambda_1^{\pi^*,t})_T^\rho\right]\right\}^{\frac1\rho}\left\{\Ex^*_{t,p,z}\left[\left(\int_t^{T\wedge\tau^1_t}|\tilde{Z}_1(u)|^2du\right)^{q}\right]\right\}^{\frac1q}+R_T\notag\\
&\leq R_T,
\end{align}
where $q>1$ satisfies that $\frac1\rho+\frac1{q}=1$, and we have used Corollary 2.1 in \cite{Kazamaki1994} for BMO $(\Px^*,\FxM)$-martingales in the last inequality. This yields that $({\cal E}(\Lambda_{2,1}^{\pi^*,t})_u)_{u\in[t,T]}$ is uniformly integrable (U.I.) under $\Px^{(0)}$. By using the orthogonality of $\Px^*$-martingales $\Lambda_1^{\pi^*,t}$ and $\Lambda_{2,1}^{\pi^*,t}$, it holds that
\begin{align}\label{eq:Ex0}
\Ex^{(0)}_{t,p,z}\left[{\cal E}(\Lambda_{2,1}^{{\pi^*},t})_T\right]=\Ex^{*}_{t,p,z}\left[{\cal E}(\Lambda_1^{\pi^*,t})_T{\cal E}(\Lambda_{2,1}^{{\pi^*},t})_T\right]=1.
\end{align}

We next define a probability measure $\Px^{(1)}\sim\Px^*$ via $\frac{d\Px^{(1)}}{d\Px^*}|_{\F_T^{\rm M}}={\cal E}(\Lambda_1^{{\pi^*};t})_T{\cal E}(\Lambda_{2,1}^{{\pi^*};t})_T$. Note that $H_1$ and $H_2$ do not jump simultaneously. Then, $H_2$ admits the unit intensity under $\Px^{(1)}$. Therefore, in the light of \eqref{eq:optimalpicontrol2} and \eqref{eq:Lambda212ndmoment}, we can derive that
\begin{align}\label{eq:Ex1Lambda22}
\Ex^{(1)}_{t,p,z}\left[{\cal E}(\Lambda_{2,2}^{\pi^*,t})_u^2\right]&\leq R_T\Ex^{(1)}_{t,p,z}\left[1+\int_t^T(1-\pi^*_2(s))^{-\theta}\left|\lambda^{\rm M}_2(u-)\right|^2e^{2\tilde{V}_2(s)}dH_2(s)\right]\notag\\
&\leq R_T\left\{1+\Ex^{(0)}_{t,p,z}\left[{\cal E}(\Lambda_{2,1}^{\pi^*,t})_T\int_t^{T\wedge\tau_2^t}|\tilde{Z}_2(u)|^2du\right]\right\}\notag\\
&\leq R_T\left\{\Ex^{(0)}_{t,p,z}\left[{\cal E}(\Lambda_{2,1}^{\pi^*,t})_T^2\right]\right\}^{\frac12}\left\{\Ex^{(0)}_{t,p,z}\left[\left(\int_t^{T\wedge\tau_2^t}|\tilde{Z}_2(u)|^2du\right)^{2}\right]\right\}^{\frac12}+R_T\notag\\
&\leq R_T\left\{\Ex^{(0)}_{t,p,z}\left[\left(\int_t^{T\wedge\tau_2^t}|\tilde{Z}_2(u)|^2du\right)^{2}\right]\right\}^{\frac12}+R_T.
\end{align}
The term $\Ex^{(0)}_{t,p,z}[(\int_t^{T\wedge\tau_2^t}|\tilde{Z}_2(u)|^2du)^{2}]$ can be estimated by
\begin{align*}
&\Ex^{(0)}_{t,p,z}\left[\left(\int_t^{T\wedge\tau_2^t}\left|\tilde{Z}_2(u)\right|^2du\right)^2\right]\nonumber\\
&\quad\leq\left\{\Ex^*_{t,p,z}\left[{\cal E}(\Lambda_1^{\pi^*,t})_T^\rho\right]\right\}^{\frac1\rho}\left\{\Ex^*_{t,p,z}\left[\left(\int_t^{T\wedge\tau_2^t}\left|\tilde{Z}_2(u)\right|^2du\right)^{2q}\right]\right\}^{\frac1{q}}.
\end{align*}
Thus, there exists a constant $R_T^{(1)}>0$ depending on $T$ such that, for all $u\in[t,T]$,
\begin{align}\label{eq:Ex1Lambda22square}
\Ex^{(1)}_{t,p,z}\left[{\cal E}(\Lambda_{2,2}^{\pi^*,t})_u^2\right]=\Ex^{*}_{t,p,z}\left[{\cal E}(\Lambda_1^{{\pi^*};t})_u{\cal E}(\Lambda_{2,1}^{{\pi^*};t})_u{\cal E}(\Lambda_{2,2}^{\pi^*,t})_u^2\right]\leq R_T^{(1)}.
\end{align}
Up to now, we have proved the following estimate with $l=2$: there exists a constant $R_T^{(l-1)}>0$ depending on $T$ such that, for all $u\in[t,T]$,
\begin{align}\label{eq:estimatel}
\Ex^*_{t,p,z}\left[{\cal E}(\Lambda_1^{\pi^*,t})_u{\cal E}\left(\sum_{i=1}^{l-1}\Lambda_{2,i}^{\pi^*,t}\right)_u{\cal E}(\Lambda_{2,l}^{\pi^*,t})_u^2\right]\leq R_T^{(l-1)}.
\end{align}

We next verify \eqref{eq:estimatel} for all $l\leq n$ using the mathematical induction argument. To this end, suppose \eqref{eq:estimatel} holds for all $l\leq k$ (where $2\leq k\leq n$). The goal is to validate \eqref{eq:estimatel} for $l=k+1$. First, following similar lines of argument to prove \eqref{eq:Ex0}, we can obtain inductively that, for all $2\leq l\leq k$,
\begin{align}\label{eq:Exlk}
\Ex^*_{t,p,z}\left[{\cal E}(\Lambda_1^{\pi^*,t})_T\prod_{i=1}^l{\cal E}(\Lambda_{2,i}^{\pi^*,t})_T\right]=1.
\end{align}
Let us define a probability measure $\Px^{(l)}\sim\Px^*$ by
\begin{align}\label{eq:Pl}
\frac{d\Px^{(l)}}{d\Px^*}\Big|_{\F_T^{\rm M}}&:={\cal E}(\Lambda_1^{\pi^*,t})_T\prod_{i=1}^l{\cal E}(\Lambda_{2,i}^{\pi^*,t})_T,\ \ \text{for $2\leq l\leq k$}.
\end{align}
Note again that $H_1,\ldots,H_k,H_{k+1}$ do not jump simultaneously and hence $H_{k+1}$ admits the unit intensity under $\Px^{(k)}$. By virtue of \eqref{eq:optimalpicontrol2} and \eqref{eq:estimatel} with $l\leq k$, we can further deduce that
\begin{align}\label{eq:Ex1Lambdak+1}
&\Ex^{(k)}_{t,p,z}\left[{\cal E}(\Lambda_{2,k+1}^{\pi^*,t})_u^2\right]
\leq R_T\left\{1+\Ex^{(k-1)}_{t,p,z}\left[{\cal E}(\Lambda_{2,k}^{\pi^*,t})_T\int_t^{T\wedge\tau_{k+1}^t}|\tilde{Z}_{k+1}(u)|^2du\right]\right\}\notag\\
&\leq R_T\left\{\Ex^{(k-1)}_{t,p,z}\left[{\cal E}(\Lambda_{2,k}^{\pi^*,t})_T^2\right]\right\}^{\frac12}\left\{\Ex^{(k-1)}_{t,p,z}\left[\left(\int_t^{T\wedge\tau_{k+1}^t}|\tilde{Z}_{k+1}(u)|^2du\right)^{2}\right]\right\}^{\frac12}+R_T\notag\\
&\leq R_T\left\{\Ex^{(k-1)}_{t,p,z}\left[\left(\int_t^{T\wedge\tau_{k+1}^t}|\tilde{Z}_{k+1}(u)|^2du\right)^{2}\right]\right\}^{\frac12}+R_T\nonumber\\
&=R_T\left\{\Ex^{(k-2)}_{t,p,z}\left[{\cal E}(\Lambda_{2,k-1}^{\pi^*,t})_T\left(\int_t^{T\wedge\tau_{k+1}^t}|\tilde{Z}_{k+1}(u)|^2du\right)^{2}\right]\right\}^{\frac12}+R_T\nonumber\\
&\leq R_T\left\{\Ex^{(k-2)}_{t,p,z}\left[\left(\int_t^{T\wedge\tau_{k+1}^t}|\tilde{Z}_{k+1}(u)|^2du\right)^{2^2}\right]\right\}^{\frac1{2^2}}+R_T\nonumber\\
&\quad\cdots\cdots\cdots\nonumber\\
&\leq R_T\left\{\Ex^{(0)}_{t,p,z}\left[\left(\int_t^{T\wedge\tau_{k+1}^t}|\tilde{Z}_{k+1}(u)|^2du\right)^{2^k}\right]\right\}^{\frac1{2^k}}+R_T\\
&\leq R_T\left\{\Ex_{t,p,z}^*\left[{\cal E}(\Lambda_1^{\pi^*,t})_T^{\rho}\right]\right\}^{\frac{1}{\rho2^k}}\left\{\Ex^{*}_{t,p,z}\left[\left(\int_t^{T\wedge\tau_{k+1}^t}|\tilde{Z}_{k+1}(u)|^2du\right)^{q2^k}\right]\right\}^{\frac1{q2^k}}+R_T\nonumber\\
&\leq R_T.\nonumber
\end{align}
This confirms the estimate \eqref{eq:estimatel} with $l=k+1$. As a result of the previous induction and the orthogonality of $\Lambda_1^{\pi^*,t}$, $\Lambda_{2,1}^{\pi^*,t},\ldots,\Lambda_{2,n}^{\pi^*,t}$, we have
\begin{align}\label{eq:UI1n}
\Ex^{*}_{t,p,z}\left[{\cal E}(\Lambda^{\pi^*,t})_T\right]=\Ex^*_{t,p,z}\left[{\cal E}(\Lambda_1^{\pi^*,t})_T\prod_{i=1}^n{\cal E}(\Lambda_{2,i}^{\pi^*,t})_T\right]=1.
\end{align}
This shows that $({\cal E}(\Lambda^{\pi^*,t})_u)_{u\in[t,T]}$ is a U.I. $(\Px^*,\FxM)$-martingale, which verifies the first assertion that $\pi^*\in\Uadt$.

Next, the first-order condition in the definition of $\pi^*$ and Theorem~\ref{thm:limitsolutionBSDE} can entail that \eqref{eq:optimal} in Lemma~\ref{lem:BSDE-valuefcn} holds valid. We can readily conclude the second assertion that $\pi^*$ is indeed an optimal strategy using Lemma~\ref{lem:BSDE-valuefcn}.
\end{proof}

It is worth noting that Theorem \ref{thm:limitsolutionBSDE} only gives the existence of a solution $(\tilde{Y},\tilde{Z},\tilde{V})\in{\cal S}_t^{\infty}\times{\Hxx}_{t,{\rm BMO}}^2\times L_t^2$ to BSDE~\eqref{eq:regulizeBSDE} while the uniqueness of the solution remains open. The next result finally confirms that our constructed solution in Theorem \ref{thm:limitsolutionBSDE} is unique that is a consequence of Lemma~\ref{lem:BSDE-valuefcn} and Theorem \ref{thm:pisatr}, which in turn implies that $\pi^*$ constructed in \eqref{defpi*} is the unique optimal portfolio.

\begin{proposition}\label{uniquecorr}
The limiting process $(\tilde{Y},\tilde{Z},\tilde{V})$ in Lemma~\ref{lem:converRes} is the unique (in the sense of $d\Px^*\otimes du$-a.e.) solution of BSDE~\eqref{eq:regulizeBSDE} in the space ${\cal S}_t^{\infty}\times{\Hxx}_{t,{\rm BMO}}^2\times L_t^2$. Moreover, the portfolio process $\pi^*$ defined in \eqref{defpi*} by $(\tilde{Y},\tilde{Z},\tilde{V})$ is the unique (in the sense of $d\Px^*\otimes du$-a.e.) optimal investment strategy for the risk-sensitive control problem \eqref{eq:value}.
\end{proposition}

\begin{proof}
In Theorem \ref{thm:limitsolutionBSDE}, we proved that there exists one solution $(\tilde{Y},\tilde{Z},\tilde{V})\in{\cal S}_t^{\infty}\times{\Hxx}_{t,{\rm BMO}}^2\times L_t^2$ to BSDE \eqref{eq:regulizeBSDE} such that $(\tilde{Y}+\int_t^\cdot f(p^{\rm M}(s),H(s),0,0)ds,\tilde{Z},\tilde{V})$ solves the original BSDE \eqref{eq:formalBSDE}. Recall $U=(-\infty,1)^n$, and we next define the set, for $t\in[0,T]$,
\begin{align*}
\hat{\cal U}_t^{ad}:=\Bigg\{&\pi=(\pi_i(u);~i=1,\ldots,n)_{u\in[t,T]}^{\top}\in U;~ \pi\ \text{is}\ \Fx^{\rm M}\text{-predictable such that both}\\
&\sum_{i=1}^n\int_t^u\pi_i(s)dW^{o,\tau}_i(s)\ \text{and}\ \sum_{i=1}^n\int_t^u(1-\pi_i(s))^{-\frac\theta2}dW^{o,\tau}_i(s),\ u\in[t,T],\\
&\text{are}\ (\Px^*,\Fx^{\rm M})\text{-BMO martingales}\Bigg\}.
\end{align*}
Let $(\tilde{Y},\tilde{Z},\tilde{V})\in{\cal S}_t^{\infty}\times{\Hxx}_{t,{\rm BMO}}^2\times L_t^2$ be a solution of BSDE \eqref{eq:regulizeBSDE} and let $\pi^*=(\pi^*(u))_{u\in[t,T]}$ be defined by \eqref{defpi*} using $(\tilde{Z},\tilde{V})$ from this solution. Then, it follows from \eqref{eq:optimalpicontrol}, \eqref{eq:optimalpicontrol2} and $\tilde{Z}\in{\Hxx}_{t,{\rm BMO}}^2$ that $\pi^*\in\hat{\cal U}_t^{ad}$. Now, for any $\pi\in\hat{\cal U}_t^{ad}$, let us define, for $i=1,\ldots,n$,
\[
\hat{Z}_i(u):=|\pi_i(u)|+(1-\pi_i(u))^{-\frac\theta2},\quad u\in[t,T].
\]
Then $\hat{Z}=(\hat{Z}_i(u);~i=1,\ldots,n)_{u\in[t,T]}^{\top}\in{\Hxx}_{t,{\rm BMO}}^2$, and we can obtain the same estimates  \eqref{eq:optimalpicontrol} and \eqref{eq:optimalpicontrol2}  with $(\pi^*,\tilde{Z})$ replaced by $(\pi,\hat{Z})$. Moreover, by applying a similar induction to prove \eqref{eq:UI1n}, we deduce that $\hat{\cal U}_t^{ad}\subset{\cal U}^{ad}_t$. This implies that $\pi^*$ constructed by $(\tilde{Z},\tilde{V})$ satisfies that
\begin{align}\label{eq:smallerverfi}
\inf_{\pi\in\hat{\cal U}_t^{ad}}J(\pi;t,p,z)=e^{Y(t;t,p,z)}=J(\pi^*;t,p,z),
\end{align}
where $J(\pi;t,p,z)$ is given by \eqref{eq:Qtu0} and $Y(t;t,p,z)=\tilde{Y}(t)$ as we have $Y:=\tilde{Y}+\int_t^\cdot f(p^{\rm M}(s),H(s),0,0)ds$ in the proof of Lemma \ref{lem:BSDE-valuefcn}. That is, we have constructed an admissible control subset $\hat{\cal U}_t^{ad}\subset{\cal U}_t^{ad}$ independent of $(\tilde{Y},\tilde{Z},\tilde{V})$ such that the optimal strategy $\pi^*$ given by \eqref{defpi*} is still in $\hat{\cal U}_t^{ad}$.

We next apply this subset $\hat{\cal U}_t^{ad}$ to conclude the uniqueness of solutions to BSDE \eqref{eq:regulizeBSDE}. To this end, let $(\tilde{Y}^{i},\tilde{Z}^{i},\tilde{V}^{i})\in{\cal S}_t^{\infty}\times{\Hxx}_{t,{\rm BMO}}^2\times L_t^2$, $i=1,2$ be two solutions of BSDE \eqref{eq:regulizeBSDE} with the same terminal condition.  We can then define $\pi^{i,*}\in\hat{\cal U}_t^{ad}$ as in \eqref{defpi*} by using $(\tilde{Y}^{i},\tilde{Z}^{i},\tilde{V}^{i})$ respectively for $i=1,2$. The verification of optimality in Lemma \ref{lem:BSDE-valuefcn}, together with \eqref{eq:smallerverfi}, yields that
\begin{align*}
  e^{\tilde{Y}^{1}(t)}=e^{\tilde{Y}^{2}(t)}=\inf_{\pi\in\hat{\mathcal{U}}^{ad}_t}J(\pi;t,p,z).
\end{align*}
This implies that
\begin{align*}
J(\pi^{1,*};t,p,z)e^{-\tilde{Y}^{2}(t)}&=\Ex^*_{t,p,z}\Bigg[{\cal E}\left(\Lambda^{\pi^{1,*},t}\right)_T\exp\Bigg(\int_t^T\Big(f\left(p^{\rm M}(u-),H(u-),\tilde{Z}^{2}(u),\tilde{V}^{2}(u)\right)\notag\\
&\quad-h\left(\pi^{1,*}(u);p^{\rm M}(u-),H(u-),\tilde{Z}^{2}(u),\tilde{V}^{2}(u)\right)\Big)du\Bigg)\Bigg]=1,
\end{align*}
where $\Lambda^{\pi,t}=(\Lambda^{\pi,t}(u))_{u\in[t,T]}$ for $\pi\in{\cal U}_t^{ad}$ is defined by \eqref{eq:Lambdapit}. Therefore, it holds that, $d\Px^*\otimes du$-a.e.
\begin{align*}
f\left(p^{\rm M}(u-),H(u-),\tilde{Z}^{2}(u),\tilde{V}^{2}(u)\right)=h\left(\pi^{1,*}(u);p^{\rm M}(u-),H(u-),\tilde{Z}^{2}(u),\tilde{V}^{2}(u)\right).
\end{align*}
Let $J(\pi;u):=\Ex\left[\left(\frac{X^{\pi}(T)}{X^{\pi}(u)}\right)^{-\frac{\theta}{2}}\Big|\F_u^{\rm M}\right]$ for $u\in[t,T]$. Then, for $u\in[t,T]$, we have that
\begin{align}\label{eq:verfiY1}
&J(\pi^{1,*};u)e^{-\tilde{Y}^{2}(u)+\int_t^uf(p^{\rm M}(s-),H(s-),0,0)ds}\notag\\
&\qquad=\Ex^*\Bigg[{\cal E}\left(\Lambda^{\pi^{1,*},u}\right)_T\exp\bigg(\int_u^T\Big(f\left(p^{\rm M}(s-),H(s-),\tilde{Z}^{2}(s),\tilde{V}^{2}(s)\right)\notag\\
&\qquad\quad-h\left(\pi^{1,*}(s);p^{\rm M}(s-),H(s-),\tilde{Z}^{2}(s),\tilde{V}^{2}(s)\right)\Big)ds\bigg)\Big|{\cal F}^{\rm M}_u\Bigg]=1.
\end{align}
On the other hand, by Lemma \ref{lem:BSDE-valuefcn}, we have that, for $u\in[t,T]$,
\begin{align}\label{eq:verfiY2}
&J(\pi^{1,*};u)e^{-\tilde{Y}^{1}(u)+\int_t^uf(p^{\rm M}(s-),H(s-),0,0)ds}\nonumber\\
&\qquad=\Ex^*\Bigg[{\cal E}\left(\Lambda^{\pi^{1,*},u}\right)_T\exp\bigg(\int_u^T\Big(f\left(p^{\rm M}(s-),H(s-),\tilde{Z}^{1}(s),\tilde{V}^{1}(s)\right)\notag\\
&\qquad\quad-h\left(\pi^{1,*}(s);p^{\rm M}(s-),H(s-),\tilde{Z}^{1}(s),\tilde{V}^{1}(s)\right)\Big)ds\bigg)\Big|{\cal F}^{\rm M}_u\Bigg]=1.
\end{align}
It follows from \eqref{eq:verfiY1} and \eqref{eq:verfiY2} that, for $u\in[t,T]$, $\tilde{Y}^{(1)}(u)=\tilde{Y}^{(2)}(u)$, $\mathbb{P}^*$-a.e..
Note that $(\tilde{Y}^{i},\tilde{Z}^{i},\tilde{V}^{i})\in{\cal S}_t^{\infty}\times{\Hxx}_{t,{\rm BMO}}^2\times L_t^2$, $i=1,2$ satisfy BSDE \eqref{eq:regulizeBSDE}. Together with Theorem~\ref{thm:hatFrep}, the unique canonical decomposition of the semimartingale $\tilde{Y}=(\tilde{Y}(u))_{u\in[t,T]}\in{\cal S}_t^{\infty}$ under $\mathbb{P}^*$ (see Theorem 34 in Chapter III of \cite{Protter2005}) implies that, for $u\in[t,T]$, $\mathbb{P}^*$-a.e.,
\begin{align*}
\int_t^u\tilde{Z}^{1}(s)^\top dW^{o,\tau}(s)=\int_t^u\tilde{Z}^{2}(s)^\top dW^{o,\tau}(s),~ \int_t^u{\tilde{V}^{1}}(s)^\top d\Upsilon^*(s)=\int_t^u{\tilde{V}^{2}}(s)^\top d\Upsilon^*(s),
\end{align*}
which proves the uniqueness of the solution to BSDE \eqref{eq:regulizeBSDE} in the sense of $d\Px^*\otimes du$-a.e..

For the unique solution $(\tilde{Y},\tilde{Z},\tilde{V})\in{\cal S}_t^{\infty}\times{\Hxx}_{t,{\rm BMO}}^2\times L_t^2$ of BSDE \eqref{eq:regulizeBSDE}, we then claim that the constructed strategy $\pi^*$ in \eqref{defpi*} is the unique optimal portfolio for the original control problem. In fact, for an arbitrary optimal strategy $\hat{\pi}\in\mathcal{\cal U}^{ad}_t$, from the proof of Lemma \ref{lem:BSDE-valuefcn}, we can see that
\begin{align*}
J(\hat{\pi};t,p,z)e^{-\tilde{Y}(t)}&=\Ex^*_{t,p,z}\Bigg[{\cal E}\left(\Lambda^{\hat{\pi},t}\right)_T\exp\bigg(\int_t^T\Big(f\left(p^{\rm M}(u-),H(u-),\tilde{Z}(u),\tilde{V}(u)\right)\notag\\
&\quad-h\left(\hat{\pi}(u);p^{\rm M}(u-),H(u-),\tilde{Z}(u),\tilde{V}(u)\right)\Big)du\bigg)\Bigg]=1.
\end{align*}
Therefore, $d\Px^*\otimes du$-a.e.
\begin{align*}
h\left(\hat{\pi}(u);p^{\rm M}(u-),H(u-),\tilde{Z}(u),\tilde{V}(u)\right)&=f\left(p^{\rm M}(u-),H(u-),\tilde{Z}(u),\tilde{V}(u)\right)\notag\\
  &=\max_{\pi\in U}h\left(\pi;p^{\rm M}(u-),H(u-),\tilde{Z}(u),\tilde{V}(u)\right).
\end{align*}
It then follows from the strict convexity of $U\ni\pi\to h(\pi;p,z,\xi,v)$ that $\hat{\pi}=\pi^*$, $d\Px^*\otimes du$-a.e.. This verifies the uniqueness of the admissible optimal strategy $\pi^*$, which completes the whole proof.
\end{proof}

\begin{appendix}
\section{Proofs of Some Auxiliary Results}\label{app:proof1}
\renewcommand\theequation{A.\arabic{equation}}
\setcounter{equation}{0}

This section collects the technical proofs of some auxiliary results that have been used in previous sections of the paper.

\begin{proof}[Proof of Proposition~\ref{prop:dynamicspM}]
For $t\in[0,T]$, let us define $\zeta_k(t):={\bf1}_{\{I(t)=k\}}$ for $k\in S_I$. It is clear that $J_k(t):=\zeta_k(t)-\zeta_k(0)-\int_0^t\sum_{i\in S_I}q_{ik}\zeta_i(s)ds$, $t\in[0,T]$, is a $(\Px,\Fx)$-martingale with bounded jumps. Taking the $\Px$-conditional expectation under $\FM_t$ on both sides, we obtain that
$J^{\rm M}_k(t)=p_k^{\rm M}(t)-p_k^{\rm M}(0)-\sum_{i\in S_I}\int_0^tq_{ik}p_{i}^{\rm M}(s)ds$ for $t\in[0,T]$ is a square-integrable $(\Px,\FxM)$-martingale with bounded jumps. Theorem~\ref{thm:hatFrep} gives the existence of $\FxM$-predictable processes $\alpha^{\rm M}=(\alpha_1^{\rm M}(t),\ldots,\alpha_n^{\rm M}(t))_{t\in[0,T]}^{\top}$ and $\beta^{\rm M}=(\beta_1^{\rm M}(t),\ldots,\beta_n^{\rm M}(t))_{t\in[0,T]}^{\top}$ such that, for $t\in[0,T]$,
\begin{align*}
J_k^{\rm M}(t)&=J_k^{\rm M}(0)+\sum_{i=1}^n\int_0^t\alpha_i^{\rm M}(s)dW_i^{\rm M}(s)+\sum_{i=1}^n\int_0^t\beta_i^{\rm M}(s)d\Upsilon_i^{\rm M}(s),
\end{align*}
and hence
\begin{align}\label{pkrep2}
p_k^{\rm M}(t)&=p_k^{\rm M}(0)+\sum_{j\in S_I}\int_0^tq_{jk}p_{j}^{\rm M}(s)ds+\sum_{i=1}^n\int_0^t\alpha_i^{\rm M}(s)dW_i^{\rm M}(s)\nonumber\\
&\quad+\sum_{i=1}^n\int_0^t\beta_i^{\rm M}(s)d\Upsilon_i^{\rm M}(s).
\end{align}

We next identify $\alpha^{\rm M}$ and $\beta^{\rm M}$ by taking $W^{o,\tau}$ defined by \eqref{eq:tildeWi} as a test process. By \eqref{eq:WM-M}, we have that $W_i^{o,\tau}(t)=W_i^{\rm M}(t)+\sigma_i^{-1}\int_0^{t\wedge\tau_i}(\mu_i^{\rm M}(p^{\rm M}(s))+\lambda_i^{\rm M}(p^{\rm M}(s),H(s)))ds$ for $t\in[0,T]$ which is $\FxM$-adapted. Then, for $i=1,\ldots,n$, it holds that
\begin{align}\label{eq:identyab}
(\zeta_k(t)W_i^{o,\tau}(t))^{\rm M}=\zeta_k^{\rm M}(t)W_i^{o,\tau}(t)=p_k^{\rm M}(t)W_i^{o,\tau}(t),\quad k\in S_I.
\end{align}
Note that $J_k$ is a semimartingale of pure jumps while $W_i^{o,\tau}$ is continuous. It is clear that $[\zeta_k,W_i^{o,\tau}]=[J_k,W_i^{o,\tau}]\equiv0$. Using integration by parts, we arrive at
\begin{align}\label{eq:condi1}
\zeta_k(t)W_i^{o,\tau}(t)
&=\int_0^tW_i^{o,\tau}(s)\sum_{j\in S_I}q_{jk}\zeta_j(s)ds+\int_0^tW_i^{o,\tau}(s)dJ_k(s)+\int_0^{t\wedge\tau_i}\zeta_k(s)dW_i(s)\notag\\
&\quad+\sigma_i^{-1}\int_0^{t\wedge\tau_i}(\mu_i(k)+\lambda_i(k,H(s)))\zeta_k(s)ds.
\end{align}
Note that both $W_i^{o,\tau}$ and $J_k$ are square-integrable semimartingales under $\Px$. Then, the second and the third terms on r.h.s. of \eqref{eq:condi1} are true $\Fx$-martingales. Taking the $\Px$-conditional expectation under $\FxM$ on both sides of \eqref{eq:condi1}, we can write the $\FxM$-semimartingale $(\zeta_kW_i^{o,\tau})^{\rm M}:=(\Ex\left[\zeta_k(t)W_i^{o,\tau}(t)|\F_t^{\rm M}\right])_{t\in[0,T]}$ by
\begin{align}\label{eq:zetaphiM}
&(\zeta_k(t)W_i^{o,\tau}(t))^{\rm M}=\Ex\left[\int_0^tW_i^{o,\tau}(s)dJ_k(s)+\int_0^{t\wedge\tau_i}\zeta_k(s)dW_i(s)\Big|\FM_t\right]\nonumber\\
&\qquad+\int_0^tW_i^{o,\tau}(s)\sum_{j\in S_I}q_{jk}p_j^{\rm M}(s)ds+\sigma_i^{-1}\int_0^{t\wedge\tau_i}(\mu_i(k)+\lambda_i(k,H(s)))p_k^{\rm M}(s)ds,
\end{align}
where the first term on the r.h.s. of \eqref{eq:zetaphiM} is a $(\Px,\FxM)$-martingale, and the rest terms are finite variation processes in the canonical decomposition of $(\zeta_kW_i^{o,\tau})^{\rm M}$. On the other hand, we also have that
\begin{align*}
p_k^{\rm M}(t)W_i^{o,\tau}(t)
&=\int_0^tW_i^{o,\tau}(s)\sum_{j\in S_I}q_{jk}p_j^{\rm M}(s)ds+\int_0^tW_i^{o,\tau}(s)dJ_k^{\rm M}(s)+\int_0^{t\wedge\tau_i}p_k^{\rm M}(s)dW_i^{\rm M}(s)\notag\\
&\quad+\sigma^{-1}_i\int_0^{t\wedge\tau_i}p_k^{\rm M}(s)(\mu_i^{\rm M}(p^{\rm M}(s))+\lambda_i^{\rm M}(p^{\rm M}(s),H(s)))ds+\int_0^{t\wedge\tau_i}\alpha^{\rm M}_i(s)ds,
\end{align*}
where the second and the third terms of the r.h.s. of the above equation are true $\FxM$-martingale due to the square integrability of $W_i^{o,\tau}$ and $p_k^{\rm M}$. By virtue of \eqref{eq:identyab}, we can compare the finite variation parts of $(\zeta_k(t)\phi_i(t))^{\rm M}$ and $p_k^{\rm M}(t)W_i^{o,\tau}(t)$ to obtain that, on $\{0<t\leq\tau_i\}$,
\begin{align*}
\alpha_i^{\rm M}(t)&=\sigma_i^{-1}p_k^{\rm M}(t)\left\{\mu_i(k)+\lambda_i(k)-\mu^{\rm M}_i(p^{\rm M}(t))-\lambda_i^{\rm M}(p^{\rm M}(t),H(t))\right\}\notag\\
&=\sigma_i^{-1}p_k^{\rm M}(t)\left\{(\mu_i(k)+\lambda_i(k,H(t))-\sum_{j\in S_I}\mu_i(j)p_j^{\rm M}(t)-\sum_{j\in S_I}\lambda_i(j,H(s))p_j^{\rm M}(t)\right\}.
\end{align*}

Finally, we replace the test process $W_i^{o,\tau}$ by the test process $H_i(t)$. Note that the Markov chain $I$ do not jump simultaneously with the default indicator process $H$. It holds that $[\zeta_k,H_i]\equiv0$. By applying a similar argument to identify $\alpha^{\rm M}$, one can show that, on $\{0<t\leq\tau_i\}$,
\begin{align*}
\beta_i^{\rm M}(t)&=\lambda^{\rm M}_i(p^{\rm M}(t-),H(t-))^{-1}p_k^{\rm M}(t-)\lambda_i(k,H(t-))-p_k^{\rm M}(t-)\nonumber\\
&=p_k^{\rm M}(t-)\left\{\frac{\lambda_i(k,H(t-))}{\sum_{l\in S_I}\lambda_i(l,H(t-))p_l^{\rm M}(t-)}-1\right\}.
\end{align*}
By substituting $(\alpha^{\rm M},\beta^{\rm M})$ in \eqref{pkrep2}, we arrive at the desired dynamics in \eqref{eq:dynamicspM000}.
\end{proof}

\begin{proof}[Proof of Lemma~\ref{lem:Xtpi2}]
We can see from \eqref{thetaX} that $\frac{X^{\pi}(T)}{X^{\pi}(t)}$ is $\FM_T$-measurable. A direct computation using \eqref{thetaX} and \eqref{eq:hatetaexpo} yields that
\begin{align*}
J(\pi;t,p,z)&=\Ex_{t,p,z}\left[\left(\frac{X^{\pi}(T)}{X^{\pi}(t)}\right)^{-\frac{\theta}{2}}\right]
=\Ex_{t,p,z}^{*}\left[\eta^{\rm M}(t,T)^{-1}\left(\frac{X^{\pi}(T)}{X^{\pi}(t)}\right)^{-\frac{\theta}{2}}\right]\nonumber\\
&=\Ex_{t,p,z}^{*}\left[e^{Q^{\pi,t}(T)}\right],
\end{align*}
which completes the proof.
\end{proof}

\begin{proof}[Proof of Lemma~\ref{lem:hNRN}] With the aid of \eqref{eq:fNhN} and the assumption {\rm({\bf H})}, we can see that, for $i=1,\ldots,n$,
\begin{align*}
h^N_i(0;p,z,\xi_i,v_i)
&\geq-\frac12\left|\xi_i\right|^2\rho_N(\xi_i){\bf1}_{\{|\xi_i|\leq N+2\}}-\lambda_i^{\rm M}(p,z)e^{v_i}\hat{\rho}_N(e^{v_i})\nonumber\\
&\geq-\left\{\frac{(N+2)^2}2+C(N+1)\right\}.
\end{align*}
On the other hand, for $\pi_i\in(-\infty,1)$,
\begin{align*}
h^N_i(\pi;p,z,\xi,v)&\leq-\frac\theta4\sigma_i^2\pi_i^2+\frac\theta2\left(\mu_i^{\rm M}(p)+\lambda^{\rm M}_i(p,z)-r\right)\pi_i+\lambda_i^{\rm M}(p,z)\notag\\
&\leq-\frac\theta4\sigma_i^2\pi_i^2+\frac\theta2(2C+r)|\pi_i|+C.
\end{align*}
For $i=1,\ldots,n$, we can take a constant $R_N>0$ only depending on $N$ such that, for all $\pi_i\in(-\infty,1)$ satisfying $|\pi_i|>R_N$, we have that
\begin{align*}
-\frac\theta4\sigma_i^2\pi_i^2+\frac\theta2(2C+r)|\pi_i|+C<-\left\{\frac{(N+2)^2}2+C(N+1)\right\}.
\end{align*}
Therefore, for all $\pi_i\in(-\infty,-R_N)$, it holds that $h^N_i(\pi_i;p,z,\xi_i,v_i)<h^N_i(0;p,z,\xi_i,v_i)$, which further implies that \eqref{eq:boundedsup} holds.
\end{proof}

\begin{proof}[Proof of Lemma~\ref{lem:boundedzeta}] By virtue of \eqref{eq:defh2}, we have that, for $(p,z,\pi_i)\in S_{p^{\rm M}}\times S_H\times(-\infty,1)$,
\begin{align*}
h_i(\pi_i;p,z,0,0)&=-\left(\frac\theta4+\frac{\theta^2}8\right)\sigma_i^2\pi_i^2+\frac\theta2(\mu_i^{\rm M}(p)+\lambda^{\rm M}_i(p,z)-r)\pi_i+\lambda_i^{\rm M}(p,z)\\
&\quad-\lambda_i^{\rm M}(p,z)(1-\pi_i)^{-\frac\theta2},\qquad i=1,\ldots,n.
\end{align*}
In light of the assumption ({\bf H}), we have that, for $i=1,\ldots,n$, $|\frac\theta2(\mu_i^{\rm M}(p)+\lambda^{\rm M}_i(p,z)-r)\pi_i|\leq\frac\theta4\{\pi_i^2+(2C+r)^2\}$. On the other hand, for $\pi_i\in(-\infty,1)$, we have that $R_2(\pi_i)\leq h_i(\pi_i;p,z,0,0)\leq R_1$,
where $R_1:=\frac\theta4(2C+r)^2+\frac{\theta}{4}+C$, and for $\pi_i\in(-\infty,1)$,
\begin{align*}
R_2(\pi_i):=-\left(\frac\theta4+\frac{\theta^2}8+\frac{\theta}{4\sigma_i^2}\right)\sigma^2_i\pi_i^2-C(1-\pi_i)^{-\frac\theta2}-\frac\theta4(2C+r)^2
+\varepsilon.
\end{align*}
Note that $R_3:=|\sup_{\pi_i\in(-\infty,1)}R_2(\pi_i)|<+\infty$. Then, for all $(p,z)\in S_{p^{\rm M}}\times S_H$,
\begin{align*}
\left|\sup_{\pi_i\in(-\infty,1)}h_i(\pi_i;p,z,0,0)\right|\leq R_1\vee R_3,\quad i=1,\ldots,n.
\end{align*}
Thanks to \eqref{eq:defh}, we deduce that $h_L(p,z,0,0)=\frac{r\theta}{2}$ for all $(p,z)\in S_{p^{\rm M}}\times S_H$. This verifies that $\zeta$ is a bounded r.v..
\end{proof}

\begin{proof}[Proof of Lemma~\ref{lem:monoYN}] For $u\in[t,T]$ and $i=1,\ldots,n$, we define that
\begin{align*}
\tilde{Z}^{N+1,N,i}(u)&:=(\tilde{Z}^{N+1}_1(u),\ldots,\tilde{Z}^{N+1}_i(u),\tilde{Z}^N_{i+1}(u),\ldots,\tilde{Z}^N_n(u)),\\ \tilde{V}^{N+1,N,i}(u)&:=(\tilde{V}^{N+1}_1(u),\ldots,\tilde{V}^{N+1}_i(u),\tilde{V}^N_{i+1}(u),\ldots,\tilde{V}^N_n(u)).
\end{align*}
Here, $\tilde{V}^N$ is the $\FxM$-predictable $\R^n$-valued bounded process satisfying \eqref{eq:estiBMO} in Lemma~\ref{lem:prioriestimate2}. We also set $\tilde{Z}^{N+1,N,0}(u)=\tilde{Z}^N(u)$, $\tilde{Z}^{N+1,N,n}(u)=\tilde{Z}^{N+1}(u)$, $\tilde{V}^{N+1,N,0}(u)=\tilde{V}^N(u)$ and $\tilde{V}^{N+1,N,n}(u)=\tilde{V}^{N+1}(u)$. For $i=1,\ldots,n$, let us define that
\begin{align*}
\gamma_i(u):=\frac{\tilde{f}^{N+1}(u,\tilde{Z}^{N+1,N,i}(u),\tilde{V}^{N+1}(u))-\tilde{f}^{N+1}(u,\tilde{Z}^{N+1,N,i-1}(u),\tilde{V}^{N+1}(u))}{\tilde{Z}^{N+1}_i(u)-\tilde{Z}^N_i(u)},
\end{align*}
if $(1-H_i(u-))\tilde{Z}^{N+1}_i(u)\neq(1-H_i(u-))\tilde{Z}^N_i(u)$, and it is $0$ otherwise. Let us also define
\begin{align*}
\eta_i(u):=\frac{\tilde{f}^{N+1}(u,\tilde{Z}^{N}(u),\tilde{V}^{N+1,N,i}(u))-\tilde{f}^{N+1}(u,\tilde{Z}^{N}(u),\tilde{V}^{N+1,N,i-1}(u))}{\tilde{V}^{N+1}_i(u)-\tilde{V}^N_i(u)},
\end{align*}
if $(1-H_i(u-))\tilde{V}^{N+1}_i(u)\neq(1-H_i(u-))\tilde{V}^N_i(u)$, and it is $0$ otherwise. Moreover, let us consider the probability measure $\Qx\sim\Px^*$ defined in \eqref{eq:Q} with $(\gamma_i(u),\eta_i(u))$ given above. By Lemma~\ref{lem:prioriestimate2}, for any $s\in[0,1]$ and $u\in[t,T]$, it holds that
\begin{align}\label{eq:VhatCT}
s\tilde{V}_i^{N+1}(u)+(1-s)\tilde{V}_i^N(u)\leq C_T,\quad\text{a.e.},
\end{align}
for some constant $C_T>0$ depending on $T>0$ only. By taking constant $N_0>C_T$, we have that, for all $N\geq N_0$,
\begin{align*}
&\frac{\tilde{f}^{N+1}(u,\tilde{Z}^{N}(u),\tilde{V}^{N+1,N,i}(u))-\tilde{f}^{N+1}(u,\tilde{Z}^{N}(u),\tilde{V}^{N+1,N,i-1}(u))}{\tilde{V}^{N+1}_i(u)
-\tilde{V}^N_i(u)}\nonumber\\
&\qquad\leq1-(1+R_{N+1})^{-\frac\theta2}e^{-C_T}.
\end{align*}
Hence, $\hat{W}^{o,\tau}=(\hat{W}^{o,\tau}(s))_{s\in[0,T]}$ and $\hat{\Upsilon}^*=(\hat{\Upsilon}^*(s))_{s\in[0,T]}$ defined by \eqref{eq:WhatM} are $(\Qx,\Fx^{\rm M})$-martingales. It follows from \eqref{eq:randomdirver} that $\tilde{f}^N(\omega,u,\xi,v)\geq \tilde{f}^{N+1}(\omega,u,\xi,v)$ for all $(\omega,u,\xi,v)$. By putting all the pieces together, \eqref{tuncatedBSDE} implies that, for $u\in[t,T]$,
\begin{align*}
\tilde{Y}^{N+1}(u)-\tilde{Y}^N(u)
&\geq-\int_u^T(\tilde{Z}^{N+1}(s)-\tilde{Z}^N(s))^\top d\hat{W}^{o,\tau}(s)\nonumber\\
&\quad-\int_u^T(\tilde{V}^{N+1}(s)-\tilde{V}^N(s))^\top d\hat{\Upsilon}^*(s).
\end{align*}
This confirms the desired comparison result that $\tilde{Y}^{N+1}(u)\geq\tilde{Y}^N(u)$, $\Px^*$-a.e., as we have $\Qx\sim\Px^*$.
\end{proof}

\end{appendix}

\begin{acks}[Acknowledgments]
We thank two anonymous referees for the careful reading and helpful comments.
\end{acks}

\begin{funding}
The first author was supported in part by Natural Science Foundation of China under grant no. 11971368 and 11961141009.
The second author was supported in part by Singapore MOE AcRF Grants R-146-000-271-112.
The third author was supported in part by the Hong Kong Early Career Scheme under grant no. 25302116.
\end{funding}




\end{document}